\documentclass[10ps]{article}
\usepackage{isabelle,isabellesym}
\usepackage{graphicx}
\usepackage{amsmath} 
\usepackage{amsfonts}
\usepackage{amssymb}
\usepackage{amscd}
\usepackage{amsthm}
\usepackage{latexsym}
\usepackage{mathrsfs}
\usepackage{stmaryrd}
\usepackage{mathpartir} 
\usepackage{hyperref}
\usepackage{color}
\usepackage{tikz}
\usepackage{tikz-cd}
\usepackage{fullpage}
\usepackage{hyperref}

\usepackage[colorinlistoftodos]{todonotes}

\usetikzlibrary{shapes,arrows,positioning,fit,backgrounds}

\usepackage[normalem]{ulem}

\newcommand{\IF}[3]{\mathbf{if}\ #1\ \mathbf{then}\ #2\ \mathbf{else}\ #3}
\newcommand{\WHILE}[2]{\mathbf{while}\ #1\ \mathbf{do}\ #2}
\newcommand{\WHILEI}[3]{\mathbf{while}\ #1\ \mathbf{inv}\ #2\ \mathbf{do}\ #3}
\newcommand{\MKA}{\mathsf{MKA}}

\newcommand{\PDL}{\mathsf{PDL}}
\newcommand{\dL}{\mathsf{d}\mathcal{L}}
\newcommand{\flow}{\varphi}
\newcommand{\orbit}{\gamma^\varphi}
\newcommand{\lipschitz}{\ell}
\newcommand{\Pow}{\mathcal{P}}
\newcommand{\id}{\mathit{id}}
\newcommand{\Id}{\mathit{Id}}
\newcommand{\reals}{\mathbb{R}}
\newcommand{\bools}{\mathbb{B}}
\newcommand{\true}{\top}

\newcommand{\ad}{\mathit{ad}}
\newcommand{\ar}{\mathit{ar}}
\newcommand{\wlp}{\mathit{wlp}}
\newcommand{\Set}{\mathbf{Set}}
\newcommand{\Rel}{\mathbf{Rel}}
\newcommand{\Sols}{\mathop{\mathsf{Sols}}}
\newcommand{\sta}{\mathsf{Sta}}
\newcommand{\rel}{\mathsf{Rel}}
\newcommand{\inv}{\mathsf{Inv}}

\newcommand{\gorbit}[1]{\gamma^#1_G}

\renewcommand{\d}{\operatorname{d\!}{}}

\newtheorem{theorem}{Theorem}[section]
\newtheorem{proposition}[theorem]{Proposition}
\newtheorem{lemma}[theorem]{Lemma}
\newtheorem{corollary}[theorem]{Corollary}

\theoremstyle{definition}

\newtheorem{example}{Example}[section]

\definecolor{jcolor}{cmyk}{1,0,1,0}
\definecolor{gcolor}{cmyk}{1,0,0,0}

\isabellestyle{it}

\begin{document}

\title{Predicate Transformer Semantics for Hybrid Systems\\\large{---Verification Components for Isabelle/HOL---}}

\author{Jonathan Juli\'an Huerta y Munive\\
University of Sheffield\\
United Kingdom
\and Georg Struth\\
University of Sheffield\\
United Kingdom
}



\maketitle

\begin{abstract}
  We present a semantic framework for the deductive verification of
  hybrid systems with Isabelle/HOL. It supports reasoning about the
  temporal evolutions of hybrid programs in the style of differential
  dynamic logic modelled by flows or invariant sets for vector
  fields. We introduce the semantic foundations of this framework and
  summarise their Isabelle formalisation as well as the resulting
  verification components. A series of simple examples shows our
  approach at work.

\vspace{\baselineskip}

\end{abstract}


\section{Introduction}\label{sec:intro}

\begin{figure}[t]
 \begin{center}
{\small
\tikzstyle{block} = [rectangle, draw, fill=blue!20, 
    text width=9em, text centered, rounded corners, minimum height=3em]
\tikzstyle{line} = [draw, -latex']
\tikzstyle{container} = [draw, rectangle, rounded corners, inner sep=.4cm, fill=gray!20,minimum height=2.2cm]

\begin{tikzpicture}[node distance = 1.7cm and -1.1cm, auto]

\node [block] (mka) {modal Kleene algebras};
\node[block, below right = of mka] (sta) {state transformers};
\node [block, above right = of sta] (back) {predicate transformer quantales};
\node [block, below right = of back] (rel) {binary relations};
\node [block, above right = of rel] (mon) {predicate transformer quantaloids};
\node [block, below left = of sta] (dyn) {dynamical systems};
\node [block, below right = of sta] (pl) {Lipschitz continuous vector fields};
\node [block, below right = of rel] (gen) {continuous vector fields};
\begin{scope}[on background layer]
\node [container,fit = (dyn) (pl) (gen)] (container) {};
\end{scope}
\node [below of = pl, node distance = .8cm] (hyb) {hybrid store dynamics};

\path [line] (mka) -- (sta);
\path [line] (mka) -- (rel);
\path [line] (back) -- (sta);
\path [line] (back) -- (rel);
\path [line] (mon) -- (sta);
\path [line] (mon) -- (rel);

\path [line] (sta) -- (dyn);
\path [line] (sta) -- (pl);
\path [line] (sta) -- (gen);
\path [line] (rel) -- (dyn);
\path [line] (rel) -- (pl);
\path [line] (rel) -- (gen);
\end{tikzpicture}
}
  \caption{Isabelle framework for hybrid systems verification}
  \label{fig:framework}
\end{center}
\end{figure}
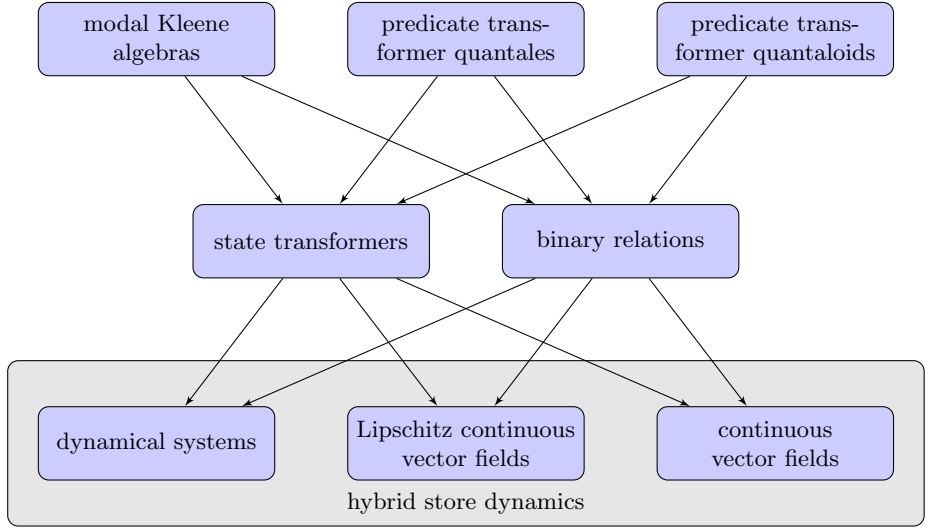


Hybrid systems combine continuous dynamics with discrete control.
Their verification is receiving increasing attention as the number of
computing systems controlling real-world physical systems is growing.
Mathematically, hybrid system verification requires integrating
continuous system dynamics, often modelled by systems of differential
equations, and discrete control components into hybrid automata,
hybrid programs or similar domain-specific modelling formalisms, and
into analysis techniques for these.  Such techniques include state
space exploration, reachability or safety analyses by model checking
and deductive verification with domain-specific
logics~\cite{DoyenFPP18}.

One of the most prominent deductive approaches is differential dynamic
logic $\dL$~\cite{Platzer10}, an extension of dynamic
logic~\cite{HarelKT00} to hybrid programs for reasoning with
autonomous systems of differential equations, their solutions and
invariant sets.  It is supported by the KeYmaera X
tool~\cite{FultonMQVP15} and has proved its worth in numerous case
studies~\cite{JeanninGKSGMP17,LoosPN11,Platzer10}.  KeYmaera X
verifies partial correctness specifications for hybrid programs using
a combination of domain-specific sequent and Hilbert calculi, which
itself is based on an intricate uniform substitution calculus. For
pragmatic reasons, its language is restricted to differential terms of
real arithmetic~\cite{FultonMQVP15} (that of hybrid automata is
usually restricted to polynomial or linear
constraints~\cite{DoyenFPP18}).

Our initial motivation for this work has been the formalisation of a
$\dL$-style approach to hybrid program verification in the
Isabelle/HOL proof assistant~\cite{MuniveS18} by combining Isabelle's
mathematical components for analysis and ordinary differential
equations~\cite{HolzlIH13,Immler12,ImmlerH12a,ImmlerT19} with
verification components for modal Kleene algebras~\cite{GomesS16}. We
are using a shallow embedding that, in general, encodes semantic
representations of domain-specific formalisms within a host-language
(deep embeddings start from syntactic representations using data types
to program abstract syntax trees).  This benefits not only from the
well known advantages of shallowness: more rapid developments and
simpler, more adaptable components. It has also shifted our focus from
encoding $\dL$'s complex syntactic proof system to developing
denotational semantics for hybrid systems and supporting the natural
style in which mathematicians, physicists or engineers reason about
them---without proof-theoretic baggage. After all, we get Isabelle's
own proof system and proof methods for free, and our expressive power
is only limited by its type theory and higher-order logic.

Our main contribution is an open compositional semantic framework for
the deductive verification of hybrid programs in a general purpose
proof assistant. In a nutshell, hybrid programs are while programs, or
simply programs with control loops, in which an evolution command for
the continuous system dynamics complements the standard assignment
command for the discrete control. Evolution commands roughly specify
vector fields (via systems of ordinary differential equations)
together with guards that model boundary conditions. Here we restrict
our attention to abstract predicate transformer algebras using modal
Kleene algebras~\cite{DesharnaisS11}, quantales of lattice
endofunctions or quantaloids of functions between
lattices~\cite{BackW98}. They are instantiated first to intermediate
relational or state transformer semantics for $\dL$-style hybrid
programs, and then to concrete semantics over program stores for
hybrid programs: for dynamical systems with global flows, Lipschitz
continuous vector fields with local flows and continuous vector fields
with multiple solutions.  Another verification component is based
directly on flows. This array of components demonstrates the
compositionality and versatility of our
framework. Figure~\ref{fig:framework} shows its basic anatomy.

Our framework benefits from compositionality and algebra in various
ways.  Using algebra allows us to derive most of the semantic
propreties needed for verification by equational reasoning, and it
reduces the overhead of developing different concrete semantics to a
minimum. Using modal Kleene algebras and predicate transformer
algebras, in particular, makes large parts of verification condition
generation equational, and thus accessible to Isabelle's simplifiers.
Compositionality of our extant framework for classical programs allows
us to localise the development of concrete semantics for hybrid
programs to the specification and formalisation of a semantics for
evolutions commands. We only need to replace standard models of the
program store by a hybrid store model. In our denotational state
transformer semantics, evolution commands are interpreted as unions of
all orbits of solutions of the vector field at some initial value,
subject to the guards constraining the durations of evolutions.  This
covers situations beyond the remits of the Picard-Lindel{\"o}f
theorem~\cite{Hirsch09,Teschl12} and supports general reasoning about
guarded invariant sets. Ultimately, we can simply plug the predicate
transformers for evolution commands into the generic algebras for
while programs and their rules for verification condition generation.

Verification condition generation for evolution commands is supported
by three workflows that are inspired by $\dL$, but work differently in
practice:

\begin{itemize}
\item The first one asks users to supply a flow and a Lipschitz
  constant for the vector field specified by the evolution command. We
  usually obtain this data using an external computer algebra system
  (integrating one into Isabelle seems routine and is left for future
  work).  After certifying the flow conditions and checking Lipschitz
  continuity of the vector field, as dictated by the Picard-Lindel\"of
  theorem, the orbit for the flow can be used to compute the weakest
  liberal preconditions for the evolution command. This workflow
  deviates from $\dL$ in allowing users to supply an interval of
  interest as domain of the flow.

\item The second workflow works more generally in situations where
  unique solutions need not exist or are difficult to work with. It
  requires users to supply an invariant set for the vector field in
  the sense of dynamical systems
  theory~\cite{Hirsch09,Teschl12}. After certifying the properties for
  invariant sets, a correctness specification for the evolution
  command and the invariant set is used in place of a weakest liberal
  precondition. Here, beyond $\dL$, we support working with solutions
  defined over chosen intervals and using $\dL$-style inference rules
  as well as arbitrary higher-order logic.

\item The third workflow uses flows ab initio in the specification and
  semantic analysis of evolution commands. This circumvents checking
  any continuity, existence, uniqueness or invariant conditions of
  vector fields mentioned. This is not at all supported by $\dL$.
\end{itemize}

With all three workflows, hybrid program verification is ultimately
performed within the concrete semantics. But, as with classical
program verification, verification condition generation eliminates all
structural conditions automatically so that proof obligations are
entirely about the dynamics of the hybrid program store.  They can be
calculated in mathematical textbook style by equational reasoning, and
of course by external solvers and decision procedures for arithmetic
(their integration, as oracles or as verified components, is very
important, but left for future work). For the introductory examples
presented, we have merely formalised some simple tactics that help
automating the computation of derivatives in multivariate Banach
spaces or that of polynomials and transcendental functions. Yet for
those who prefer $\dL$-style reasoning we have formalised a
rudimentary set of its inference rules that are sound relative to our
semantics. Overall, unlike $\dL$, which prescribes its domain-specific
set of inference rules, we grant users the freedom of choice between
various workflows and even of developing their own one within our
semantic framework.

The entire framework, including the mathematical development in this
article, has been formalised with Isabelle/HOL. All Isabelle
components can be found in the Archive of Formal
Proofs~\cite{afp:kad,afp:vericomp,afp:transem,afp:hybrid}. We are
currently using them to verify hybrid programs post hoc in the
standard weakest liberal precondition style outlined above.  Yet the
approach is flexible enough to support the verification of hybrid
systems using Hoare logic~\cite{FosterMS20}, symbolic execution with
strongest postconditions, program refinement with predicate
transformers in the style of Back and von Wright~\cite{BackW98} and
Morgan~\cite{FosterMS20}, and reasoning about hybrid program
equivalences in the elegant equational style of Kleene algebra with
tests~\cite{Kozen97}.

While our approach is powerful enough to tackle most problems of a
recent systems competition~\cite{MitschMJZWZ20}, the work documented
in this article focuses mainly on the semantic foundations and the
proof of concept that the approach works. A more user-friendly
specification language, a less simplistic hybrid store model, enhanced
tactics for reasoning with flows and invariants, and mathematical
background theories for reasoning about affine and linear systems of
differential equations have been added while this article has been
under review~\cite{Munive20,FosterGMS21}.  The doctoral dissertation
of the first author contains a more comprehensive description of the
framework and further generalisations~\cite{Munive21}.

The remainder of this article is organised as follows:
Section~\ref{sec:KA}-\ref{sec:pt-monad} introduce the algebras of
relations, state and predicate transformers needed.
Section~\ref{sec:discrete-store} explains the shallow embedding used
to formalise verification components for while programs. After
recalling the basics of differential equations in
Section~\ref{sec:ODE}, we introduce our semantics for evolution
commands in
Section~\ref{sec:hybrid-store}-\ref{sec:differential-invariants} and
explain our procedures for computing weakest liberal preconditions and
reasoning with differential invariants for
them. Section~\ref{sec:isa-pt}-\ref{sec:isa-wlp} summarise the
corresponding Isabelle components. Section~\ref{sec:dL} and
\ref{sec:isa-dL} briefly list the derivation and formalisation of
semantic variants of $\dL$ inference rules.
Section~\ref{sec:examples} presents four verification examples in our
framework using the main two workflows. Section~\ref{sec:flow-component}
presents our third workflow and a brief example for it.
Sections~\ref{sec:related-work} and~\ref{sec:conclusion} discuss related
work and conclude the article. A glossary of cross-references between theorems
in the text and the Isabelle theories is presented in
Appendix~\ref{sec:crossref}.


\section{Kleene Algebra}\label{sec:KA}

This section summarises the mathematical foundations of our simplest
and most developed predicate transformer algebra---modal Kleene
algebra. It introduces the basics of Kleene algebras, and the state
transformer model and relational model used.  The relational model is
standard for Kleene algebra. The state transformer model has so far
received less attention and is therefore explained in detail.

A \emph{dioid} $(S,+,\cdot,0,1)$ is an additatively idempotent
semiring, $\alpha+\alpha=\alpha$ holds for all $\alpha\in S$. The
underlying abelian monoid $(S,+,0)$ is therefore a semilattice with
order defined by $\alpha\le \beta\leftrightarrow \alpha+\beta=\beta$.
The order is preserved by $\cdot$ and $+$ in both arguments; $0$ is
its least element.

A \emph{Kleene algebra} $(K,+,\cdot,0,1,^\ast)$ is a dioid expanded by
the Kleene star $(-)^\ast:K\to K$ that satisfies the left and right
unfold and induction axioms
\begin{align*}
  1+\alpha\cdot\alpha^\ast &\le \alpha^\ast, \qquad \gamma+\alpha\cdot
  \beta\le \beta\rightarrow \alpha^\ast \cdot \gamma\le \beta,\\
  1+\alpha^\ast\cdot\alpha &\le \alpha^\ast, \qquad \gamma+\beta\cdot
  \alpha\le \beta\rightarrow \gamma\cdot \alpha^\ast \le \beta.
\end{align*}
By these axioms, $\alpha^\ast\cdot \gamma$ is the least fixpoint of
the function $\gamma+\alpha\cdot (-)$ and $\gamma\cdot \alpha^\ast$
that of $\gamma+(-)\cdot \alpha$, where we use $-$ as a wildcard for
function arguments.  The fixpoint $\alpha^\ast$ arises as a special
case. The more general induction axioms combine its definition with
sup-preservation or continuity of left and right multiplication.

Opposition is an important duality of Kleene algebras: swapping the
order of multiplication in any Kleene algebra yields another one. The
class of Kleene algebras is therefore closed under opposition.

Kleene algebras were conceived as algebras of regular expressions. But
here we interpret their elements as programs.  Addition models their
nondeterministic choice, multiplication their sequential composition
and the Kleene star their unbounded finite iteration.  The element $0$
models abort; $1$ models the ineffective program.  Two programs are
deemed equal if they lead from the same inputs to the same
outputs. These intuitions are grounded in concrete program semantics.

With the relational composition of $R\subseteq X\times Y$ and
$S\subseteq Y\times Z$ defined as $(R;S)\, x \,z$ if $R\, x\, y$ and
$S\, y\, z$ for some $y\in Y$, with $\Id_X\, x\, y$ if $x=y$, and the
reflexive-transitive closure of $R\subseteq X\times X$ defined as
$R^\ast=\bigcup_{i\in\mathbb{N}} R^i$, where $R^0=\Id_X$ and
$R^{i+1}=R;R^i$, where we write $R\, x\, y$ instead of $(x,y)\in R$,
the following holds.

\begin{proposition}\label{P:rel-ka}
  Let $X$ be a set. Then
  $\rel\, X = (\Pow\, (X\times X),\cup,;,\emptyset,\Id_X,^\ast)$ forms
  a Kleene algebra---the \emph{full relation Kleene algebra} over $X$.
\end{proposition}
A \emph{relation Kleene
  algebra} over $X$ is thus any subalgebra of $\rel\, X$.

Opposition can be expressed in $\rel\, X$ by conversion, where the
\emph{converse} of relation $R$ is defined by
$R^\smallsmile\, x\, y \leftrightarrow R\, y\, x$. It satisfies in particular
$(R;S)^\smallsmile = S^\smallsmile ; R^\smallsmile$.

The isomorphism $\Pow\, (X\times Y) \cong (\Pow\, Y)^X$ between
categories of relations and non-deterministic functions---so-called
\emph{state transformers}---yields an alternative representation.  It
is given in terms of the bijections
$\mathcal{F}:\Pow\, (X\times Y) \to (\Pow\, Y)^X$ and
$\mathcal{R}:(\Pow\, Y)^X\to \Pow (X\times Y)$ defined by
$\mathcal{F}\, R\, x = \{y\in Y\mid R\, x\, y\}$ and by
$\mathcal{R}\, f\, x\, y \Leftrightarrow y \in f\, x$.  Following
Isabelle syntax, we use juxtaposition with a space to denote function
application. State transformers $f:X\to \Pow\, Y$ and
$g:Y\to \Pow\, Z$ are composed by the (forward) Kleisli composition of
the powerset monad
\begin{equation*}
(f\circ_K g)\, x = \bigcup\{g\, y\mid y \in f\ x \}.
\end{equation*}
The function $\eta_X = \{-\}$ is a unit of this monad. The functors
$\mathcal{F}$ and $\mathcal{R}$ preserve arbitrary sups and infs,
extended pointwise to state tranformers, and stars
$f^{\ast_K}\, x = \bigcup_{i\in\mathbb{N}} f^{i_K}\, x$, which are
defined with respect to Kleisli composition.

\begin{proposition}\label{P:kleisli-ka}
  Let $X$ be a set. Then
  $\sta\, X = ((\Pow\, X)^X,\cup,\circ_K,\lambda x.\ \emptyset,
  \eta_X,^{\ast_K})$ forms a Kleene algebra---the \emph{full state
    transformer Kleene algebra} over $X$.
\end{proposition}
A \emph{state transformer Kleene algebra} over $X$ is any subalgebra
of $\sta\, X$.  Opposition is now expressed using the (contravariant)
functor
$(-)^{\mathit{op}} = \mathcal{F}\circ (-)^\smallsmile\circ
\mathcal{R}$ that associates $f^{\mathit{op}}:Y\to \Pow\, X$ with
every $f:X\to \Pow\, Y$.

The category $\mathbf{Rel}$, with relations of type $X\times Y$ or
state transformers of type $X\to \Pow\, Y$ as arrows, is beyond
mono-type Kleene algebra.

For a more refined hierarchy of variants of Kleene algebras, their
calculational properties and the most important computational models,
see our formalisation in the Archive of Formal
Proofs~\cite{afp:ka}. The state transformer model has been formalised
with Isabelle for this article.


\section{Modal Kleene Algebra}\label{sec:MKA}

Kleene algebras must be extended to express conditionals or
while-loops more faithfully. This requires tests, which are not prima
facie actions, but propositions. Assertions and correctness
specifications cannot be expressed directly either.

Two standard extensions bring Kleene algebra closer to program
semantics. Kleene algebra with tests~\cite{Kozen97} yields a simple
algebraic semantics for while-programs and a partial correctness
semantics for these in terms of an algebraic propositional Hoare
logic---ignoring assignments.  Predicate transformer
semantics, however, cannot be expressed~\cite{Struth16}.
Alternatively, Kleene algebras can be enriched by modal box and
diamond operators in the style of propositional dynamic logic
($\PDL$), which yields test and assertions as well as predicate
transformers. Yet once again, assignments cannot be expressed within
the algebra. We outline the second approach.

An \emph{antidomain semiring} \cite{DesharnaisS11} is a semiring $S$
expanded by an \emph{antidomain operation} $\ad:S\to S$ axiomatised by
\begin{align*}
 \ad\, \alpha \cdot \alpha = 0,\qquad
\ad\, \alpha + \ad^2\, \alpha = 1,\qquad \ad\, (\alpha\cdot \beta) \le \ad\, (\alpha \cdot \ad^2\, \beta).
\end{align*}
By opposition, an \emph{antirange semiring} \cite{DesharnaisS11} is a
semiring $S$ expanded by an \emph{antirange operation} $\ar:S\to S$
axiomatised by
\begin{equation*}
 \alpha\cdot \ar\, \alpha  = 0,\qquad
\ar\, \alpha + \ar^2\, \alpha = 1,\qquad \ar\, (\alpha\cdot \beta) \le
\ar\, (\ar^2\, \alpha \cdot \beta).
\end{equation*}
Antidomain and antirange semirings 
are a fortiori dioids.

The antidomain $\ad\, \alpha$ of program $\alpha$ models the set of
those states from which $\alpha$ cannot be executed. The operation
$d=\ad^2$ thus defines the \emph{domain} of a program: the set of
those states from which it can be executed.  By opposition, the
antirange $\ar\, \alpha$ of $\alpha$ yields those states into which
$\alpha$ cannot be executed and $r=\ar^2$ defines the \emph{range} of
$\alpha$: those states into which it can be executed.

A \emph{modal Kleene algebra} ($\MKA$) \cite{DesharnaisS11} is a
Kleene algebra that is both an antidomain and an antirange Kleene
algebra in which $d\circ r=r$ and $r\circ d = d$.

In a $\MKA$ $K$, the set $\Pow\, \ad\, K$---the image of $K$ under
$\ad$---models the set of all tests or propositions. We henceforth
often write $p,q,\dots$ for its elements. Moreover,
$\Pow\, \ad\, K=\Pow\, d\, K=\Pow\, r\, K =\Pow\, ar\, K = K_d=K_r$,
where $K_f=\{\alpha\in S\mid f\, \alpha = \alpha\}$ for $f\in\{d,r\}$.
Hence $p\in \Pow\, \ad\, K \leftrightarrow d\, p = p$. It follows that
the class $\MKA$ is closed under opposition. In addition, $K_d$ forms a
boolean algebra with least element $0$, greatest element $1$, join
$+$, meet $\cdot$ and complementation $\ad$---the algebra of
propositions, assertions or tests.

Axiomatising $\MKA$ based on domain and range would lack the power to
express complementation: $K_d$ would only be a distributive lattice.

The programming intuitions for $\MKA$ are once again grounded in
concrete semantics.
\begin{proposition}\label{P:rel-mka}
  If $X$ is a set, then $\rel\, X$ is the full relation $\MKA$
  over $X$ with
  \begin{align*}
    \ad\, R\, x\, x \leftrightarrow \neg \exists y \in X.\
    R\, x\, y \qquad\text{ and }\qquad
    \ar\, R = \ad\, R^\smallsmile.
  \end{align*}
\end{proposition}
Every subalgebra of a full relation $\MKA$ is a \emph{relation}
$\MKA$.

Similarly, $\ar = \ad \circ (-)^\smallsmile$,
$d = r \circ (-)^\smallsmile$ and $r = d\circ (-)^\smallsmile$.
Furthermore, 
\begin{equation*}
(\Pow\, (X\times X))_d= \{P\mid P\subseteq \Id_X\}.
\end{equation*}
Henceforth we often identify such relational subidentities, sets and
predicates and their types via the isomorphisms
$(\Pow\, (X\times X))_d\, \cong\, X\to \bools\, \cong\, \Pow\, X$.

\begin{proposition}\label{P:kleisli-mka}
  Let $X$ be a set. Then $\sta\, X$ is the full state
  transformer $\MKA$ over $X$ with
  \begin{align*}
    \ad\, f\, x = 
    \begin{cases}
      \eta_X\, x, & \text{ if } f\, x = \emptyset,\\
\emptyset, & \text{ otherwise},
    \end{cases}
\qquad\text{ and }\qquad
\ar\, f = \ad\, f^{\mathit{op}}.
  \end{align*}
\end{proposition}
Every subalgebra of a full relation $\MKA$ is a \emph{state
  transformer} $\MKA$. Similarly, 
\begin{align*}
     d\, f\, x = 
    \begin{cases}
      \emptyset, & \text{ if } f\, x = \emptyset,\\
\eta_X\, x, & \text{ otherwise},
    \end{cases}
\qquad\text{ and }\qquad
r\, f = d\, f^{\mathit{op}}.
\end{align*}
These propositions generalise again beyond mono-types, but algebras of
such typed relations and state transformers cannot be captured by
$\MKA$.

In every $\MKA$, $p\cdot \alpha$ and $\alpha\cdot p$
model the domain and range restriction of $\alpha$ to states
satisfying $p$.  Conditionals and while loops can thus be
expressed:
\begin{align*}
  \IF{p}{\alpha}{\beta} = p\cdot \alpha + \bar p \cdot
  \beta\qquad\text{ and }\qquad
\WHILE{p}{\alpha} = (p\cdot \alpha)^\ast \cdot \bar p,
\end{align*}
where we write $\bar p = \ad\, p = \ar\, p$. Together with sequential
composition $\alpha ; \beta= \alpha\cdot \beta$ this yields an
algebraic semantics of while programs without assignments. It
is grounded in the relational and the state transformer semantics.
A more refined hierarchy of variants of $\MKA$s, starting from domain
and antidomain semigroups, their calculational properties and the most
important computational models, can be found in the Archive of Formal
Proofs~\cite{afp:kad}. The state transformer model of $\MKA$ has been
formalised with Isabelle for this article.


\section{Modal Kleene Algebra, Predicate Transformers and Invariants}\label{sec:mka-pt}

$\MKA$ can express the modal operators of $\PDL$, both with a
relational Kripke semantics and a coalgebraic state transformer semantics.
\begin{equation*}
|\alpha\rangle p = d\, (\alpha\cdot p), \qquad |\alpha ] p = \ad\,
(\alpha\cdot \ad\, p),\qquad
\langle \alpha | p  = r\, (p\cdot \alpha), \qquad [\alpha |p = \ar\, (\ar
\, p\cdot \alpha).
\end{equation*}
This is consistent with J\'onsson and Tarski's boolean algebras with
operators~\cite{JonssonT51}: Each of $|\alpha\rangle$,
$\langle\alpha |$, $|\alpha ]$ and $[\alpha |$ is an endofunction
$K_d\to K_d$ on the boolean algebra $K_d$.  Yet another view of modal
operators is that of predicate transformers.  The function $|-]-$
yields the \emph{weakest liberal precondition} operator $\wlp$;
$\langle -|-$ the \emph{strongest postcondition} operator.

The boxes and diamonds of $\MKA$ are related by De Morgan duality:
\begin{equation*}
  |\alpha\rangle p = \overline{ |\alpha ]\bar p}, \qquad |\alpha ] p  =
  \overline{ |\alpha \rangle \bar p},\qquad
  \langle\alpha| p = \overline{ [\alpha |\bar p}, \qquad [\alpha | p  =
  \overline{ \langle \alpha | \bar p}\, ;
\end{equation*}
 their dualities are captured by the adjunctions and conjugations
\begin{alignat*}{4}
  |\alpha\rangle p \le q &\leftrightarrow p \le [\alpha|q, &\qquad
  \langle \alpha |p \le q &\leftrightarrow p \le |\alpha]q,\\
  |\alpha\rangle p\cdot q = 0 &\leftrightarrow p\cdot \langle
  \alpha|q=0,&\qquad |\alpha]p+ q = 1&\leftrightarrow
  p+[\alpha|q=1. 
\end{alignat*}

In $\rel\, X$, as in standard Kripke semantics of modal logics in
general, and of $\PDL$ in particular,
\begin{equation*}
|R\rangle P = \{x\mid \exists y\in X.\ R\, x\, y \wedge 
P\, y\}\quad\text{ and }\quad
|R]P = \{x\mid \forall y\in X.\ R\, x\, y \rightarrow P\, y\},
\end{equation*}
where we identify predicates and subidentity relations. For the
remaining two modalities, $\langle -|=|-\rangle\circ (-)^\smallsmile$
and $ [-|=|-]\circ (-)^\smallsmile$.  Hence $|R\rangle P$ is the
preimage of $P$ under $R$ and $\langle R|P$ the image of $P$ under
$R$.  The isomorphism between subidenties, predicates and sets also
allows us to see $|R\rangle$, $\langle R|$, $|R]$ and $[R|$ as
operators on the complete atomic boolean algebra $\Pow\, X$, which
carries algebraic structure beyond $K_d$ that is reminiscent of a
module.

In $\sta\, X$, alternatively,
\begin{equation*}
\langle f | P =\{y\mid \exists x.\ y \in
                f\, x \wedge P\, x\}\quad \text{ and }\quad
|f]P = \{x\mid f\, x \subseteq P\}. 
\end{equation*} 
Moreover, $|-\rangle = \langle-|\circ (-)^{\mathit{op}}$ and
$[-| = |-]\circ (-)^{\mathit{op}}$.  Here, $\langle f|$ is the Kleisli
extension of $f$ for the powerset monad and $|f\rangle$ that of the
opposite function (see Section~\ref{sec:pt-monad}).

The isomorphism $\Pow\, (X\times X) \cong (\Pow\, X)^X$  makes the
approaches coherent: 
\begin{align*}
  |f\rangle = |\mathcal{R}\,  f\rangle,\qquad |R\rangle =
  |\mathcal{F}\, 
  R\rangle, \qquad
  | f] = |\mathcal{R}\,  f],\qquad |R] = |\mathcal{F}\,  R],
\end{align*}
and, dually,   $\langle f| = \langle \mathcal{R}\,  f|$,  $\langle R| =
  \langle \mathcal{F}\, 
  R|$,  $[f| = [\mathcal{R}\,  f|$ and $[R| = [\mathcal{F}\,  R|$.

Predicate transformers are useful for specifying program correctness
conditions and for verification condition generation.  The identity
\begin{equation*}
p\le |\alpha]q
\end{equation*}
captures the standard partial correctness specification for programs:
if $\alpha$ is executed from states where precondition $p$ holds, and
if it terminates, then postcondition $q$ holds in the states where it
does.  Verifying it amounts to computing $|\alpha]q$ recursively over
the program structure from $q$ and checking that the result is greater
or equal to
$p$.  Intuitively, $|\alpha]q$ represents the largest set of states
from which one must end up in set $q$ when executing $\alpha$, or
alternatively the weakest precondition from which postcondition $q$
must hold when executing $\alpha$.

Calculating $|\alpha]q$ for straight-line programs is completely
equational, but loops require invariants.  To this end one usually
adds annotations to loops,
\begin{equation*}
\WHILEI{p}{i}{\alpha} = \WHILE{p}{\alpha},
\end{equation*}
where $i$ is the \emph{loop invariant} for $\alpha$, and calculates
$\wlp$s as follows~\cite{GomesS16,afp:vericomp}.  For all
$p,q,i,t\in K_d$ and $\alpha,\beta\in K$,
  \begin{gather}
|\alpha \cdot \beta] q = |\alpha] |\beta] q,\label{eq:wlp-seq}\tag{wlp-seq}\\
|\IF{p}{\alpha}{\beta}] q = (\bar p + |\alpha] q)\cdot (p + |\beta] q) = p
\cdot |\alpha] q+ \bar p \cdot |\beta] q,\label{eq:wlp-cond}\tag{wlp-cond}\\
  i\le |\alpha]i \rightarrow i\le |\alpha^\ast]i,\label{eq:wlp-star}\tag{wlp-star}\\
p \le i \wedge i\cdot t\le |\alpha] i \wedge i\cdot \bar t\le q\, \rightarrow \, p \le
  |\WHILEI{t}{i}{\alpha}] q.\label{eq:wlp-loop}\tag{wlp-while}
  \end{gather}
  In the rule (\ref{eq:wlp-star}), $i$ is a an invariant for the star
  as well.  In addition we support a while rule without an invariant
  annotation.

  More generally, beyond loops, an element $i\in K_d$ is an
  \emph{invariant} for $\alpha$ if it is a postfixpoint of $|\alpha]$
  in $K_d$:
\begin{equation*}
  i \le |\alpha]i.
\end{equation*}
By the adjunction between boxes and diamonds, this is the case if and
only if $\langle\alpha | i \le i$, that is, $i$ is a prefixpoint of
$\langle \alpha|$ in $K_d$.  We return to this equivalence in the
context of differential invariants and invariant sets of vector fields
in Section~\ref{sec:differential-invariants}. We write
$\inv\, \alpha $ for the set of invariants of $\alpha$.

\begin{lemma}\label{P:inv-lemma}
In every $\MKA$, if $i,j\in \inv\, \alpha$, then  $i+j,i\cdot j\in \inv\, \alpha$. 
\end{lemma}

As a generalisation of the rule (\ref{eq:wlp-loop}) for annotated
while-loops we can derive a rule for commands annotated with
tentative invariants $\alpha\, \mathbf{inv}\ i = \alpha$. For all
$i,p,q\in K_d$ and $\alpha\in K$,
\begin{equation}
  p\le i \land i \le |\alpha] i \land i\le q\rightarrow p\le
  |\alpha\, \mathbf{inv}\, i] q. \label{eq:wlp-cmd}\tag{wlp-cmd}
\end{equation}
Combining (\ref{eq:wlp-cmd}) with (\ref{eq:wlp-star}) then yields, for
$\mathbf{loop}\, \alpha\, \mathbf{inv}\, i = \alpha^\ast$,

\begin{equation}
    p\le i \land i\le |\alpha]i \land i\le q\rightarrow p\le |\mathbf{loop}\, \alpha\,
    \mathbf{inv}\, i]q. \label{wlp-loop-inv}\tag{wlp-loop-inv}
\end{equation}
We use such annotated commands for reasoning about differential
invariants and loops of hybrid programs below.

The modal operators of $\MKA$ have, of course, a much richer algebra
beyond verification condition generation. For a comprehensive list see
the Archive of Formal Proofs~\cite{afp:kad}.  We have already derived
the rules of propositional Hoare logic, which ignores assignments, and
those for verification condition generation for symbolic execution
with strongest postconditions in this setting~\cite{afp:vericomp}. A
component for total correctness is also available. It supports
refinement proofs in the style of Back and von
Wright~\cite{BackW98}. But this is beyond the scope of this
article. The other two abstract predicate transformer algebras from
Figure~\ref{fig:framework} are surveyed in the following two sections.


\section{Predicate Transformers \`a la Back and von
  Wright}\label{sec:pt-backvwright}

While $\MKA$ has so far been our most developed setting for verifying
(hybrid) programs, our framework is compositional and supports other
predicate transformer algebras as well. Two of them are outlined in
this and the following section. Their Isabelle
formalisation~\cite{afp:transem} is discussed in
Section~\ref{sec:pt-monad}.

The first approach follows Back and von Wright~\cite{BackW98} in
modelling predicate transformers, or simply \emph{transformers}, as
functions between complete lattices.  Readers not familiar with
lattice theory can freely skip this section. To obtain useful laws for
program construction or verification, conditions are imposed.

A function $f:L_1\to L_2$ between two complete lattices $(L_1,\le_1)$
and $(L_2,\le_2)$ is \emph{order-preserving} if
$x\le_1 y \rightarrow f\, x \le_2 f\, y$, \emph{sup-preserving} if
$f \circ \bigsqcup = \bigsqcup \mathop{\circ} \Pow\, f$ and
\emph{inf-preserving} if
$f \circ \bigsqcap = \bigsqcap \mathop{\circ} \Pow\, f$. All sup- or
inf-preserving functions are order-preserving.

We write $\mathcal{T}(L)$ for the set of transformers over the
complete lattice $L$, and $\mathcal{T}_\le(L)$,
$\mathcal{T}_{\sqcup}(L)$, $\mathcal{T}_{\sqcap}(L)$ for the subsets
of order-, sup- and inf-preserving transformers. Obviously,
$\mathcal{T}_{\sqcap}(L)=\mathcal{T}_{\sqcup}(L^\mathit{op})$. The
following fact is well known~\cite{BackW98,GierzHKLMS80}.
\begin{proposition}\label{P:pt-lattice}
  Let $X$ be a set, let $L$ be a complete lattice. Then $L^X$ forms a
  complete lattice with order and sups extended pointwise.
\end{proposition}
Infs, least and greatest elements can then be defined from sups on
$L^X$ as usual. Function spaces $L^L$, in particular, form monoids
with respect to function composition $\circ$ and $\id_L$.  In
addition, $\circ$ preserves sups and infs in its first argument, but
not necessarily in its second one. Algebraically, this is captured as
follows.

A \emph{near-quantale} $(Q,\le,\cdot)$ is a complete lattice $(Q,\le)$
with an associative composition $\cdot$ that preserves sups in its
first argument. It is \emph{unital} if composition has a unit $1$. A
\emph{prequantale} is a near-quantale in which composition is
order-preserving in its second argument. A \emph{quantale} is a near
quantale in which composition preserves sups in its second argument.
See~\cite{Rosenthal90} for more information about quantales.

\begin{proposition}\label{P:trafo-quantale}
Let $L$ be a complete lattice. Then
  \begin{enumerate}
  \item $\mathcal{T}(L)$ and
    $\mathcal{T}(L^\mathit{op})$ form unital near-quantales;
  \item $\mathcal{T}_\le(L)$ ($\mathcal{T}_\le(L^\mathit{op}))$ forms
    a unital sub-prequantale of $\mathcal{T}(L)$ ($\mathcal{T}(L^\mathit{op})$);
  \item $\mathcal{T}_\sqcup(L)$ ($\mathcal{T}_\sqcap(L)$) forms a
    unital sub-quantale of $\mathcal{T}_\le(L)$ ($\mathcal{T}_\le(L^\mathit{op})$).
   \end{enumerate}
\end{proposition}

Transformers for while-loops are obtained by connecting quantales with
Kleene algebras.  This requires fixpoints of
$\varphi_{\alpha\gamma}= \gamma\sqcup \alpha\cdot (-)$ and
$\varphi_\alpha=1\sqcup \alpha\cdot (-)$ as well as the Kleene star
$\alpha^\ast=\bigsqcup_{i\in \mathbb{N}}\alpha^i$.  A \emph{left
  Kleene algebra} is a dioid in which $\varphi$ has a least fixpoint
that satisfies
$\mathit{lfp}\, \varphi_{\alpha\gamma} = \mathit{lfp}\, \varphi_\alpha
\cdot \gamma$.  Hence $\varphi_\alpha$ satisfies the left unfold and
left induction axioms
$1\sqcup \alpha\cdot \varphi_\alpha\le \varphi_\alpha$ and
$\gamma\sqcup \alpha\cdot \beta\le \beta\rightarrow
\varphi_\alpha\cdot \gamma\le \beta$.  By opposition, a \emph{right
  Kleene algebra} is a dioid in which the least fixpoint of a dual
function $1\sqcup (-)\cdot \alpha$ satisfies the right unfold and
right induction axioms.

\begin{proposition}\label{P:quantale-ka}~
\begin{enumerate}
\item   Every near-quantale is a right Kleene algebra with
  $\mathit{lfp}\, \varphi_\alpha = \alpha^\ast$. 
\item Every prequantale is also a left Kleene algebra. 
\item Every quantale is a Kleene algebra with
  $\mathit{lfp}\, \varphi_\alpha = \alpha^\ast$.
\end{enumerate}
\end{proposition}
The proofs of (1) and (3) use sup-preservation and Kleene's fixpoint
theorem. That of (2) uses the Knaster-Tarski fixpoint theorem to show
that $\varphi_{\alpha\gamma}$ has a least fixpoint, and fixpoint
fusion~\cite{MeijerFP91} to derive
$\mathit{lfp}\, \varphi_{\alpha\gamma} = \mathit{lfp}\, \varphi_\alpha
\cdot \gamma$,
which yields the left Kleene algebra axioms. In prequantales,
$\mathit{lfp}\, \varphi_\alpha \cdot \gamma\le \alpha^\ast \cdot
\gamma$;
equality generally requires sup-preservation in the first argument of
composition.

The fixpoint and iteration laws on functions spaces, which follow from
Proposition~\ref{P:quantale-ka} and~\ref{P:trafo-quantale}, still need
to be translated into laws for transformers operating on the
underlying lattice. This is achieved again by fixpoint
fusion~\cite{BackW98}.  In $\mathcal{T}_\le(L)$,
\begin{equation*}
  \mathit{lfp}\, (\lambda g.\ \id_L\sqcup f \circ g)\, x =
  \mathit{lfp}\, (\lambda y.\ x
\sqcup f\, y),
\end{equation*}
and $\mathit{lfp}$ preserves isotonicity.  In $\mathcal{T}_\sqcup(L)$,
moreover,
\begin{equation*}
f\, x \le x\rightarrow f^\ast \, x \le x,
\end{equation*}
$\id_L\sqcup f\circ f^\ast = f^\ast=f^\ast \circ f \sqcup \id_L$ and
$(-)^\ast$ preserves sups.  All results dualise to inf-preserving
transformers.

Relative to $\MKA$, backward diamonds correspond to sup-preserving
forward transformers and forward boxes to inf-preserving backward
transformers in the opposite quantale, where the lattice has been
dualised and the order of composition been swapped. An analogous
correspondence holds for forward diamonds and backward boxes. Sup- and
inf-preserving transformers over complete lattices are less
general than $\MKA$ in that preservation of arbitrary sups or infs is
required, whereas that of $\MKA$ is restricted to finite sups and
infs. Isotone transformers, however, are more general, as not even
finite sups or infs need to be preserved, and finite sup- or
inf-preservation implies order preservation.

We are mainly using the $\wlp$ operator for verification condition
generation and hence briefly outline $\wlp$s for conditionals and
loops in this setting. We assume that the underlying lattice $L$ is a
complete boolean algebra, that is, a complete lattice as well as a
complemented distributive lattice. We can then lift elements of $L$ to
$\wlp$s as $|p]q= p \to q$ and define, in
$\mathcal{T}_\le(L^\mathit{op})$,
\begin{align*}
  \mathbf{if}\, p\, \mathbf{then}\, f\, \mathbf{else}\, g = |p]\circ f
  \sqcap |\bar p]\circ g\qquad \text{ and } \qquad
\mathbf{while}\, p\, \mathbf{do}\, f =\mathit{lfp}\, \varphi_{|p]f}
                                       \circ |\bar p].
\end{align*}
In $\mathcal{T}_\sqcap(L)$, we even obtain
\begin{equation*}
  \mathbf{while}\, p\, \mathbf{do}\, f = (|p]\circ f)^\ast \circ |\bar p].
\end{equation*}
These equations allow generating verification conditions as with
(\ref{eq:wlp-cond}) and (\ref{eq:wlp-loop}) from
Section~\ref{sec:mka-pt}. Overall, our Isabelle components for
lattice-based predicate transformers in the Archive of Formal
Proofs~\cite{afp:transem} contain essentially the same equations and
rules for verification condition generation as those for $\MKA$.

We have so far restricted the approach to endofunctions on a complete
lattice to relate it to $\MKA$. Yet it generalises to functions in
$L_2^{L_1}$ and hence to categories~\cite{BackW98}. The corresponding
poly-typed generalisations of quantales are known as
\emph{quantaloids}~\cite{Rosenthal96}. In particular, composition is
then a partial operation.


\section{Predicate Transformers from the Powerset Monad}\label{sec:pt-monad}

A second, more coalgebraic approach to predicate transformers starts
from monads~\cite{Maclane71}.  In addition, it details the relational
and state transformer semantics of $\MKA$ in a more modern algebraic
approach. We need to assume basic knowledge of categories and
monads. Once again, readers unfamiliar with these concepts can freely
skip it.

Recall that $(\Pow,\eta_X,\mu_X)$, for $\Pow:\Set\to \Set$,
$\eta_X:X\to \Pow\, X$ defined by $\eta_X= \{-\}$ and
$\mu_X:\Pow^2\, X\to \Pow\, X$ defined by $\mu_X = \bigcup$ is the
monad of the powerset functor in the category $\Set$ of sets and
functions. The morphisms $\eta$ and $\mu$ are natural transformations.
They satisfy, for every $f:X\to Y$,
\begin{equation*}
\eta_Y\circ \id\, f = \Pow\, f \circ \eta_X\qquad\text{ and }\qquad
\mu_Y \circ \Pow^2\, f = \Pow\, f \circ \mu_X.
\end{equation*}

From the monadic point of view, state transformers $X\to \Pow\, Y$ are
arrows $X\to Y$ in the Kleisli category $\Set_\Pow$ of $\Pow$ over
$\Set$. They are composed by (forward) Kleisli composition
$f\circ_K g = \mu\circ \Pow\, g\circ f$ as explained before
Proposition~\ref{P:kleisli-ka} in Section~\ref{sec:KA}.  The category
$\Set_\Pow$ is known to be isomorphic to $\Rel$, the category of sets
and binary relations.

The isomorphism between state and forward predicate transformers is
based on the contravariant functor
$(-)^\dagger:\Set_\Pow(X,\Pow\, Y)\to \Set_\Pow(\Pow\, X,\Pow\,
Y)$---the Kleisli extension. Its definition
$f^\dagger = \mu \circ \Pow\, f$ implies that
$(-)^\dagger =\langle - |$ on morphisms, which is the strongest
postcondition operator.

The structure of state  spaces---boolean algebras for $\MKA$, complete
lattices  in  Back and  von  Wright's  approach---is captured  by  the
Eilenberg-Moore algebras of the powerset  monad. It is well known that
$(-)^\dagger$ embeds $\Set_\Pow$ into  their category. Its objects are
complete (sup-semi)lattices;  its morphisms  sup-preserving functions,
hence transformers. More precisely, $(-)^\dagger$ embeds into powerset
algebras, complete atomic  boolean algebras that are  the free objects
in this category.

The isomorphism
$\Set_\Pow(X, \Pow\, Y)\cong
\Set^\sqcup(\Pow\, X,\Pow\, Y)$
between state transformers and sup-preserving predicate transformers
then arises as follows. The embedding $\langle -|$ has an injective
inverse $\langle-|^{-1}$ on the subcategory of sup-preserving
transformers. It is defined by $\langle -|^{-1} = (-)\circ \eta$,
which can be spelled out as
$\langle \varphi|^{-1}\, x = \{y\mid y \in \varphi\, \{x\}\}$. The
isomorphism preserves the quantaloid structures of state and predicate
transformers that is, compositions (contravariantly), units and sups,
hence least elements, but not necessarily infs and greatest
elements. These results extend to 
$\Set^\sqcup(\Pow\, X,\Pow\, Y)
\cong\Rel(X, Y)$
via $\Set_\Pow\cong \Rel$. In addition, predicate
transformers $\langle f|: \Pow\, X\to \Pow\, Y$ preserve of course
sups in powerset lattices, hence least elements, but not necessarily
infs and greatest elements.

Forward boxes or $\wlp$s emerge from state transformers via a
functor
$|-]:\Set_\Pow(X,\Pow\, Y)\to
\Set(\Pow\, Y,\Pow\, X)$,
embedding Kleisli arrows into the opposite of the category of
Eilenberg-Moore algebras formed by complete (inf-semi)lattices and
inf-preserving functions. It is defined on morphisms as
$|-]=\partial_F\circ \langle- |\circ (-)^{\mathit{op}}$, where
$\partial_F\, f = \partial\circ f \circ \partial$ and $\partial$
dualises the lattice. Unfolding definitions, once again
$|f]\, P=\{x\mid f\, x\subseteq P\}$.

Furthermore, its inverse $|-]^{-1}$ on the subcategory of
inf-preserving transformers is
$|\varphi]^{-1}\, x = \bigcap\{P\mid x \in \varphi\, P\}$. The duality
$\Set_\Pow(X,\Pow\, Y)\cong \Set^\sqcap(\Pow\, Y,\Pow\, X)$
reverses Kleisli arrows and preserves the quantaloid structures up-to
lattice duality, mapping sups to infs and vice versa.  It extends to
relations as before.  In addition, predicate transformers $|f]$
preserve of course infs of powerset lattices, hence greatest elements,
but not necessarily sups and least elements.

The remaining transformers $|-\rangle$ and $[-|$ and their inverses
arise from $\langle -|$ and $|-]$ by opposition:
$|-\rangle= \langle -| \circ (-)^\mathit{op}$,
$|-\rangle^{-1} = (-)^\mathit{op}\circ \langle -|^{-1}$,
$[-|= |-] \circ (-)^\mathit{op}$ and
$[-|^{-1} = (-)^\mathit{op}\circ |-]^{-1}$. Taken together, the four
modal operators satisfy the laws of the $\MKA$ modalities outlined in
Section~\ref{sec:mka-pt} and those of the abstract sup/inf-preserving
transformers discussed in Section~\ref{sec:pt-backvwright}. They give
in fact semantics to the algebraic developments, when restricted to
mono-types, and once again yield the same rules for verification
condition in the state transformer and the relational semantics,
albeit in a more general categorical setting.

The categorical approach to predicate transformers outlined is not
new, apart perhaps from the emphasis on quantales and quantaloids. The
emphasis on monads is due at least to Manes~\cite{Manes92}. More
recently, Jacob's work on state-and-effect triangles~\cite{Jacobs17}
has explored similar connections and their generalisation beyond
sequential programs. A formalisation with Isabelle, which is further
discussed in Section~\ref{sec:isa-pt}, is a contribution of this
article.


\section{Assignments}\label{sec:discrete-store}

Two important ingredients for concrete program semantics and
verification condition generation are still missing: a mathematical
model of the program store and program assignments, and rules for
calculating $\wlp$s for these basic commands. To prepare for hybrid
programs (see Section~\ref{sec:hybrid-store} for a syntax) we model
stores and assignments as discrete dynamical systems over state
spaces.

Formally, a \emph{dynamical system}~\cite{Arnold,Teschl12} is an
\emph{action} of a monoid $(M,\star,u)$ on a set or state space $S$,
that is, a monoid morphism $\flow :M\to S\to S$ into the
\emph{transformation monoid} $(S^S,\circ, \id_S)$ on $S^S$. Thus, by definition,
\begin{equation*}
\flow\, (m\star n) = (\flow\, m) \circ (\flow\, n)\qquad\text{ and }\qquad
\flow\, u = \id_S.
\end{equation*} The first action axiom captures the inherent
determinism of dynamical systems. Conversely, each transformation
monoid $(S^S,\circ,\id_S)$ determines a monoid action in which the
action $\flow:S^S\to S\to S$ is function application.

States of simple while programs can be modelled simply as maps
$s:V\to E$ from program variables in $V$ to values in $E$. State
spaces for such discrete dynamical systems are function spaces
$S=E^V$.

An update function $f_a:V\to (S \to E) \to S\to S$ for assignment
commands can be defined as
\begin{equation*}
f_a\, v\, e\, s = s[v\mapsto e\, s],
\end{equation*}
where $f[a\mapsto b]$ updates $f:A\to B$ by associating $a\in A$ with
$b$ and every $y\neq a$ with $f\, y$.  The ``expression'' ${e:S\to E}$
is evaluated in state $s$ to $e\, s$. The maps $f_a\, v\, e$ generate
a transformation monoid, hence a monoid action $S^S\to S \to S$ on
$S^S$. They also connect the concrete program store semantics with the
$\wlp$ semantics used for verification condition generation.

We lift $f_a\, v\, e:S\to S$ to a state transformer
$v:=_\mathcal{F} e:S \to \Pow\, S$ as
\begin{equation*}
  (v:=_\mathcal{F} e) = \eta_S \circ (f_a\, v\, e) = \lambda s.\ \{ 
  f_a\, v\, e\, s\},
\end{equation*}
thus creating a semantic illusion for syntactic assignment commands in
the $\MKA$ $\sta\, S$.  For $\rel\, S$, the isomorphism between
$\Set_\Pow$ and $\Rel$ yields
\begin{equation*}
  (v :=_\mathcal{R} e) = \mathcal{R}\, (v:=_\mathcal{F} e),
\end{equation*}
hence
$(v :=_\mathcal{R} e) = \{(s,f_a\, v\, e\, s)\mid s\in
E^V\}=\{(s,s[v\mapsto e\, s])\mid s\in E^V\}$.
Alternatively, we could have defined the state transformer semantics
from the relational one via
$(v :=_\mathcal{F} e) = \mathcal{F}\, (v:=_\mathcal{R} e)$.

The $\wlp$s for assignment commands in $\rel\, S$ and $\sta\, S$ are
of course the same. Hence we drop the indices $\mathcal{F}$ and
$\mathcal{R}$ and write
\begin{equation}
|v:=  e] Q = \lambda s.\ Q\,
(s[v\mapsto e\, s]) = \lambda s.\ Q\,  (f_a\, v\, e\,
s).\label{eq:wlp-asgn}\tag{wlp-asgn}
\end{equation}

Adding the $\wlp$ law for assignments in either semantics to the
algebraic ones for the program structure suffices to generate
data-level verification conditions for while programs. 

The approach outlined so far is suited for building verification
components via shallow embeddings with proof assistants such as
Isabelle. The predicate transformer algebras of the previous sections,
as shown in the first row of Figure~\ref{fig:framework}, can all be
instantiated to intermediate state transformer and relational
semantics, as shown in Proposition~\ref{P:rel-mka} and
\ref{P:kleisli-mka} for $\MKA$. These form the second row in
Figure~\ref{fig:framework}. Each of these can be instantiated further
to concrete semantics with predicate transformers for assignments, as
described in this section.

In Isabelle, these instantiations are enabled by type polymorphism. If
modal Kleene algebras have type ${\isacharprime}a$, then the
intermediate semantics have the type of relations or state
transformers over ${\isacharprime}a$, and Proposition~\ref{P:rel-mka}
and \ref{P:kleisli-mka} can be formalised, so that all facts known for
$\MKA$ are available in the intermediate semantics.  The concrete
semantics then require another simple instantiation of the types of
relations or state transformers to those of program stores.  All facts
known for $\MKA$ and the two intermediate semantics are then available
in the concrete predicate transformer semantics for while programs. A
particularity of the semantic approach and the shallow embedding is
that assignment semantics are based on function updates instead of
substitutions---see the rule (\ref{eq:wlp-asgn})---so that an explicit
substitution calculus like that of $\dL$ is not needed. We can simply
rely on that of Isabelle/HOL.

The use of algebra and the modularity of the shallow semantics
simplify the construction of program verification
components~\cite{afp:vericomp} considerably.  The overall approach
discussed has been developed initially for Hoare logics
in~\cite{ArmstrongGS16}. It has been extended to predicate transformer
semantics based on $\MKA$ in~\cite{GomesS16}.


\section{Ordinary Differential Equations}\label{sec:ODE}

Before developing relational and state transformer models for the
basic evolution commands of hybrid programs in the next section, we
briefly review some basic facts about continuous dynamical systems and
ordinary differential equations.

Continuous dynamical systems $\flow:T\to S\to S$ are \emph{flows},
which often represent solutions to systems of ordinary differential
equations (ODEs)~\cite{Arnold,Hirsch09,Teschl12}. They are called
continuous because $T$, which models time, is assumed to form a
non-discrete submonoid of $(\reals,+,0)$, and the state space or phase
space $S$ is usually a manifold with topological structure.  By
definition, flows are monoid actions. Hence $\flow$ satisfies, for all
$t_1,t_2\in T$,
\begin{equation*}
\flow\, (t_1 + t_2) = \flow\,
t_1 \circ \flow\, t_2\qquad \text{ and } \qquad \flow\, 0 = \id.
\end{equation*}
We always assume that $T$ is an open interval in $\reals$ and $S$ an
open subset of $\reals^n$.  Beyond that, one usually assumes that
actions are compatible with the structure on $S$. As $S$ is a
manifold, we assume that flows are continuously differentiable.

The \emph{trajectory} of $\flow$ through state $s\in S$ is the
function $\flow_s:T \to S$ defined by
$\flow_s = \lambda t.\ \flow\, t\, s$, that is,
$\flow_s\, t=\flow\, t\, s$. It describes the system's evolution in
time passing through state $s$.

The \emph{orbit} of $s$ is the set of all states on the trajectory
passing through $s$, but not necessarily starting in this state. We
model it as the function $\orbit:S\to \Pow\, S$ defined by
\begin{equation*}
\orbit\, s= \Pow\, \flow_s\, T,
\end{equation*}
the canonical map sending each $s\in S$ to its equivalence class
$\gamma^\flow\, s$.  Orbit functions are state transformers, as their
type indicates. They form our basic semantics for evolution commands
and hybrid programs.

Flows arise from ODEs as follows.  In a system of ODEs
\begin{equation*}
x_i'\ t  = f_i\ (t, (x_1\, t),\dots,(x_n\, t)),\qquad (1\le i\le n),
\end{equation*}
each $f_i$ is a continuous real-valued function and
$t\in T\subseteq\reals$. Any such system can be made
time-independent---or \emph{autonomous}---by adding the equation
$x_0'\, t=1$. We henceforth restrict our attention to autonomous
systems and write
\begin{equation*}
X'\, t=
\begin{pmatrix}x_1'\, t\\ x_2'\, t\\ \vdots\\ x_n'\, t \end{pmatrix}=
\begin{pmatrix}f_1\, (x_1\, t) \dots (x_n\, t)\\ 
f_2\, (x_1\, t)\dots (x_n\, t)\\
\vdots\\ f_n\, (x_1\ t)\dots (x_n\, t)\end{pmatrix}= 
f\, (X\,t).
\end{equation*}
The continuous function $f:S\to S$ on $S\subseteq \reals^n$ is a
\emph{vector field}.  It assigns a vector to each point in $S$.

An autonomous system of ODEs is thus simply a
vector field $f$, and a \emph{solution} a continuously differentiable
function $X:T\to S$ that satisfies $X'\, t=f\, (X\,t)$ for all
$t\in T$, or more briefly $X'=f\circ X$.

An \emph{initial value problem} (IVP) is a pair $(f,s)$ of a vector
field $f$ and an initial value
$(0,s)\in T\times S$~\cite{Hirsch09,Teschl12}, where $t_0=0$ and $s$
represent the initial time and initial state of the system.  A
solution to the IVP $(f,s)$ satisfies
\begin{equation*}
X'=f\circ X\qquad \text{ and }\qquad X\, 0 = s.
\end{equation*} 

If solutions $X$ to an IVP $(f,s)$ are unique and $T=\reals$, then it
is easy to show that $X = \flow^f_s$ is the trajectory of the flow
$\flow^f$ through $s$.

Geometrically, $\flow_s^f$ is the unique curve in $S$
that is parametrised by $t$, passes through $s$ and is tangential to
$f$ at any point. As trajectories arise from integrating both sides of
$(\flow_s^f)'=f\circ\flow_s^f$, they are also called \emph{integral
  curves}. We henceforth write $\flow_s$ when the dependency on $f$
is clear.

The following example provides some physical intuition for readers
unfamiliar with these concepts. 

\begin{example}[Particles in fluid]\label{ex:fluid}
We use the autonomous system of ODEs
\begin{equation*}
x'\, t = v,\qquad
y'\, t = 0,\qquad
z'\, t = - \sin\,  (x\, t),
\end{equation*} 
where $v\in \reals\setminus\{0\}$ is a constant, as a simple model for the
movement of particles in a three-dimensional fluid. Its vector field
$f:\reals^3\to\reals^3$,
\begin{equation*}
f\begin{pmatrix}x\\ y\\ z\end{pmatrix} = 
\begin{pmatrix}v\\ 0\\ -\sin x\end{pmatrix},
\end{equation*} 
associates a velocity vector with each point of $S=\reals^3$
(vectors in Figure \ref{fig:velvecs}).

For each point $s=(s_1,s_2,s_3)^T$, the solutions
$\flow_s:\reals\to \reals^3$ of the IVP $(f,s)$ are uniquely
defined. They are the trajectories of particles through time passing
through state $s$ (dot and line in Figure \ref{fig:velvecs}),
given by
\begin{equation*}
\flow_s\, t=\begin{pmatrix}s_1\\ s_2\\ s_3\end{pmatrix}+\begin{pmatrix}vt\\ 0\\ \frac{\cos\,  (s_1+vt)}{v}-\frac{\cos
  s_1}{v}\end{pmatrix},
\end{equation*}
where we use juxtaposition without spaces as multiplication of real numbers.

Checking that they are indeed solutions to the IVP requires simple calculations:
\begin{align*}
\flow_s'\, t &=
\begin{pmatrix}v\\ 0\\ -\sin\, (s_1+vt)\end{pmatrix}=
f\, \begin{pmatrix}s_1 + vt\\ s_2\\ s_3+\frac{\cos\, (s_1+vt)}{v}-\frac{\cos s_1}{v}\end{pmatrix}
= f\, (\flow_s\, t),\\
  \flow_s\, 0 &= 
\begin{pmatrix}s_1\\ s_2\\ s_3\end{pmatrix}+\begin{pmatrix}v0\\ 0\\ \frac{\cos\, (s_1 + v0)}{v}-\frac{\cos s_1}{v}\end{pmatrix}
= 
\begin{pmatrix}s_1\\ s_2\\ s_3\end{pmatrix} = s.
\end{align*}
Checking that $\flow:\reals\to\reals^3\to\reals^3$,
$\flow\, t\, s =\flow_s\, t$, is a flow is calculational, too:
\begin{align*}
\flow\, t_1 (\flow\, t_2\, s)
&=\begin{pmatrix}s_1+vt_2\\ s_2\\ s_3+\frac{\cos\,  (s_1+vt_2)}{v}-\frac{\cos
  s_1}{v}\end{pmatrix}
+
\begin{pmatrix}
  vt_1\\
0\\
\frac{\cos\,  (s_1+vt_2+vt_1)}{v}-\frac{\cos\, (s_1+vt_2)}{v}
\end{pmatrix}\\
&=\begin{pmatrix}
s_1\\ 
s_2\\ 
s_3
\end{pmatrix}
+
\begin{pmatrix}
v(t_1+t_2)\\ 
0\\ 
\frac{\cos\, (s_1 + v(t_1 + t_2))}{v}-\frac{\cos\, s_1}{v}
\end{pmatrix}\\
&= \flow\, (t_1 + t_2)\, s.
\end{align*}
The condition $\flow\, 0\, s = s$ has already been checked. \qed
\end{example}

\begin{figure}
\begin{center}
\includegraphics[scale=0.5, trim=0 20 0 30, clip]{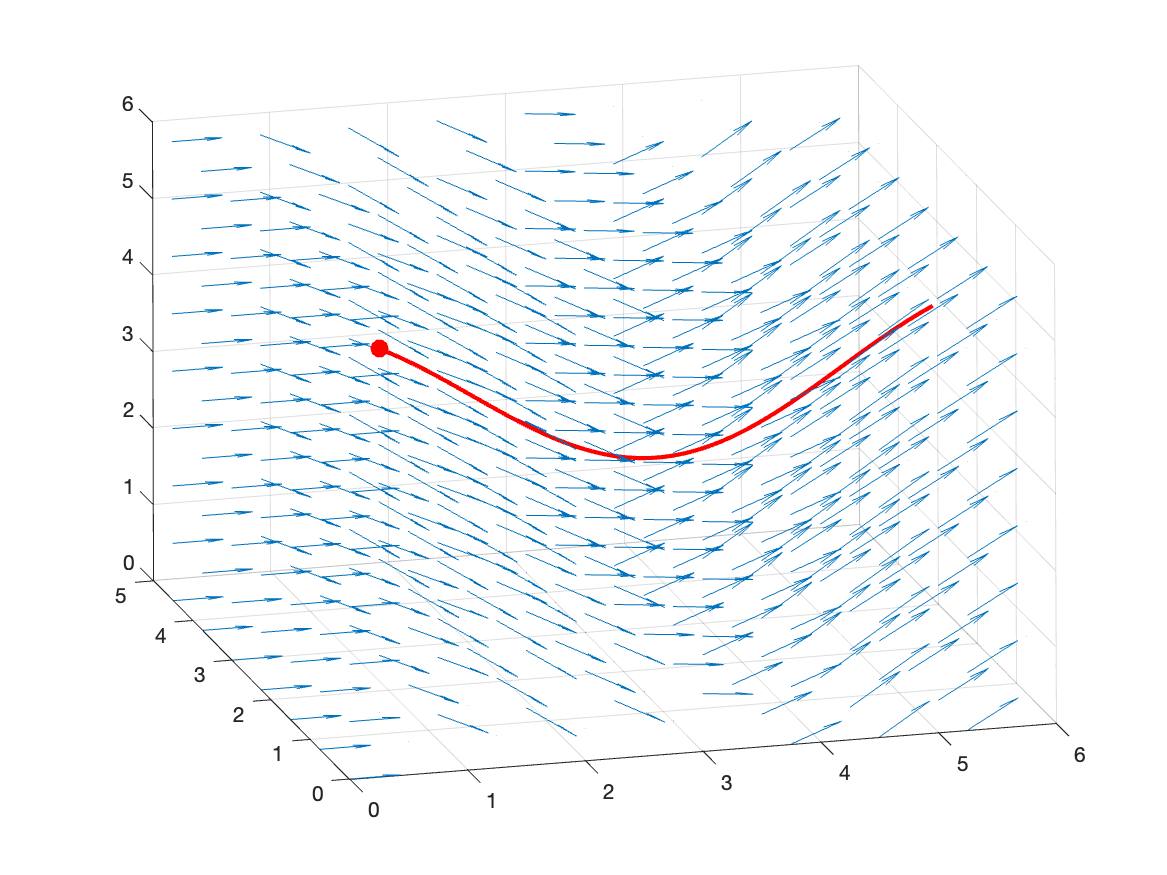} 
\caption{Vector field and trajectory for a particle in a fluid
(Example~8.1)}
\label{fig:velvecs} 
\end{center}
\end{figure}

It is well known that not all IVPs admit flows: not all of them have
unique solutions, and in many situations, flows exist locally on a
subset of $\reals$ that does not form a submonoid. Peano's theorem
guarantees the local existence of solutions for systems of ODEs whose
associated vector field is continuous. Conditions for local existence
and uniqueness are provided by the Picard-Lindel\"of
theorem~\cite{Hirsch09,Teschl12}, which we briefly discuss, as we use
it for our first workflow.

By the fundamental theorem of calculus, any solution to an IVP must
satisfy
\begin{equation*}
X\, t - X\, 0 = \int_0^t f\, (X\, \tau)\d\tau.
\end{equation*}
It can be shown that this equation holds if, for $X\, 0= s$, the
function
\begin{equation*}
h\, x\, t = s + \int_0^t f\, (x\, \tau)\d\tau
\end{equation*}
has a fixpoint. This, in turn, is the case if the limit $X$ of the sequence
$(h^n)_{n\in\mathbb{N}}$, defined by $h^0\, x\, t= s$ and
$h^{n+1}= h\circ h^n$, exists. Indeed, with this assumption,
\begin{equation*}
X\, t = \lim_{n\to\infty} \left(s + \int_0^t f\, (h^{n-1}\, \tau)d\tau \right) = s + \int_0^t f\, (X\, \tau)\d\tau,
\end{equation*}
using continuity of addition, integration and
$f$ in the second step.  Finally, existence of the limit of $(h^n)_{n\in\mathbb{N}}$ is
guaranteed by constraining the domain of the $h^n$, and by Banach's
fixpoint theorem there must be a \emph{Lipschitz constant}
$\lipschitz\geq 0$ such that
\begin{equation*}
  \lVert f\, s_1 - f\, s_2 \rVert \le \lipschitz \lVert s_1- s_2\rVert,
\end{equation*}
for any $s_1,s_2\in S$, where $\lVert-\rVert$ is the Euclidean norm on
$\reals^n$.  Vector fields satisfying this condition are called
\emph{Lipschitz continuous}.

\begin{theorem}[Picard-Lindel\"of]\label{P:picard-lindeloef}
  Let $S\subseteq\reals^n$ be an open set and $f:S\to S$ a Lipschitz
  continuous vector field. The IVP $(f,s)$ then has a unique solution
  $X:T_s\to S$ on some open interval $T_s\subseteq\reals$.
\end{theorem}

The Picard-Lindel\"of theorem makes it possible to patch together
intervals $T_s$ to a set
$U=\bigcup_{s \in S} T_s \times \{s\}\subseteq \reals \times S$, from
which a largest interval of existence $T=\bigcup_{s\in S} T_s$ can be
extracted.  One can then define a \emph{local flow}
$\flow:T\to S\to S$ such that $\flow_s\, t$ is the maximal integral
curve at $s$. The monoid action identities $\flow\, 0 = \mathit{id}$
and $\flow\, (t_1+t_2)\, s = \flow\, t_1 (\flow\, t_2\, s)$ can thus
be shown for all $t_2,t_1+t_2\in T_s$~\cite{Teschl12}, but $U$ need
not be closed under addition. The Picard-Lindel\"of theorem, in the
form presented, thus provides sufficient conditions for the existence
and uniqueness of local flows for autonomous systems of ODEs.  Flows
are global and hence monoid actions if $T$ is equal to $\reals$ or its
non-negative or non-positive subset.

Hybrid systems often deal with dynamical systems where $T=T_s= \reals$
for any $s\in S$ and $S$ is isomorphic to $\reals^n$ for some
$n\in\mathbb{N}$. Our approach supports local flows with
$T\subset \reals$ and $S\subset \reals^n$ as well, and even IVPs with
multiple solutions beyond the realm of Picard-Lindel\"of.


\section{Evolution Commands for Lipschitz Continuous Vector Fields}\label{sec:hybrid-store}

Simple hybrid programs of $\dL$~\cite{Platzer10} are defined by the syntax
\begin{equation*}
\mathcal{C}\ ::= \ x:=e \mid x' = f \, \&\, G \mid ?P\mid \mathcal{C};\mathcal{C}\mid \mathcal{C}+\mathcal{C}\mid \mathcal{C}^*,
\end{equation*}
which adds \emph{evolution commands} $x' = f \ \&\ G$ to the program
syntax of dynamic logic.  Intuitively, evolution commands introduce a
vector field $f$ for an autonomous system of ODEs and a \emph{guard}
$G$, which models boundary conditions or similar constraints that
restrict temporal evolutions. Guards are also known as \emph{evolution
  domain restrictions} or \emph{invariants} in the hybrid automata
literature~\cite{DoyenFPP18}, but henceforth we consistently refer to
them as ``guards''.  Nondeterministic choice and finite iteration can
be adapted for modelling conditionals and while loops as with $\MKA$
or predicate transformer semantics.

We are only interested in the semantics of hybrid programs. Relative
to the semantics of standard while programs, it thus remains to define
the $\wlp$s for evolution commands. This requires relational and state
transformer semantics for evolution commands over hybrid program
stores.  In this section we describe our first workflow that certifies
solutions using the Picard-Lindel\"of theorem. We thus assume that
vector fields are Lipschitz continuous, such that the
Picard-Lindel\"of theorem guarantees at least local flows. This is
more general than needed for dynamical systems. A further
generalisation to continuous vector fields is presented in the next
section in preparation for our second, more powerful workflow.

We begin with hybrid program stores for $\dL$~\cite{Platzer12}. These
are maps $s:V\to\reals$ that assign real numbers to program variables
in $V$. Variables may appear both in differential equations and the
discrete control of a hybrid system.  One usually assumes that $|V|=n$
for some $n\in\mathbb{N}$, which makes $\reals^V$ isomorphic to the
vector space $\reals^n$. The results from Section \ref{sec:ODE} then
apply to any state space $S\subseteq\reals^V$.

Next we describe a state transformer semantics and a $\dL$-style
relational semantics of evolution commands with Lipschitz continuous
vector fields.  Intuitively, the semantics of $x' = f \, \&\, G$ in
state $s\in S\subseteq \reals^V$ is the longest segment of the
trajectory $\flow^f_s$ at $s$ along which all points satisfy $G$.

For the remainder of this section, we fix  a Lipschitz 
continuous vector field $f:S\to S$ and a guard $G:S\to\bools$, for 
$S\subseteq \reals^V$. We freely consider $G$, and any other 
function of that type, as a set or a predicate. As explained in 
Section~\ref{sec:ODE}, there is a (local) flow $\flow:T\to S\to S$ 
defined on a maximal interval $T\subseteq \reals$ with $0\in T$. 
Thus, we can pick any interval $U\subseteq T$ with $0 \in U$ to 
compute $wlp$s over subintervals of the interval of existence $T$.
In examples, we typically use the subinterval $[0,t]$, from
the time at which the system dynamics starts to a maximal time $t$ of
interest, or the subinterval $\reals_+$, the set of non-negative real
numbers.

For each $t\in U$, let ${\downarrow} t= \{t'\in U \mid t' \le t\}$. The 
$G$-\emph{guarded orbit} on $U$ at $s\in S$ is then defined as
$\gamma^\flow_{G,U}:S\to \Pow\, S$ by
\begin{equation*}
  \gamma^\flow_{G,U}\, s  = \bigcup\{\Pow\, \flow_s\, {\downarrow}t
    \mid t\in U \land \Pow\, \flow_s\, {\downarrow}t \subseteq G\}.
\end{equation*}
Intuitively, $\gamma^\flow_{G,U}\, s$ is the orbit at $s$ defined
along the longest interval of time in $U$ that satisfies guard
$G$.  This intuition is more apparent in the following lemma.
\begin{lemma}\label{P:g-orbit-props}
Let $s\in S$. Then
\begin{enumerate}
\item $ \gamma^\flow_{G,U}\, s   = \bigcup\{\gamma^{\flow|_{{\downarrow}t}}\, s
    \mid  t\in U\land \gamma^{\flow|_{{\downarrow}t}}\, s\subseteq G\}$, 
\item $\gamma^\flow_{G,U}\, s= \{\flow_s\, t
    \mid t\in U \land \forall \tau\in {\downarrow}t.\ G\, (\flow_s\,
    \tau)\}$. 
\end{enumerate}
\end{lemma}
We have not formalised (1) with Isabelle because reasoning with
partial functions may be tedious.  As a special case, for $U=T_+$, any
subinterval of $\reals_+$,
\begin{equation*}
\gamma^\flow_{G,T_+}\, s = \{\flow_s\, t \mid t\in T_+\land \forall
\tau\in [0,t].\ G\, (\flow_s\, \tau)\}.
\end{equation*}

We can now define the state transformer semantics of $x'= f\, \&\, G$
simply as
\begin{equation*}
{(x'=_\mathcal{F} f\, \&\, G)_U} = \gamma^\flow_{G,U}.
\end{equation*}
Hence the denotation of an evolution command in state $s$ is the
guarded orbit at $s$ in time interval $U$.  Alternatively, in $\rel\,
S$, 
\begin{equation*}
{(x'=_\mathcal{R} f\, \&\, G)_U} = \mathcal{R}\, {(x'=_\mathcal{F} f\, \&\,
  G)_U} = \{(s,\flow\, t\, s)\mid t\in U\land \forall \tau\in {\downarrow}t.\ G\, (\flow_s\,
    \tau)\}
\end{equation*}
like in Section \ref{sec:discrete-store}. Restricting this further to
$U=\reals_+$ yields the standard semantics of evolution 
commands of $\dL$.

It remains to derive the $\wlp$s for evolution commands. These are the
same in $\rel\, S$ and $\sta\, S$, so we drop $\mathcal{F}$ and
$\mathcal{R}$.

\begin{proposition}\label{P:wlpprop}
  Let $Q: S\to\bools$. Then
\begin{equation*}
  |{(x'=f\, \&\, G)_U}] Q  = \lambda s\in S.\ \{s\mid \forall t\in U.\ 
  \Pow\, \flow_s\, {\downarrow}t \subseteq G \rightarrow \Pow\, \flow_s\, {\downarrow}t
\subseteq Q\}.
\end{equation*}
\end{proposition}
\noindent By Lemma~\ref{P:g-orbit-props}, alternatively,
\begin{equation*}
|{(x'=f\, \&\, G)_U}] Q  
= \lambda s\in S.\ \{s\mid \forall t\in U.\
  \gamma^{\flow|_{{\downarrow}t}}\, s \subseteq G \rightarrow \gamma^{\flow|_{{\downarrow}t}}\,
  s\subseteq Q\}.
\end{equation*}
For verification condition generation, the following variant is most useful.
\begin{lemma}\label{P:wlpprop-var}
  Let $Q: S\to\bools$. Then
\begin{equation}
|{(x'=f\, \&\, G)_U}] Q  = \lambda s\in S.\forall t\in U.\ (\forall
\tau\in {\downarrow}t.\ G\, (\flow_s\, \tau)) \rightarrow Q\, (\flow_s\, t).\label{eq:wlp-evl}\tag{wlp-evl}
\end{equation}
\end{lemma}
In particular, for $T=\reals$ and $U=\reals_+$, 
\begin{equation*}
|{(x'=f\, \&\, G)_{\reals_+}}] Q  = \lambda s\in S.\forall t\in \reals_+.\ (\forall
\tau\in [0,t].\ G\, (\flow_s\, \tau)) \rightarrow Q\, (\flow_s\,
t).
\end{equation*}

Accordingly, and consistently with $\dL$, $Q$ is no longer a
postcondition in the traditional sense: by definition it is supposed
to hold along the trajectory and therefore on any orbit at any
particular initial condition $s$ guarded by $G$.

For a more categorical view on the weakest liberal precondition of evolution commands,
remember from Section~\ref{sec:pt-monad} that 
$\langle (x'=f\, \&\, G)_U| = (\gamma^\flow_{G,U})^\dagger$, where
$(-)^\dagger$ is the Kleisli extension map, and that the $\wlp$ of
$(x'=f\, \&\, G)_U$ is its right adjoint. It therefore satisfies
\begin{align*}|(x'=f\,
\&\, G)_U]P 
=\bigcup\{Q\mid (\gamma^\flow_{G,U})^\dagger\,  Q \subseteq P\}
= \{s \mid \gamma^\flow_{G,U}\, s \subseteq P\}.
\end{align*}
The identity in Proposition~\ref{P:wlpprop} can then be calculated
from there.

The $\wlp$ laws in Proposition~\ref{P:wlpprop} and
Lemma~\ref{P:wlpprop-var} complete the laws for verification condition
generation for hybrid programs in the relational and state transformer
semantics.  In practice, Proposition~\ref{P:wlpprop},
Lemma~\ref{P:wlpprop-var} and the Picard-Lindel\"of theorem support
our first workflow for computing the $\wlp$ of an evolution command
$x'=f\, \&\, G$ on a set $U$ for a Lipschitz continuous vector field:
\begin{enumerate}
\item check that the vector field $f$ is indeed Lipschitz continuous and
  $S\subseteq R^V$ open;
\item supply the (local) flow $\flow$ for $f$ with $U$, a subinterval of 
the interval of existence around $0$;
\item certify that $\flow_s$ is indeed the unique solution for $(f,s)$
  for any $s\in S$ and for $U$:
  \begin{enumerate}
  \item $\flow_s' = f \circ \flow_s$ on $U$ for any $s\in S$,
\item $\flow_s\, 0 = s$ for any $s\in S$, 
\item $U$ is subset of open set $T$ with $0\in U$;
  \end{enumerate}
\item if successful, apply the identity in Proposition~\ref{P:wlpprop} or Lemma~\ref{P:wlpprop-var}.
\end{enumerate}
In practice, computer algebra tools are helpful for finding flows.
Their integration into proof assistants for this purpose is routine
and therefore not pursued in this article. The existence of unique
solutions can be guaranteed uniformly, for instance, for affine or
linear systems of ordinary differential equations. See~\cite{Munive20}
for the formalisation of such an approach with Isabelle.

The following classical example illustrates our algebraic approach and
gives a first glimpse of the mathematics involved. It should be noted
that we are not embellishing our natural semantical notation with any
fa\c{c}ade program syntax in this article; see~\cite{FosterGMS21} for
such an extension. A formal verification with Isabelle can be found in
Example~\ref{ex:bouncing-ball-flow} below.

\begin{example}[Bouncing ball]\label{ex:ball}
  A ball of mass $m$ is dropped from height $h\geq 0$. Its state space
  is $s\in\reals^V$ for $V=\{x,v\}$, where $x$ denotes its position
  and $v$ its velocity. Its kinematics is specified by the vector field
  $f:\reals^V\to\reals^V$ with

  \begin{equation*}
    f\,
    \begin{pmatrix}
      s_x\\
      s_v
    \end{pmatrix}
=
\begin{pmatrix}
  s_v\\
-g
\end{pmatrix},
\end{equation*}

\noindent where $g$ is the acceleration due to gravity and we
abbreviate $s_x = s\, x$ and $s_v = s\, v$.  The ball is assumed to
bounce back from the ground in an elastic collision. This is modelled
using a discrete control, which checks for $s_x=0$ and then flips the
velocity.  A guard $G=(\lambda s.\ s_x\geq 0)$ precludes any motion
below the ground. The system is modelled by the hybrid
program~\cite{Platzer10}
\begin{align*}
	\mathsf{Cntrl} &= \IF {(\lambda\, s.\ s_x=0)} {v:=(\lambda\, s.\ - s_v)} \mathit{skip},\\
	\mathsf{Ball} &= ({x'=f\, \&\, G}\,  {;}\, \mathsf{Cntrl})^\ast,
\end{align*}
where $\isa{skip}$ denotes the program that maps each state to itself
(represented by $1$ in $\MKA$). Its correctness specification is
\begin{equation*}
    P\leq|\mathsf{Ball}]Q\qquad\text{ for }\quad
  P= (\lambda s.\ s_x = h\land s_v = 0)\quad\text{ and }\quad Q  = (\lambda s.\ 0\leq s_x\leq h).
\end{equation*}
We also need the loop invariant
\begin{equation*}
      I = \left(\lambda s.\ 0\le s_x \land \frac{1}{2}s_v^2= g(h - s_x)\right),
\end{equation*}
which uses a variant of energy conservation with $m$ cancelled out. 

The first step of our verification proof shows that $P\le I$ and
$I\le Q$. The first inequality holds because
$\frac{1}{2} 0^2 = 0 = h - h$; the second one because $0\le s_x$
appears both in $I$ and in $Q$ and because $s_x \le h$ is guaranteed
by $g(h-s_x)\ge 0$, which holds as $\frac{1}{2}s_v^2 \ge 0$. With
transitivity and isotonicity of boxes, we can thus bring the
correctness specification into the form $I \leq|\mathsf{Ball}]I$.

Applying (\ref{eq:wlp-star}) then yields the proof obligation
$I\le |{x'=f\ \&\ G}\, {;}\, \mathsf{Cntrl}]I$.  To discharge it, we
use (\ref{eq:wlp-seq}) to calculate the $\wlp$s
\begin{align*}
J &=|\IF {(\lambda\, s.\ s_x=0)} {v:=(\lambda\, s.\ - s_v)} \mathit{skip}]I,\\
K &=|{x'=f\, \&\, G}]J
\end{align*}
incrementally and finally show that $I\le K$.

For the first $\wlp$ we calculate, with (\ref{eq:wlp-cond}) and for $T=
(\lambda\, s.\ s_x=0)$, 
\begin{align*}
   J &= (T \to |v:=(\lambda\, s.\ - s_v)] I) \cdot (\overline{T} \to I)\\
&= \left( T \to |v:=(\lambda\, s.\ - s_v)] \left(\lambda s.\ 0\le s_x
   \land \frac{1}{2}s_v^2= g(h - s_x)\right)\right) \cdot (\overline{T} \to I)\\
 &= \left(T \to \left(\lambda s.\ 0\le s_x \land \frac{1}{2}(-s_v)^2=
   g(h - s_x)\right)\right) \cdot (\overline{T} \to I)\\
 &= (T \to I) \cdot (\overline{T} \to I)\\
&= I.
\end{align*}
For the second $\wlp$, we wish to apply (\ref{eq:wlp-evl}). This
requires checking that $f$ is Lipschitz continuous---$\lipschitz = 1$
does the job, supplying a flow and checking that it solves the IVP
$(f,s)$ for all $s\in S$ and satisfies the flow conditions for
$T=\reals$ and $S=\reals^V$.  We leave it to the reader to verify that
$\flow:\reals\to\reals^V\to\reals^V$ defined by
\begin{equation*}
  \flow_s\, t = 
  \begin{pmatrix}
    s_x\\
     s_v
  \end{pmatrix}
+
\begin{pmatrix}
  s_v\\
-g
\end{pmatrix}
t
-
\frac{1}{2}\begin{pmatrix}
  g\\
  0
\end{pmatrix}
t^2
\end{equation*}
meets the requirements in the procedure outlined above,
cf. Example~\ref{ex:fluid}.  Then, expanding definitions and applying
(\ref{eq:wlp-evl}) from Lemma~\ref{P:wlpprop-var},
\begin{align*}
&K\, s \\
&= \left(\forall t\in\reals_+.\ (\forall\tau\in
                        [0,t].\ 0 \le \flow_s\, \tau\, x)\rightarrow 0\le
                                      \flow_s\, t\,
                                      x \land \frac{1}{2}\left(\flow_s\, t\, v\right)^2=
                                      g(h - \flow_s\, t\, x)\right)\\
\ & = \left(\forall t.\ (\forall\tau\in
                        [0,t].\ 0 \le \flow_s\, \tau\, x)\rightarrow \frac{1}{2}(\flow_s\, t\, v)^2=
                                      g(h - \flow_s\, t\, x)\right)\\
\ & = \left(\forall t. \left(\forall\tau\in
                        [0,t].\ 0 \le
  s_x + s_vt -\frac{1}{2}g\tau^2\right)\right.\\
& \qquad\qquad\left.\rightarrow \frac{1}{2}(s_v-gt)^2=
                                      g\left(h -s_x - s_v t + \frac{1}{2}gt^2\right)\right).
\end{align*}
Finally, for $I\le K$, suppose $0\le s_x$,
$\frac{1}{2}s_v^2= g\left(h - s_x\right)$ and
$0 \le s_x+s_v\tau -\frac{1}{2}g\tau^2$ for all $\tau\in [0,t]$.  It
remains to show that
$\frac{1}{2}\left(s_v-gt\right)^2= g\left(h -s_x-s_vt
  +\frac{1}{2}gt^2\right)$. Indeed, using the second assumption in the
second step,
\begin{align*}
  \frac{1}{2}(s_v-gt)^2 
&= \frac{1}{2}s_v^2-g\left(s_v t+\frac{1}{2}gt^2\right)\\
&= g(h-s_x)-g\left(s_v t+\frac{1}{2}gt^2\right) \\
&=g\left(h-s_x +s_vt + \frac{1}{2}gt^2\right).
\end{align*}
The verification with Isabelle described in
Example~\ref{ex:bouncing-ball-flow} is far more automatic than this
proof on paper suggests, and there is ample scope for further
automation.  As already pointed out: the main purpose of this example
is to illustrate our first workflow and give an impression of the
mathematical reasoning involved. \qed
\end{example}

Certifying solutions of systems of ODEs can be tedious and hard to
automate and many ODEs do not admit analytic solutions.  It is
possible to circumvent these obstacles to practical verification
applications in various ways.  One approach, using invariant sets for
systems of ODEs, is pursued by $\dL$ and described in the following
sections. It constitues the second workflow supported by our
framework.  Another approach aims at particular types of vector fields
for which (global) flows always exist and are easy to compute. A
classical example are linear systems of ODEs~\cite{Hirsch09,Teschl12},
for which the first author has already developed methods in a
successor article~\cite{Munive20}.  A final approach abandons
differential equations and vector fields altogether and starts from
flows---as known from hybrid automata~\cite{DoyenFPP18}. This requires
changing the syntax of hybrid programs. The approach is outlined in
Section~\ref{sec:flow-component}. It constitutes the third workflow
supported by our framework.


\section{Evolution Commands for Continuous Vector
  Fields}\label{sec:generalisation}

As the semantic approach to evolution commands developed in the
previous section depends mainly on orbits, which are nothing but sets
of states, it can be generalised beyond trajectories and flows. In
this section we drop the requirement of uniqueness of solutions to
IVPs and hence assume that vector fields are merely continuous. In
fact, if vector fields are non-continuous, the set of solutions
defined below will simply be empty. We therefore generalise the
definitions in the previous section to obtain weakest liberal
preconditions for evolution commands that do not admit unique
solutions, for instance, IVPs of the form $x'\, t=k\sqrt{x\, t}$ with
$x\, 0=0$ for any $k\in \reals$~\cite{HubbardW91}. Our second workflow
using invariant sets is based on this generalisation. 

Consider the IVP $(f,s)$ for  continuous vector field $f:S\to S$ and initial state
$s\in S\subseteq\reals^V$.  Let
\begin{equation*}
\Sols f\, T\, s = \{X \mid \forall t\in T.\  X'\, t = f\, (X\, t)\land X\, 0 = s\}
\end{equation*}
denote its set of solutions on $T\subseteq \reals$ with $0\in T$.  
Here, $T$ is no longer the maximal interval of existence defined by 
the Picard-Lindel\"of theorem; it can be changed like the set
$U$ in the previous section. Then each solution $X$ is still
continuously differentiable and thus $f\circ X$ integrable in
$T$.

For all $X\in \Sols\, f\, T\, s$ and $G:S\to\bools$, we define the $G$-\emph{guarded
  orbit} of $X$ along $T$ in $s$ via the function $\gamma^X_G:S\to \Pow\, S$ as
\begin{equation*}
  \gamma^X_{G}\, s = \bigcup\{\Pow\, X\,
  {\downarrow}t \mid t\in T\land \Pow\, X\, {\downarrow}t\subseteq G\},
\end{equation*}
which simplifies to
$\gamma^X_{G}\, s= \{X\, t\mid t\in T\land \forall \tau\in
{\downarrow}t.\ G\, (X\, \tau)\}$.  By Kneser's
theorem~\cite{Kneser1923}, when non-uniqueness occurs at some point,
infinitely many solutions exist for it. Thus, we define the
$G$-\emph{guarded orbital} of $f$ along $T$ in $s$ via the function
$\gamma^f_G:S\to \Pow\, S$ as
\begin{equation*}
  \gamma^f_G\ s = \bigcup\{\gamma^X_G\, s\mid X\in \Sols\, f\, T\, s\}.
\end{equation*}
We thus patch the guarded orbit of each solution to the associated IVP
together so that $\gamma^f_G\ s$ represents all possible evolutions in
time that pass through $s$. This is evident from the following result.
\begin{lemma}\label{P:gorbital}
Let $f:S\to S$ be continuous and $G:S\to \bools$. Then
  \begin{equation*}
  \gamma^f_G\, s = \{X\, t \mid  t\in T \land \Pow\, X\,
  {\downarrow}t\subseteq G \land X\in \Sols\, f\, T\, s\}.
\end{equation*}
\end{lemma}
If $G=\top$, the constantly true predicate on $S$ or the set $S$
itself, we simply write $\gamma^f$ instead of $\gamma^f_\top$. 

The state transformer semantics of the evolution command for a
continuous vector field $f$ can then be defined as
\begin{equation*}
{(x'=_\mathcal{F} f\, \&\, G)} = \gamma^f_G.
\end{equation*}
The corresponding relational semantics is 
\begin{equation*}
{(x'=_\mathcal{R} f\, \&\, G)}  = \{(s,X\, t)\mid t\in T\land
  \forall \tau\in {\downarrow}t.\ G\, (X\, \tau) \land X\in \Sols f\, T\, s\}.
\end{equation*}
Once again, $\langle x'= f\, \&\, G| = (\gamma^f_G)^\dagger$. This
leads to a $\wlp$ for evolution commands.
\begin{proposition}\label{P:wlpprop-gen}
  Let $S\subseteq\reals^V$ and $T\subseteq \reals$. Let $f: S\to S$ be a
  continuous vector field and $G,Q: S\to\bools$. Then
\begin{align*}
|x'=f\, \&\, G] Q 
 = \lambda s\in S.\ \{s\mid \forall X\in \Sols\, f\, T\, s.\forall t\in T.\ 
  \Pow\, X\, {\downarrow}t \subseteq G \rightarrow \Pow\, X\, {\downarrow}t
  \subseteq Q\}.
\end{align*}
\end{proposition}
This identity can be rewritten, for predicates,  as
\begin{equation*}
|{x'=f\, \&\, G}] Q  = \lambda s\in S.\forall X\in \Sols\, f\, T\, s.\forall t\in T.\ (\forall
\tau\in {\downarrow}t.\ G\, (X\, \tau)) \rightarrow Q\, (X\, t).
\end{equation*}

Whether this fact is useful for verification applications, as outlined
above, remains to be seen. Yet the next section shows that it is
certainly useful for reasoning with invariant sets. The following
corollary is important for verification proofs with invariants as
well.
\begin{corollary}\label{P:wlpprop-gen2}
  Let $f:S\to S$, $S\subseteq\reals^V$, be a continuous vector field,
  $T\subseteq\reals$
  and $G,Q: S\to\bools$. Then
\begin{equation*}
|{x'=f\, \&\, G}] Q  = |{x'=f\, \&\, G}](G\cdot Q).
\end{equation*}
\end{corollary}


\section{Invariants for Evolution
  Commands}\label{sec:differential-invariants}

In $\dL$, differential invariants are predicates $I$ that satisfy
$I\leq |{x'=f\ \&\ G}] I$~\cite{Platzer12}. In the terminology of
Section~\ref{sec:mka-pt}, they are simply invariants for evolution
commands. They play a crucial role in $\dL$ and KeYmaera X because of
the limited support for solving ODEs and their greater generality.

In dynamical systems theory, when all guards are $\top$ and global
flows exist, and in (semi)group theory, \emph{invariant sets} for
actions or flows $\flow:T\to S\to S$ are sets $I\subseteq S$
satisfying $\gamma^\flow\, s\subseteq I$ for all
$s\in I$~\cite{Teschl12}.  Based on the results from
Section~\ref{sec:generalisation}, we generalise both notions uniformly.

A predicate or set $I:S\to\bools$ is an \emph{invariant} of the 
continuous vector field $f:S\to S$ and guard $G:S\to\bools$
\emph{along} $T\subseteq \reals$ if
\begin{equation*}
 (\gamma^f_G)^\dagger\,  I\subseteq I.
\end{equation*}
Note that the parameter $T$ is hidden in the definition of
$\gamma^f_G$. 
For $G=\top$, when $(\gamma^f)^\dagger\,  I\subseteq I$, we 
call $I$ simply an \emph{invariant} of $f$ along $T$. 

The following proposition yields a structural insight in the
relationship between invariant sets of dynamical systems and
differential invariants of $\dL$ in terms of an adjunction.
\begin{proposition}\label{P:inv-prop}
  Let $f:S\to S$ be continuous, $G:S\to\bools$ and
  $T\subseteq \reals$. Then the following are equivalent.
\begin{enumerate}
\item $I$ is an invariant for $f$ and $G$ \emph{along} $T$;
\item $\langle x'= f\, \&\, G | I \subseteq I$; 
\item $I \subseteq | x'= f\, \&\, G ] I$.
\end{enumerate}
\end{proposition}
\begin{proof}
 \begin{equation*}
   (\gamma^f_G)^\dagger I \subseteq I
   \leftrightarrow \langle x'= f\, \&\, G | I \subseteq I
   \leftrightarrow I \subseteq | x'= f\, \&\, G ] I.
  \end{equation*}
  The first step uses the definition of backward diamonds as Kleisli
  extensions in Section~\ref{sec:pt-monad} and that of the semantics
  of evolution commands in Section~\ref{sec:generalisation}. The final
  step uses the adjunction between boxes and diamonds from
  Section~\ref{sec:mka-pt}.
\end{proof}
For our $\wlp$-calculus, condition (3) is of course most useful. Yet
instead of checking that a flow is a solution to a vector field, as
previously, we now need to check whether a predicate is an
invariant---without having to solve the system of ODEs. This may in
some case be a condicio sine qua non and in others a considerable
simplification of reasoning. The following lemmas lead to our second
workflow. We show some proofs although they have been formalised with
Isabelle, as they explain why the approach works.

First, towards Corollary~\ref{P:inv-props3} below, we may ignore guards when 
checking for invariants and we can use a simple second-order formula.

\begin{lemma}\label{P:inv-props2}
Let $f:S\to S$ be continuous and $I:S\to\bools$. Then
\begin{enumerate}
\item $  I \subseteq |x'= f\, \&\, \top]I  \rightarrow I \subseteq |x'= f\,
\&\, G]I$, 
\item  $I \subseteq |x'= f\, \&\, \top]I  \leftrightarrow \left(I\, s \to \forall X \in \Sols\, f\, T\, s.\forall t\in T.\ I\,
  (X\, t)\right)$. 
\end{enumerate}
\end{lemma}
\begin{proof}
  For (1), $\gamma^f_G\subseteq \gamma^f$ for all $G$ and hence
  $\langle x'=f\, \&\, G|I \subseteq \langle x'=f\, \&\top|I \subseteq
  I$. The proof of (2) is a simple calculation.
\end{proof}
Second, we can recurse over predicates as follows.
\begin{lemma}\label{P:invrules}
  Let $f:S\to S$ be a continuous vector field, $\mu,\nu:S\to\reals$
  differentiable and $T\subseteq \reals$ with $0\in T$. 
\begin{enumerate}
\item If $(\mu\circ X)' =(\nu\circ X)'$ for all $X$ such that $X'\, t = f\, t\, (X\, t)$ and $G\, (X\, t)$ when $t\in T$, then $\mu = \nu$ is an invariant for $f$ along $T$,
\item if $(\mu\circ X)'\, \tau\leq(\nu\circ X)'\, \tau$ when $\tau> 0$, and $(\mu\circ X)'\, \tau\geq(\nu\circ X)'\, \tau$ when $\tau< 0$, for all $X$ such that $X'\, t = f\, t\, (X\, t)$ and $G\, (X\, t)$,
  then both $\mu < \nu$ and $\mu \leq \nu$ are invariants for $f$ along $T$,
\item if $\mu < \nu$ and $\nu < \mu$ are invariants for $f$ along $T$, then $\mu\neq \nu$ is too (and conversely if $0$ is the least element in $T$),
\item $\mu \not\le \nu$ is an invariant for $f$ along $T$ if and only if $\nu < \mu$ is too.
\end{enumerate}
\end{lemma}
\begin{proof}
  We only show the proof of (1), as it reveals the main idea of the
  procedure outlined below.  By definition, $\mu = \nu$ is an
  invariant for $f$ along $T$ if and only if $\mu\, s = \nu\, s$
  implies $\mu\, (X\, t) = \nu\, (X\, t)$ for all
  $X\in\Sols\, f\, T\, s$. It is a well known consequence of the mean
  value theorem that two continuously differentiable functions are the
  same if and only if they intersect at some point and have the same
  derivative.  Hence $(\mu\circ X)' =(\nu\circ X)'$ and
  $\mu\, s = \nu\, s$ imply $\mu\, (X\, t) = \nu\, (X\, t)$ for all
  $X\in\Sols\, f\, T\, s$.
\end{proof}

Proposition~\ref{P:wlpprop-gen2}, the properties 
in this section---in particular Lemma~\ref{P:invrules}---and 
Lemma~\ref{P:inv-lemma} about invariants that are conjunctions or 
disjunctions support our second workflow for proving a 
correctness specification $P\le |x' = f\, \&\, G]Q$.
\begin{enumerate}
\item Check whether a candidate predicate $I$ is a differential
  invariant:
  \begin{enumerate}
\item transform $I$ into negation normal form;
  \item if $I$ is complex, reduce it with
    Lemma~\ref{P:inv-lemma}, and Proposition~\ref{P:invrules}(3) and (4);
\item if $I$ is atomic, apply Proposition~\ref{P:invrules}(1) and (2);
  \end{enumerate}
(if successful,  $I\le |x' = f\, \&\, G]I$ holds by
Proposition~\ref{P:inv-prop}(3) and Lemma~\ref{P:inv-props2});
\item if successful, prove $P\le I$ and $|x' = f\, \&\, G](G\cdot I)
  \le |x' = f\, \&\, G]Q$.
\end{enumerate}

For $G=\top$ and Lipschitz continuous vector fields, the notions of
invariant can be strengthened.
\begin{corollary}\label{P:inv-props3}
Let $f:S\to S$ be Lipschitz continuous. Then the following are equivalent.
\begin{enumerate}
\item $I$ is an invariant for $f$ along $T$;
\item $\langle x'= f\, \&\, \true | I = I$; 
\item $I = | x'= f\, \&\, \true ] I$.
\end{enumerate}
\end{corollary}
The identities (2) and (3) hold because $0\in T$.

Next we revisit the bouncing ball example from
Section~\ref{sec:hybrid-store} to illustrate our second work flow
that reasons with differential invariants. Once again we give detailed
mathematical calculations to indicate the kind of mathematical
reasoning involved. A verification with Isabelle, which is much more
automatic, can be found in Example~\ref{ex:bouncing-ball-inv}.

\begin{example}[Bouncing ball with differential
  invariant]\label{ex:ball-inv}
  We can avoid solving the system of ODEs in
  Example~\ref{ex:ball} using a differential invariant to show that
\begin{equation*}
  I\le |{x'=f\, \&\, G}]I
\end{equation*}
for the loop invariant $I$ and vector field
$f\, (s_x,s_v)^T = (s_v,-g)^T$.  The most natural candidate for a
differential invariant is of course energy conservation.  Cancelling
the mass, we use
\begin{equation*}
I_d = \left(\lambda s.\ \frac{1}{2} s_v^2=g (h-s_x)\right).
\end{equation*}
We now apply our procedure for reasoning with differential
invariants.
\begin{enumerate} 
\item We use Proposition \ref{P:invrules} with
  $\mu\, s = \frac{1}{2}s_v^2$ and $\nu\, s = g(h- s_x)$ to check that
  $I_d$ is indeed an invariant.  We thus need to show that
  $(\mu\circ X)' =(\nu\circ X)'$ for all $X\in \Sols\, f\, T\, s$,
  which unfolds to
\begin{equation*}
\left(\frac{1}{2}(X\, t\, v)^2\right)' = g (h - X\, t\, x)',
\end{equation*}
because $s= X\, t$ and therefore $s_v = X\, t\, v$ and $s_x = X\, t\, x$.
And indeed, 
\begin{align*}
\left(\frac{1}{2}(X\, t\, v)^2\right)' &= (X\, t\, v) (X'\, t\, v) =
  (X\, t\, v) (f\, (X\,t)\, v) = -(X\, t\, v)g\\
&= -g(f\, (X\,t)\, x) =
  -g(X'\, t\, x) = \left( g (h - X\, t\, x)\right)'.
\end{align*}
By Proposition~\ref{P:invrules}(1), $I_d$ is thus an invariant for $f$
along $\reals^V$.  Proposition \ref{P:inv-prop}(3) and
Lemma~\ref{P:inv-props2} then
imply that
\begin{equation*}
 I_d \le |{x'=f\ \&\ G}]  I_d.
\end{equation*}

\item It remains to show that $I\le I_d$ and $|{x'=f\ \&\ G}] I_d \le
  |{x'=f\ \&\ G}]I$. 
\begin{itemize}
\item The first inequality is trivial. 
\item For the second one, we calculate
\begin{equation*}
  (G \cdot I_d)\, s = \left(0\le s_x \land \frac{1}{2}s_v^2= g(h - s_x)\right) = I\, s.
\end{equation*}
By Corollary~\ref{P:wlpprop-gen2},  therefore, 
\begin{equation*}
|{x'=f\ \&\ G}] I_d = |{x'=f\ \&\ G}](G\cdot I_d) = |{x'=f\ \&\ G}]I.
\end{equation*}
\end{itemize}
\end{enumerate}
This shows that $I\le |{x'=f\, \&\, G}]I$. The remaining proof of
$P\leq|\mathsf{Ball}]Q$ is the same as in Example~\ref{ex:ball}. 
\qed
\end{example}

This example shows that one can reason about invariants of evolution
commands in a natural mathematical style as it can be found in
textbooks on differential equations~\cite{Arnold,Hirsch09,Teschl12}.
By contrast, $\dL$ relies on syntactic substitution-based reasoning in
the term algebra of differential rings~\cite{Platzer12} to check
invariants, and complex domain-specific inference rules to manipulate
them.  The following section shows that we can derive semantic
variants of most of the $\dL$ inference rules for those who like this
style of reasoning, see~\cite{FosterGMS21} for a complete list.

Next, we briefly specialise our approach to $\dL$-invariants, the
invariants sets used in dynamical systems theory and those in
(semi)group theory. We assume a setting where global flows exist and
indices $U$ can be dropped.

\begin{corollary}\label{P:dl-invset}
  Let $f:S\to S$ be Lipschitz continuous. Then $I:S\to \bools$ is a
  $\dL$-invariant for $x'=f\, \&\, \top$ if and only if $I$ is an
  invariant set for $\flow^f$.
\end{corollary}
\begin{proof}
  It is easy to check that
  $(\forall s\in I.\ I\, s \rightarrow \gamma^\flow\, s \subseteq I)
  \leftrightarrow (\gamma^\flow)^\dagger I \subseteq I$.
  The claim then follows from Proposition~\ref{P:inv-prop}. In the
  Lipschitz continuous case, of course,
  $\Sols\, f\, T\, s=\{\flow^f\}$.
\end{proof}

It remains to point out that the difference between the definition of
invariant sets for dynamical systems and that for (semi)group actions
is merely notational: In group theory, an invariant set $I$ of a
(semi)group action $\flow:T\to S\to S$ satisfies
$T\cdot I\subseteq I$, where
$T\cdot I= \{\flow\, t\, s \mid t\in T\land s\in I\}$. In the presence
of a unit, therefore $T\cdot I = I$. Yet of course
$(\gamma^\flow)^\dagger\, I = \{\flow\, t\, s \mid t\in T\land s\in
I\}$ as well.

At then end of this section we summarise the two main workflows
presented. Both use the standard laws for predicate transformer
algebras for automating verification condition generation with respect
to the structural part of hybrid programs. For straight-line programs,
this requires only equational reasoning and can be dealt with by
Isabelle's simplifiers. The remaining verification conditions for
basic commands---evolution and assignment commands---are generated by
equational reasoning in the concrete semantics of the hybrid program
store.  In fact, only this concrete semantics had to be added to a
standard Isabelle verification component to make our verification
components work.

The verification conditions generated are then at the level of
reasoning with functions over $\reals^n$, and in some cases in linear
algebra~\cite{Munive20}. At this level, by contrast with $\dL$, we do
not require any domain-specific inference rules and can rely on
Isabelle's support for semantic reason about the hybrid dynamics
within its higher-order logic, an approach that has allowed us to
verify a large number of benchmark examples~\cite{MitschMJZWZ20}.  Yet
our approach is versatile enough to derive inference rules in the
style of $\dL$, as the following section shows.


\section{Derivation of  $\dL$ Inference Rules}\label{sec:dL}

As a proof of concept, we derive semantic variants of some axioms and
inference rules of $\dL$, thus proving their soundness with respect to
our semantics. The first one introduces solutions of IVPs with
constant vector fields~\cite{BohrerRVVP17}. It is a trivial instance
of Proposition~\ref{P:wlpprop} with $f = \lambda s.\ c$ for some
$c\in \reals$. Such vector fields are Lipschitz continuous; their
flows are $\flow\, t\, s = s + ct$.  Hence
\begin{equation}
|{x'= (\lambda s.\ c)\ \&\ G}] Q  = \lambda s\in S.\ \forall
t\in T.\ (\forall \tau\le t.\ G\, (s + c\tau))\rightarrow Q\, (s +
ct). \tag{DS}\label{eq:DS}
\end{equation}
For a second $\dL$ inference rule we simply rewrite the $\wlp$ in
Proposition~\ref{P:wlpprop} as a Hoare-style inference rule.
\begin{lemma}\label{P:solve} Let $S\subseteq\reals^V$ and
  $T=\reals$. Let $\flow:T\to S\to S$ be the flow for the Lipschitz
  continuous vector field $f:S\to S$, and
  $G,Q:S\to\bools$. Then
\begin{equation}
\inferrule*
{\forall s\in S.\ P\, s\rightarrow (\forall t\in T.\ 
  (\forall \tau\le t.\ G (\flow_s\, \tau))\rightarrow Q\, (\flow_s\, t))}
{P\leq |{x'=f\, \&\, G}] Q}\label{eq:dSolve}\tag{dSolve}
\end{equation}
\end{lemma}
To apply this rule in our setting, the procedure in Section~\ref{sec:hybrid-store}
must be followed.

Next we derive five semantic counterparts of the $\dL$ axioms 
and inference rules for differential invariants in the setting of
Section~\ref{sec:differential-invariants}. The 
\emph{differential cut} axiom (\ref{eq:DC}) and rule (\ref{eq:dC}),
\emph{differential weakening}, (\ref{eq:DW}) and (\ref{eq:dW}),
and the the \emph{differential induction} rule (\ref{eq:dI}). These rules 
are typically applied backwards as follows: $\mathit{dC}$
introduces an invariant. Its left premise is discharged via
$\mathit{dI}$, Proposition~\ref{P:inv-prop} and logical reasoning,
while its right premise is discharged via $\mathit{dW}$.  Note that
the conclusions of all these rules are semantically equivalent to
Hoare triples. Verification examples using these rules and the $\dL$
approach can be found in our Isabelle components.

\begin{lemma}\label{P:dcut} Let $P,G,I,Q:S\to\bools$, $T\subseteq \reals$ and 
  $f:S\to S$ be a continuous vector field. Then, with $\eta_S$ the unit of the power set monad,
\begin{gather}
|{x'=f\, \&\, G}]  I = \eta_S \to
|{x'= f\, \&\, (\lambda\, s.\,  G\, s \land I\, s)}] Q = |{x'=f\, \&\, G}] Q,\tag{DC}\label{eq:DC}\\
\inferrule*
{P\leq |{x'=f\, \&\, G}]  I\\
P\leq |{x'= f\, \&\, (\lambda\, s.\ G\, s \land I\, s)}] Q}
{P\leq |{x'=f\, \&\, G}] Q}\tag{dC}\label{eq:dC}\\
  |{x'=f\ \&\ G}] (\lambda\, s.\ G\, s\to Q\, s) = |{x'=f\ \&\ G}] Q,\tag{DW}\label{eq:DW}\\
  \inferrule*{G\leq Q}{P\leq |{x'=f\ \&\ G}] Q}\tag{dW}\label{eq:dW}
\end{gather}
Finally, if $I$ is a
  differential invariant for $f$
  along $T$, then 
\begin{equation}
\inferrule*
{P\leq I\qquad I\leq Q}
{P\leq |{x'=f\ \&\ G}] Q}\tag{dI}\label{eq:dI}
\end{equation}
\end{lemma}
Axiom (\ref{eq:DC}) and rule (\ref{eq:dC}) introduce differential 
invariants in guards of evolution commands. Axiom and rule (\ref{eq:DW}) and
(\ref{eq:dW}) summarise the fact that if a guard is strong enough to 
imply a postcondition, then no invariant or solution needs to be found.  
Finally, the differential induction rule follows from 
Proposition~\ref{P:inv-prop}(3), transitivity and isotonicity of boxes.

A differential ghost rule~\cite{PlatzerT18} (dG), and sometimes a
differential effect axiom~\cite{Platzer17} have also been proposed for
reasoning with invariants in $\dL$. Our semantics approach has so far
no need for these~\cite{MitschMJZWZ20}---we do not anticipate any
reason why we should not be able to freely introduce ghost variables
for the continuous dynamics as we have so far done for the discrete
one using Isabelle's higher-order logic---but see~\cite{FosterGMS21}
for a derivation of (dG) within our semantic framework.


\section{Isabelle Components for $\MKA$ and Predicate Transformers}\label{sec:isa-pt}

The entire mathematical development of $\MKA$ in
Section~\ref{sec:KA}-\ref{sec:mka-pt} has been formalised with
Isabelle~\cite{afp:ka,afp:kad}. Verification components for Isabelle
and the relational store model in Section~\ref{sec:discrete-store}
have been developed, too~\cite{GomesS16,afp:vericomp}, using the
shallow embedding approach discussed in Section~\ref{sec:intro} and
\ref{sec:discrete-store}.  Predicate transformers \`a la Back and von
Wright have been formalised previously in Isabelle by
Preoteasa~\cite{Preoteasa11,Preoteasa11a}. Our alternative
formalisation emphasises the quantalic structure of
transformers~\cite{afp:quantales,afp:transem}, as in
Section~\ref{sec:pt-backvwright}, and we have added a third component
based on quantaloids~\cite{afp:transem}. It is based on a
formalisation of the powerset monad~\cite{afp:transem}, as outlined in
Section~\ref{sec:pt-monad}.  Our formalisation is compositional in
that all three approaches to predicate transformers can be combined
with relational and state transformer semantics and different models
of the (hybrid) program store, as shown in Figure~\ref{fig:framework}.

This section summarises the Isabelle components for predicate
transformers and the verification component based on $\MKA$. More
detailed information can be found in the proof documents for these
components~\cite{afp:vericomp,afp:transem}.

The $\MKA$ component is integrated into the Kleene algebra hierarchy
that formalises variants of Kleene algebras~\cite{afp:ka} and modal
Kleene algebras~\cite{afp:kad}, as outlined in Section~\ref{sec:KA}
and~\ref{sec:MKA}.  In these mathematical components, algebras are
formalised as type classes, their models via instantiation and
interpretation statements. For Kleene algebras, many computationally
interesting models have been formalised; for $\MKA$ only the
relational model is present in the Archive of Formal Proofs.  The
state transformer model has been formalised for quantales in a
different component~\cite{afp:transem}.

Instantiation and interpretation statements have several purposes in
Isabelle. They make algebraic facts available in all models, establish
soundness of algebraic hierarchies and ultimately make the axiomatic
approaches consistent with respect to Isabelle's small trustworthy
core. Finally, they unify developments of multiple concrete semantics.

In our $\MKA$-based verification components~\cite{afp:vericomp},
program syntax is absent and semantic illusions of program syntax are
provided in the concrete program semantics, as outlined in
Section~\ref{sec:discrete-store}. Consequently, verification
conditions for the control structure of programs are generated within
the algebra; those for assignments in the concrete store semantics.
We currently model stores simply as functions from strings
representing variables to values of arbitrary type. Expressions are
simulated by functions from stores to values, as outlined in
Section~\ref{sec:discrete-store}; stores with poly-typed values are
modelled via sum-types. An extension to verification components for
hybrid programs is described in the following sections.

A second component is based on predicate transformers \`a la Back and
von Wright~\cite{BackW98}, for which we have built special purpose
components with advanced features for orderings and
lattices~\cite{afp:order} and for quantales~\cite{afp:quantales}.
These structures are once again formalised as type classes.  Predicate
transformers, however, are modelled as global functions that may have
different source and target types. Isabelle's simple type system can
infer most general types for definitions. These can be associated with
predicate transformers by sort constraints; definitions can often be
declared in the point-free style of functional programming. This makes
the formalisation of quantaloids of transformers with partial
compositions straightforward. Mono-typed transformer algebras are
obtained from these via subtyping. They are linked with quantales and
Kleene algebras by interpretation or instantiation.

Isabelle's type system is too weak for a deep embedding of general
categorical concepts, but formalising instances such as the powerset
monad, its Kleisli category and Eilenberg-Moore algebras is
straightforward. We have formalised the isomorphisms and dualities
between relations, state transformers and the four predicate
transformers corresponding to backward and forward boxes and diamonds
in this setting. Using these dualities to transport theorems
automatically requires Isabelle's transfer package, which is ongoing
work.

We have created a second verification component for hybrid systems
based on Back and von Wright's approach, using the monadic
transformers to obtain a concrete semantics.  Finally, we have once again
restricted the categorical approach to the mono-typed case in a third
component. Via subtyping we can then show that the categorical
transformers form quantales, and more specifically $\MKA$s.
Everything Isabelle knows about $\MKA$ is then available in this
instance.


\section{Isabelle Components for ODEs and Orbits}\label{sec:isa-ODE}

This section and the two following ones describe the formalisation of
the material in Sections~\ref{sec:ODE}-\ref{sec:dL} in Isabelle, from
mathematical components for ODEs and orbits to verification components
for hybrid programs based on (local) flows, differential invariants
and $\dL$-style inference rules.

We begin with summarising Immler and H\"olzl's formalisation of the
Picard-Lindel\"of theorem based on the Isabelle hierarchy for analysis
and ordinary differential
equations~\cite{HolzlIH13,Immler12,ImmlerH12a,ImmlerT19}. We have
adapted their results to show that unique solutions to IVPs for
autonomous systems of ODEs guaranteed by this theorem satisfy the
local flow conditions, as discussed in previous sections.

H\"olzl and Immler have proved the Picard-Lindel\"of theorem for
time-dependent vector fields of type
${\isachardoublequoteopen}\mathit{real}\ {\isasymRightarrow}\
{\isacharparenleft}{\isacharprime}a{\isacharcolon}{\isacharcolon}{\isacharbraceleft}heine{\isacharunderscore}borel{\isacharcomma
}banach{\isacharbraceright}{\isacharparenright}\ {\isasymRightarrow}\
{\isacharprime}a{\isachardoublequoteclose}$~\cite{ImmlerH12a}.  They
have called their theorem
$\mathit{unique}{\isacharunderscore}\mathit{solution}$ and have
formalised it within a locale called
$\isa{unique{\isacharunderscore}on{\isacharunderscore
  }bounded{\isacharunderscore}closed}$ to bundle the assumptions for
the local existence of unique solutions within a closed interval in
$\reals$. They have specialised and hence extended this locale in
various ways.

Our approach builds on top of their extension
$\isa{ll{\isacharunderscore}on{\isacharunderscore
  }open{\isacharunderscore}it}$ that bundles more or less the
conditions of Theorem~\ref{P:picard-lindeloef}, but for the
time-dependent case. In our formalisation, we add the condition
$t_0\in T$ to have this parameter available in the following
developments. Thus, we have generated the following variant.

\begin{isabellebody}
\isanewline
\isacommand{locale}\isamarkupfalse%
\ picard{\isacharunderscore}lindeloef\ {\isacharequal}\isanewline
\ \ \isakeyword{fixes}\
f{\isacharcolon}{\isacharcolon}{\isachardoublequoteopen}real\
{\isasymRightarrow}\
{\isacharparenleft}{\isacharprime}a{\isacharcolon}{\isacharcolon}{\isacharbraceleft}heine{\isacharunderscore}borel{\isacharcomma
}banach{\isacharbraceright}{\isacharparenright}\ {\isasymRightarrow}\
{\isacharprime}a{\isachardoublequoteclose}\isanewline
\ \ \ \  \isakeyword{and}\
T{\isacharcolon}{\isacharcolon}{\isachardoublequoteopen}real\
set{\isachardoublequoteclose}\isanewline
\ \ \ \ \isakeyword{and}\
S{\isacharcolon}{\isacharcolon}{\isachardoublequoteopen}{\isacharprime}a\
set{\isachardoublequoteclose}\isanewline
\ \ \ \ \isakeyword{and}\ t\isactrlsub {\isadigit{0}}{\isacharcolon}{\isacharcolon}real\isanewline
\ \ \isakeyword{assumes}\ open{\isacharunderscore}domain{\isacharcolon}\ {\isachardoublequoteopen}open\ T{\isachardoublequoteclose}\ {\isachardoublequoteopen}open\ S{\isachardoublequoteclose}\isanewline 
\ \ \ \ \isakeyword{and}\ interval{\isacharunderscore}time{\isacharcolon}\ {\isachardoublequoteopen}is{\isacharunderscore}interval\ T{\isachardoublequoteclose}\isanewline
\ \ \ \ \isakeyword{and}\ init{\isacharunderscore}time{\isacharcolon}\ {\isachardoublequoteopen}t\isactrlsub {\isadigit{0}}\ {\isasymin}\ T{\isachardoublequoteclose}\isanewline
\ \ \ \ \isakeyword{and}\ cont{\isacharunderscore}vec{\isacharunderscore}field{\isacharcolon}\ {\isachardoublequoteopen}{\isasymforall}s\ {\isasymin}\ S{\isachardot}\ continuous{\isacharunderscore}on\ T\ {\isacharparenleft}{\isasymlambda}t{\isachardot}\ f\ t\ s{\isacharparenright}{\isachardoublequoteclose}\isanewline
\ \ \ \ \isakeyword{and}\ lipschitz{\isacharunderscore}vec{\isacharunderscore}field{\isacharcolon}\ {\isachardoublequoteopen}local{\isacharunderscore}lipschitz\ T\ S\ f{\isachardoublequoteclose}\isanewline
\isakeyword{begin}
\isanewline
\isanewline
\isacommand{sublocale}\isamarkupfalse%
\ ll{\isacharunderscore}on{\isacharunderscore}open{\isacharunderscore}it\ T\ f\ S\ t\isactrlsub {\isadigit{0}}
\isanewline
\ \ $\langle \isa{proof}\rangle$
\isanewline

\isacommand{lemma}\isamarkupfalse%
\ unique{\isacharunderscore}solution{\isacharcolon}\isanewline
\ \ \isakeyword{assumes}\ xivp{\isacharcolon}\
{\isachardoublequoteopen}D\ X\ {\isacharequal}\
{\isacharparenleft}{\isasymlambda}t{\isachardot}\ f\ t\
{\isacharparenleft}X\ t{\isacharparenright}{\isacharparenright}\ on\
{\isacharbraceleft}t\isactrlsub
{\isadigit{0}}{\isacharminus}{\isacharminus}t{\isacharbraceright}{\isachardoublequoteclose}\
{\isachardoublequoteopen}X\ t\isactrlsub {\isadigit{0}}\
{\isacharequal}\ s{\isachardoublequoteclose}\
{\isachardoublequoteopen}X\ {\isasymin}\
{\isacharbraceleft}t\isactrlsub
{\isadigit{0}}{\isacharminus}{\isacharminus}t{\isacharbraceright}\
{\isasymrightarrow}\ S{\isachardoublequoteclose}\isanewline
\ \ \ \  \isakeyword{and}\ {\isachardoublequoteopen}t\ {\isasymin}\ T{\isachardoublequoteclose}\isanewline
\ \ \ \ \isakeyword{and}\ yivp{\isacharcolon}\
{\isachardoublequoteopen}D\ Y\ {\isacharequal}\
{\isacharparenleft}{\isasymlambda}t{\isachardot}\ f\ t\
{\isacharparenleft}Y\ t{\isacharparenright}{\isacharparenright}\ on\
{\isacharbraceleft}t\isactrlsub
{\isadigit{0}}{\isacharminus}{\isacharminus}t{\isacharbraceright}{\isachardoublequoteclose}\
{\isachardoublequoteopen}Y\ t\isactrlsub {\isadigit{0}}\
{\isacharequal}\ s{\isachardoublequoteclose}\
{\isachardoublequoteopen}Y\ {\isasymin}\
{\isacharbraceleft}t\isactrlsub
{\isadigit{0}}{\isacharminus}{\isacharminus}t{\isacharbraceright}\
{\isasymrightarrow}\ S{\isachardoublequoteclose}\isanewline
\ \ \ \  \isakeyword{and}\ {\isachardoublequoteopen}s\ {\isasymin}\ S{\isachardoublequoteclose}\ \isanewline
\ \ \isakeyword{shows}\ {\isachardoublequoteopen}X\ t\ {\isacharequal}\ Y\ t{\isachardoublequoteclose}
\isanewline
\ \ $\langle \isa{proof}\rangle$
\isanewline

\isacommand{end}
\isanewline
\end{isabellebody}

\noindent The locale declaration lists the assumptions of the
Picard-Lindel\"of theorem: the vector field $f$---which is still
time-dependent---is defined on an open time interval $T$ that contains
the initial time $t_0$, and an open subset $S$ of the state space. The
vector field $f$ is continuous in time and, for each
$(t,s)\in T\times S$, Lipschitz continuous on a closed subset of
$T\times S$ around $(t,s)$. The sublocale statement shows that these
assumptions imply those of the locale
$\isa{ll{\isacharunderscore}on{\isacharunderscore
  }open{\isacharunderscore}it}$.
Lemma $\mathit{unique}{\isacharunderscore}\mathit{solution}$ ensures
that the Picard-Lindel\"of theorem is derivable within this locale. The
notation $D\, X$ stands for $X'$, and $g\in A\to B$ indicates that
function $g$ maps from the set $A$ into the set $B$, as opposed to the
type of $g$, which can be larger. The notation
${\isacharbraceleft}t\isactrlsub
{\isadigit{0}}{\isacharminus}{\isacharminus}t{\isacharbraceright}$
indicates the set of real numbers between
$t\isactrlsub {\isadigit{0}}$ and $t$ (including both), where $t$ may
be above or below $t_0$.  The formalisation of the Picard-Lindel\"of
theorem comprises a formal definition of solutions to IVPs of system
of ODEs in Isabelle. As an abbreviation, we have defined the set
$\Sols\, f\, T\, s$ of Section~\ref{sec:generalisation} with the
additional requirement that $X\in T\to S$.

\begin{isabellebody}
\isanewline
\isacommand{definition}\isamarkupfalse%
\ ivp{\isacharunderscore}sols\ {\isacharcolon}{\isacharcolon}\ {\isachardoublequoteopen}{\isacharparenleft}real\ {\isasymRightarrow}\ {\isacharprime}a\ {\isasymRightarrow}\ {\isacharparenleft}{\isacharprime}a {\isacharcolon}{\isacharcolon}\ real{\isacharunderscore}normed{\isacharunderscore }vector{\isacharparenright}{\isacharparenright}\ {\isasymRightarrow}\ real\  set\ {\isasymRightarrow}\  {\isacharprime}a\ set\ {\isasymRightarrow} \isanewline
\ \ real\ {\isasymRightarrow}\ {\isacharprime}a\ {\isasymRightarrow}\ {\isacharparenleft}real\ {\isasymRightarrow}\ {\isacharprime}a{\isacharparenright}\ set {\isachardoublequoteclose}\ {\isacharparenleft}{\isachardoublequoteopen}Sols{\isachardoublequoteclose}{\isacharparenright}\isanewline
\ \ \isakeyword{where}\ {\isachardoublequoteopen}Sols\ f\ T\ S\ t\isactrlsub {\isadigit{0}}\ s\ {\isacharequal}\ {\isacharbraceleft}X\ {\isacharbar}X{\isachardot}\ {\isacharparenleft}D\ X\ {\isacharequal}\ {\isacharparenleft}{\isasymlambda}t{\isachardot}\ f\ t\ {\isacharparenleft}X\ t{\isacharparenright}{\isacharparenright}\ on\ T{\isacharparenright}\ {\isasymand}\ X\ t\isactrlsub {\isadigit{0}}\ {\isacharequal}\ s\ {\isasymand}\ X\ {\isasymin}\ T\ {\isasymrightarrow}\ S{\isacharbraceright}{\isachardoublequoteclose}\isanewline
\end{isabellebody}

\noindent We restrict locale $\isa{picard-lindeloef}$ to
autonomous systems and to $t_0=0$, while introducing the variable
$\flow$ for the local flow of the vector field. In support of our 
open approach to hybrid program verification, this allows users to 
supply any characterisation of the flow that suits them best, as a 
successor paper illustrates~\cite{Munive20}.

\begin{isabellebody}
\isanewline\isacommand{locale}\isamarkupfalse%
\ local{\isacharunderscore}flow\ {\isacharequal}\ picard{\isacharunderscore}lindeloef\ {\isachardoublequoteopen}{\isacharparenleft}{\isasymlambda}\ t{\isachardot}\ f{\isacharparenright}{\isachardoublequoteclose}\ T\ S\ {\isadigit{0}}\ \isanewline
\ \ \isakeyword{for}\
f{\isacharcolon}{\isacharcolon}{\isachardoublequoteopen}{\isacharprime}a{\isacharcolon}{\isacharcolon}{\isacharbraceleft}heine{\isacharunderscore}borel{\isacharcomma
}banach{\isacharbraceright}\ {\isasymRightarrow}\
{\isacharprime}a{\isachardoublequoteclose}\isanewline
\ \ \ \ \isakeyword{and}\ T\ S\ L\ {\isacharplus}\isanewline
\ \ \isakeyword{fixes}\ {\isasymphi}\ {\isacharcolon}{\isacharcolon}\ {\isachardoublequoteopen}real\ {\isasymRightarrow}\ {\isacharprime}a\ {\isasymRightarrow}\ {\isacharprime}a{\isachardoublequoteclose}\isanewline
\ \ \isakeyword{assumes}\ ivp{\isacharcolon}\isanewline
\ \ \ \ {\isachardoublequoteopen}{\isasymAnd}\ t\ s{\isachardot}\ t\ {\isasymin}\ T\ {\isasymLongrightarrow}\ s\ {\isasymin}\ S\ {\isasymLongrightarrow}\ D\ {\isacharparenleft}{\isasymlambda}t{\isachardot}\ {\isasymphi}\ t\ s{\isacharparenright}\ {\isacharequal}\ {\isacharparenleft}{\isasymlambda}t{\isachardot}\ f\ {\isacharparenleft}{\isasymphi}\ t\ s{\isacharparenright}{\isacharparenright}\ on\ {\isacharbraceleft}{\isadigit{0}}{\isacharminus}{\isacharminus}t{\isacharbraceright}{\isachardoublequoteclose}\isanewline
\ \ \ \  {\isachardoublequoteopen}{\isasymAnd}\ s{\isachardot}\ s\ {\isasymin}\ S\ {\isasymLongrightarrow}\ {\isasymphi}\ {\isadigit{0}}\ s\ {\isacharequal}\ s{\isachardoublequoteclose}\isanewline
\ \ \ \  {\isachardoublequoteopen}{\isasymAnd}\ t\ s{\isachardot}\ t\ {\isasymin}\ T\ {\isasymLongrightarrow}\ s\ {\isasymin}\ S\ {\isasymLongrightarrow}\ {\isacharparenleft}{\isasymlambda}t{\isachardot}\ {\isasymphi}\ t\ s{\isacharparenright}\ {\isasymin}\ {\isacharbraceleft}{\isadigit{0}}{\isacharminus}{\isacharminus}t{\isacharbraceright}\ {\isasymrightarrow}\ S{\isachardoublequoteclose}\isanewline
\end{isabellebody}

\noindent The assumptions \isa{ivp} force $T$ to coincide with its largest subinterval (\isa{ex{\isacharunderscore}ivl}) where solutions exist (lemma \isa{ex{\isacharunderscore}ivl{\isacharunderscore}eq} below). Thus, $\flow$ is the unique solution on the whole of $T$
---and not only on its subsets ${\isacharbraceleft}{\isadigit{0}}{\isacharminus}{\isacharminus}t{\isacharbraceright}$ unlike $\isa{picard{\isacharunderscore}lindeloef}$ or 
$\isa{ll{\isacharunderscore}on{\isacharunderscore}open{\isacharunderscore}it}$. This 
allows users of the locale to choose $T$ as small as they wish.

\begin{isabellebody}
\isanewline
\isacommand{lemma}\isamarkupfalse%
\ ex{\isacharunderscore}ivl{\isacharunderscore}eq{\isacharcolon}\
 {\isachardoublequoteopen}s\ {\isasymin}\ S{\isachardoublequoteclose}
\ {\isasymLongrightarrow}\ 
{\isachardoublequoteopen}ex{\isacharunderscore}ivl\ s\ {\isacharequal}\ T{\isachardoublequoteclose}\isanewline
\ \ $\langle \isa{proof}\rangle$
\isanewline

\isacommand{lemma}\isamarkupfalse%
\ has{\isacharunderscore}vderiv{\isacharunderscore}on{\isacharunderscore }domain{\isacharcolon }\
 {\isachardoublequoteopen}s\ {\isasymin}\ S{\isachardoublequoteclose}
\ {\isasymLongrightarrow}\ 
 {\isachardoublequoteopen}D\ {\isacharparenleft}{\isasymlambda}t{\isachardot}\ {\isasymphi}\ t\ s{\isacharparenright}\ {\isacharequal}\ {\isacharparenleft}{\isasymlambda}t{\isachardot}\ f\ {\isacharparenleft}{\isasymphi}\ t\ s{\isacharparenright}{\isacharparenright}\ on\ T{\isachardoublequoteclose}\isanewline
\ \ $\langle \isa{proof}\rangle$
\isanewline

\isacommand{lemma}\isamarkupfalse%
\ in{\isacharunderscore}ivp{\isacharunderscore}sols{\isacharcolon}\ 
{\isachardoublequoteopen}s\ {\isasymin}\ S{\isachardoublequoteclose}
\ {\isasymLongrightarrow}\ 
 {\isachardoublequoteopen}{\isacharparenleft}{\isasymlambda}t{\isachardot}\ {\isasymphi}\ t\ s{\isacharparenright}\ {\isasymin}\ Sols\ {\isacharparenleft}{\isasymlambda}t{\isachardot}\ f{\isacharparenright}\ T\ S\ {\isadigit{0}}\ s{\isachardoublequoteclose}\isanewline
\ \ $\langle \isa{proof}\rangle$
\isanewline

\isacommand{lemma}\isamarkupfalse%
\ eq{\isacharunderscore}solution{\isacharcolon}\
 {\isachardoublequoteopen}X\ {\isasymin}\
Sols\
{\isacharparenleft}{\isasymlambda}t{\isachardot}\
f{\isacharparenright}\ T\ S\ {\isadigit{0}}\
s{\isachardoublequoteclose}
 {\isasymLongrightarrow}\ 
 {\isachardoublequoteopen}t\ {\isasymin}\
T{\isachardoublequoteclose}
\ {\isasymLongrightarrow}\ 
 {\isachardoublequoteopen}s\ {\isasymin}\ S{\isachardoublequoteclose}
\ {\isasymLongrightarrow}\ 
 {\isachardoublequoteopen}X\ t\ {\isacharequal}\ {\isasymphi}\ t\ s{\isachardoublequoteclose}\isanewline
\ \ $\langle \isa{proof}\rangle$\isanewline
\end{isabellebody}

\noindent Finally, in this locale we can prove that if the maximal
interval of existence $T$ equals $\reals$, then the flow $\flow$ is
global and hence a proper monoid action.

\begin{isabellebody}
\isanewline
\isacommand{lemma}\isamarkupfalse%
\
ivp{\isacharunderscore}sols{\isacharunderscore}collapse{\isacharcolon}\ 
 {\isachardoublequoteopen}T\ {\isacharequal}\
 UNIV{\isachardoublequoteclose} 
 {\isasymLongrightarrow}\ 
 s\ {\isasymin}\ {\isachardoublequoteopen}S{\isachardoublequoteclose}\ 
 {\isasymLongrightarrow}\ 
{\isachardoublequoteopen}Sols\ {\isacharparenleft}{\isasymlambda}t{\isachardot}\ f{\isacharparenright}\ T\ S\ {\isadigit{0}}\ s\ {\isacharequal}\ {\isacharbraceleft}{\isacharparenleft}{\isasymlambda}t{\isachardot}\ {\isasymphi}\ t\ s{\isacharparenright}{\isacharbraceright}{\isachardoublequoteclose}\isanewline
\ \ $\langle \isa{proof}\rangle$
\isanewline

\isacommand{lemma}\isamarkupfalse%
\ is{\isacharunderscore}monoid{\isacharunderscore }action{\isacharcolon}\isanewline
\ \ \isakeyword{assumes}\ {\isachardoublequoteopen}s\ {\isasymin}\
S{\isachardoublequoteclose}\isanewline
\ \ \ \  \isakeyword{and}\ {\isachardoublequoteopen}T\ {\isacharequal}\ UNIV{\isachardoublequoteclose}\isanewline
\ \ \isakeyword{shows}\ {\isachardoublequoteopen}{\isasymphi}\
{\isadigit{0}}\ s\ {\isacharequal}\
s{\isachardoublequoteclose}\isanewline
\ \ \ \  \isakeyword{and}\ {\isachardoublequoteopen}{\isasymphi}\ {\isacharparenleft}t\isactrlsub {\isadigit{1}}\ {\isacharplus}\ t\isactrlsub {\isadigit{2}}{\isacharparenright}\ s\ {\isacharequal}\ {\isasymphi}\ t\isactrlsub {\isadigit{1}}\ {\isacharparenleft}{\isasymphi}\ t\isactrlsub {\isadigit{2}}\ s{\isacharparenright}{\isachardoublequoteclose}\isanewline
\ \ $\langle \isa{proof}\rangle$\isanewline
\end{isabellebody}

\noindent We have not generated a locale for this case, as
the assumptions needed remain unchanged.  Locale
$\mathit{picard}{\isacharunderscore}\mathit{lindeloef}$ thus
guarantees the existence of unique solutions for IVPs of
time-dependent systems. Locale
$\mathit{local}{\isacharunderscore}\mathit{flow}$ specialises it to
autonomous systems with Lipschitz continuous vector fields and local
flows. It covers dynamical systems with global flows and thus the
verification of hybrid systems. This provides the basic Isabelle
infrastructure for formalising the concrete semantics for hybrid
systems with Lipschitz continuous vector fields from
Figure~\ref{fig:framework}.

Next we describe our formalisation of the orbits and orbitals from
Section~\ref{sec:generalisation}. These form the basis for our
verification components for continuous vector fields beyond the scope
of  the Picard-Lindel\"of  theorem, as shown in
Figure~\ref{fig:framework}. Yet we can instantiate all concepts to
settings where (local) flows exist. First, we have formalised the
$G$-guarded orbit $\gamma^X_{G}$ of $X$ along $T$, with
$\isa{down\ T\ t}$ standing for ${\downarrow}t$.

\begin{isabellebody}
\isanewline
\isacommand{definition}\isamarkupfalse%
\ g{\isacharunderscore}orbit\ {\isacharcolon}{\isacharcolon}\ {\isachardoublequoteopen}{\isacharparenleft}real\ {\isasymRightarrow}\ {\isacharprime}a{\isacharparenright}\ {\isasymRightarrow}\ {\isacharparenleft}{\isacharprime}a\ {\isasymRightarrow}\ bool{\isacharparenright}\ {\isasymRightarrow}\ real\ set\ {\isasymRightarrow}\ {\isacharprime}a\ set{\isachardoublequoteclose}\ {\isacharparenleft}{\isachardoublequoteopen}{\isasymgamma}{\isachardoublequoteclose}{\isacharparenright}\isanewline
\ \ \isakeyword{where}\ {\isachardoublequoteopen}{\isasymgamma}\ X\ G\ T\ {\isacharequal}\ {\isasymUnion}{\isacharbraceleft}{\isasymP}\ X\ {\isacharparenleft}down\ T\ t{\isacharparenright}\ {\isacharbar}t{\isachardot}\ {\isasymP}\ X\ {\isacharparenleft}down\ T\ t{\isacharparenright}\ {\isasymsubseteq}\ {\isacharbraceleft}s{\isachardot}\ G\ s{\isacharbraceright}{\isacharbraceright}{\isachardoublequoteclose}\isanewline
\isanewline
\isacommand{lemma}\isamarkupfalse%
\ g{\isacharunderscore}orbit{\isacharunderscore}eq{\isacharcolon}\ {\isachardoublequoteopen}{\isasymgamma}\ X\ G\ T\ {\isacharequal}\ {\isacharbraceleft}X\ t\ {\isacharbar}t{\isachardot}\ t\ {\isasymin}\ T\ {\isasymand}\ {\isacharparenleft}{\isasymforall}{\isasymtau}{\isasymin}down\ T\ t{\isachardot}\ G\ {\isacharparenleft}X\ {\isasymtau}{\isacharparenright}{\isacharparenright}{\isacharbraceright}{\isachardoublequoteclose}\isanewline
\ \ $\langle \isa{proof}\rangle$\isanewline
\end{isabellebody}

\noindent We have also formalised the $G$-guarded orbital of $f$ along
$T$ in $s$ (as $\gamma^f_G\, s$) together with Lemma~\ref{P:gorbital}.

\begin{isabellebody}
\isanewline
\isacommand{definition}\isamarkupfalse%
\ g{\isacharunderscore}orbital\ {\isacharcolon}{\isacharcolon}\ {\isachardoublequoteopen}{\isacharparenleft}{\isacharprime}a\ {\isasymRightarrow}\ {\isacharprime}a{\isacharparenright}\ {\isasymRightarrow}\ {\isacharparenleft}{\isacharprime}a\ {\isasymRightarrow}\ bool{\isacharparenright}\ {\isasymRightarrow}\ real\ set\ {\isasymRightarrow}\ {\isacharprime}a\ set\ {\isasymRightarrow}\ real\ {\isasymRightarrow}\ \isanewline
\ \ {\isacharparenleft}{\isacharprime}a{\isacharcolon}{\isacharcolon}real{\isacharunderscore}normed{\isacharunderscore }vector{\isacharparenright}\ {\isasymRightarrow}\ {\isacharprime}a\ set{\isachardoublequoteclose}\ \isanewline
\ \ \isakeyword{where}\ {\isachardoublequoteopen}g{\isacharunderscore}orbital\ f\ G\ T\ S\ t\isactrlsub {\isadigit{0}}\ s\ {\isacharequal}\ {\isasymUnion}{\isacharbraceleft}{\isasymgamma}\ X\ G\ T\ {\isacharbar}X{\isachardot}\ X\ {\isasymin}\ Sols\ {\isacharparenleft}{\isasymlambda}t{\isachardot}\ f{\isacharparenright}\ T\ S\ t\isactrlsub {\isadigit{0}}\ s{\isacharbraceright}{\isachardoublequoteclose}\isanewline
\isanewline
\isacommand{lemma}\isamarkupfalse%
\ g{\isacharunderscore}orbital{\isacharunderscore}eq{\isacharcolon}\ {\isachardoublequoteopen}g{\isacharunderscore}orbital\ f\ G\ T\ S\ t\isactrlsub {\isadigit{0}}\ s\ {\isacharequal}\ \isanewline
\ \ {\isacharbraceleft}X\ t\ {\isacharbar}t\ X{\isachardot}\ t\ {\isasymin}\ T\ {\isasymand}\ {\isasymP}\ X\ {\isacharparenleft}down\ T\ t{\isacharparenright}\ {\isasymsubseteq}\ {\isacharbraceleft}s{\isachardot}\ G\ s{\isacharbraceright}\ {\isasymand}\ X\ {\isasymin}\ Sols\ {\isacharparenleft}{\isasymlambda}t{\isachardot}\ f{\isacharparenright}\ T\ S\ t\isactrlsub {\isadigit{0}}\ s\ {\isacharbraceright}{\isachardoublequoteclose}\ \isanewline
\ \ $\langle \isa{proof}\rangle$\isanewline
\end{isabellebody}

\noindent We have shown that their counterparts from dynamical 
systems are special cases by instantiating our definitions to the 
parameters of the locale $\isa{local{\isacharunderscore}flow}$.
Hence, the $\top$-guarded orbital of $f$ along $T$ in $s$ 
becomes the standard orbit of $s$, and its $G$-guarded version 
is the set in Lemma~\ref{P:g-orbit-props}.

\begin{isabellebody}
\isanewline
\isacommand{context}\isamarkupfalse%
\ local{\isacharunderscore}flow\isanewline
\isakeyword{begin}
\isanewline
\isanewline
\isacommand{definition}\isamarkupfalse\ orbit {\isacharcolon}{\isacharcolon} {\isachardoublequoteopen}{\isacharprime}a\ {\isasymRightarrow}\ {\isacharprime}a\ set{\isachardoublequoteclose}\ 
{\isacharparenleft}{\isachardoublequoteopen}{\isasymgamma}\isactrlsup {\isasymphi}{\isachardoublequoteclose}{\isacharparenright}\isanewline
\ \ \isakeyword{where}\ {\isachardoublequoteopen}{\isasymgamma}\isactrlsup {\isasymphi}\ s\ {\isacharequal}\ g{\isacharunderscore}orbital\ f\ {\isacharparenleft}{\isasymlambda}s{\isachardot}\ True{\isacharparenright}\ T\ S\ {\isadigit{0}}\ s{\isachardoublequoteclose}\isanewline
\isanewline
\isacommand{lemma}\isamarkupfalse%
\
orbit{\isacharunderscore}eq{\isacharbrackleft}simp{\isacharbrackright}{\isacharcolon}\
{\isachardoublequoteopen}s\ {\isasymin}\ S{\isachardoublequoteclose}
 {\isasymLongrightarrow}\ 
{\isachardoublequoteopen}{\isasymgamma}\isactrlsup {\isasymphi}\ s\
{\isacharequal}\ {\isacharbraceleft}{\isasymphi}\ t\ s {\isacharbar}t{\isachardot}\ t\ {\isasymin}\ T{\isacharbraceright}{\isachardoublequoteclose}\isanewline
\ \ $\langle \isa{proof}\rangle$\isanewline

\isacommand{lemma}\isamarkupfalse%
\ g{\isacharunderscore}orbital{\isacharunderscore}collapses{\isacharcolon}\ \isanewline
\ \ {\isachardoublequoteopen}s\ {\isasymin}\
S{\isachardoublequoteclose}
 {\isasymLongrightarrow}\ 
 {\isachardoublequoteopen}g{\isacharunderscore}orbital\ f\ G\ T\ S\
 {\isadigit{0}}\ s\ {\isacharequal}\ {\isacharbraceleft}{\isasymphi}\
 t\ s {\isacharbar}t{\isachardot}\ t\ {\isasymin}\ T\ {\isasymand}\ {\isacharparenleft}{\isasymforall}{\isasymtau}{\isasymin}down\ T\ t{\isachardot}\ G\ {\isacharparenleft}{\isasymphi}\ {\isasymtau}\ s{\isacharparenright}{\isacharparenright}{\isacharbraceright}{\isachardoublequoteclose}\isanewline
\ \ $\langle \isa{proof}\rangle$\isanewline

\isakeyword{end}
\isanewline
\end{isabellebody}

Overall, the set-theoretic concepts introduced in
Section~\ref{sec:generalisation} are easily definable in
Isabelle. Similarly, lemmas formalising their properties and relating
them are often proved automatically in one or two lines. Analytical
properties like the existence of derivatives in a region of space or
the uniqueness of solutions for IVPs are harder to prove. Such lemmas
often require long structured proofs with proofs by cases and explicit
calculations, that is, a considerable amount of user interaction. Yet
most proofs remain at least roughly at the level of textbook
reasoning.


\section{Isabelle Components for Hybrid Programs}\label{sec:isa-wlp}

This section describes the integration of the state transformer and
relational semantics for dynamical systems and Lipschitz-continuous
vector fields from Section~\ref{sec:hybrid-store} and the continuous
vector fields from Section~\ref{sec:generalisation} into the three
verification components for predicate transformers outlined in
Section~\ref{sec:isa-pt} and Figure~\ref{fig:framework}. This requires
formalising hybrid stores and the semantics of evolution commands for
dynamical systems, Lipschitz continuous vector fields with local flows
and continuous vector fields. As explained in
Section~\ref{sec:hybrid-store} and \ref{sec:differential-invariants},
this supports two different workflows using the procedures introduced
in these sections: the first one is for reasoning with (local) flows
and orbits, the second, more general one, for reasoning with
invariants.

First we explain our formalisation of the hybrid store type
$\reals^V$.  We use Isabelle's type
${\isacharparenleft}\mathit{real}{\isacharcomma}{\isacharprime}n{\isacharparenright}\
\mathit{vec}$ (abbreviated as
$\mathit{real}{\isacharcircum}{\isacharprime}n$) of real valued
vectors of dimension $n$, formalised as the type
${\isacharprime}n\ {\isasymRightarrow}\ \mathit{real}$ of functions
from the finite type ${\isacharprime}n$ into $\reals$. This represents
hybrid stores in $\reals^V$ with $|V|=n$.  Isabelle uses the notation
$s{\isachardollar}i$ for the $i$th coordinate of a vector $s$ and
hence the value of store $s$ at variable $i$. More mathematically,
${\isachardollar}$ is the bijection from
$\mathit{real}{\isacharcircum}{\isacharprime}n$ to
${\isacharprime}n\ {\isasymRightarrow}\ \mathit{real}$. Its inverse is
written using a binder ${\isasymchi}$ that replaces
$\lambda$-abstraction. Thus
$({\isasymchi} i{\isachardot}\ s){\isachardollar}i = s$ for any
$s{\isacharcolon}{\isacharcolon}\mathit{real}{\isacharcircum}{\isacharprime}n$
and $({\isasymchi} i{\isachardot}\ x){\isachardollar}i = x$ for any
$x{\isacharcolon}{\isacharcolon}\mathit{real}$. As a consequence of
this simple approach, variables are formalised as natural
numbers. More general namespaces have been included in our framework
more recently~\cite{FosterMS20,FosterGMS21} to make it more user
friendly.

Our state transformer semantics uses functions of type
$\mathit{real}{\isacharcircum}{\isacharprime}n\ {\isasymRightarrow}\
{\isacharparenleft}\mathit{real}{\isacharcircum}{\isacharprime}n{\isacharparenright}\
\mathit{set}$,
which we abbreviate as
${\isacharparenleft}\mathit{real}{\isacharcircum}{\isacharprime}n{\isacharparenright}\
\mathit{nd}{\isacharunderscore}\mathit{fun}$
(for non-deterministic functions). These are instances of the more
general type
${\isacharprime}a\ \mathit{nd}{\isacharunderscore}\mathit{fun}$ of
nondeterministic endofunctions. 

Alternatively, we use relations of type
${\isacharparenleft}\mathit{real}{\isacharcircum}{\isacharprime}n{\isacharparenright}\
\mathit{rel}$, which are instances of
${\isacharprime}a\ \mathit{rel}$.  For both intermediate semantics we
have shown with Isabelle that they form $\MKA$s, but we have also
integrated them into the two quantalic predicate transformer semantics
in Figure~\ref{fig:framework}.
\begin{isabellebody}
\isanewline
\isacommand{interpretation}\ rel{\isacharunderscore}aka{\isacharcolon}\ antidomain{\isacharunderscore}kleene{\isacharunderscore}algebra
\  Id\ {\isacharbraceleft}{\isacharbraceright}\ {\isacharparenleft}{\isasymunion}{\isacharparenright}\ {\isacharparenleft}{\isacharsemicolon}{\isacharparenright}\ {\isacharparenleft}{\isasymsubseteq}{\isacharparenright}\ {\isacharparenleft}{\isasymsubset}{\isacharparenright}\ rtrancl\ rel{\isacharunderscore}ad\isanewline
\ \ $\langle \isa{proof}\rangle$\isanewline

\isacommand{instantiation}\isamarkupfalse%
\ nd{\isacharunderscore}fun\ {\isacharcolon}{\isacharcolon}\ {\isacharparenleft}type{\isacharparenright}\ antidomain{\isacharunderscore}kleene{\isacharunderscore }algebra\isanewline
\ \ $\langle \isa{proof}\rangle$\isanewline
\end{isabellebody}
\noindent After these proofs, all statements proved in Isabelle's $\MKA$
components are available for state transformers and relations.  We
have formalised $\wlp$s for both models, where
${\isasymlceil}-{\isasymrceil}$ ambiguously denotes the isomorphism
between predicates and binary relations or nondeterministic functions.

\begin{isabellebody}
\isanewline
\isacommand{lemma}\isamarkupfalse%
\ wp{\isacharunderscore}rel{\isacharcolon}{\isachardoublequoteopen }\ wp\ R\ {\isasymlceil}P{\isasymrceil}\ {\isacharequal}\ {\isasymlceil}{\isasymlambda}\ x{\isachardot}\ {\isasymforall}\ y{\isachardot}\ {\isacharparenleft}x{\isacharcomma}y{\isacharparenright}\ {\isasymin}\ R\ {\isasymlongrightarrow}\ P\ y{\isasymrceil}{\isachardoublequoteclose}\isanewline
\ \ $\langle \isa{proof}\rangle$\isanewline

\isacommand{lemma}\isamarkupfalse%
\ wp{\isacharunderscore}nd{\isacharunderscore}fun{\isacharcolon}\ {\isachardoublequoteopen}wp\ F\ {\isasymlceil}P{\isasymrceil}\ {\isacharequal}\ {\isasymlceil}{\isasymlambda}\ x{\isachardot}\ {\isasymforall}\ y{\isachardot}\ y\ {\isasymin}\ {\isacharparenleft}F\ x{\isacharparenright}\ {\isasymlongrightarrow}\ P\ y{\isasymrceil}{\isachardoublequoteclose}\isanewline
\ \ $\langle \isa{proof}\rangle$\isanewline
\end{isabellebody}

Alternatively, we use the categorical forward box operator
$\mathit{fb}\isactrlsub {\isasymF}$ for Kleisli arrows of type
$F{\isacharcolon}{\isacharcolon} {\isacharprime}a\
{\isasymRightarrow}\ {\isacharprime}b\ \mathit{set}$ described in Section~\ref{sec:pt-monad},
\begin{isabellebody}
\isanewline
\isacommand{lemma}\isamarkupfalse%
\ ffb{\isacharunderscore}eq{\isacharcolon}\ {\isachardoublequoteopen}fb\isactrlsub {\isasymF}\ F\ X\ {\isacharequal}\ {\isacharbraceleft}x{\isachardot}\ {\isasymforall}y{\isachardot}\ y\ {\isasymin}\ F\ x\ {\isasymlongrightarrow}\ y\ {\isasymin}\ X{\isacharbraceright}{\isachardoublequoteclose}\isanewline
\ \ $\langle \isa{proof}\rangle$\isanewline
\end{isabellebody}
\noindent or its relational counterpart \isa{fb\isactrlsub {\isasymR}}. 

We now switch to the categorical approach to predicate transformers
based on state transformers and the Kleisi monad of the powerset
functor, as a preliminary $\MKA$-based one with relations has already
been described elsewhere~\cite{MuniveS18}.  Apart from typing and some
minor syntactic differences, the other approaches---predicate
transformers based on $\MKA$ and quantales, and an intermediate
relational semantics for these---yield analogous results and are
equally suitable for verification. This evidences the compositionality
of our approach.

The state and predicate transformer semantics of assignment commands
is based on store update functions, as described in
Section~\ref{sec:discrete-store}. For hybrid programs, it must be
adapted to type ${\isacharprime}a{\isacharcircum}{\isacharprime}n$.

\begin{isabellebody}
\isanewline
\isacommand{definition}\isamarkupfalse%
\ vec{\isacharunderscore}upd\ {\isacharcolon}{\isacharcolon}\ {\isachardoublequoteopen}{\isacharprime}a{\isacharcircum}{\isacharprime}n\ {\isasymRightarrow}\ {\isacharprime}n\ {\isasymRightarrow}\ {\isacharprime}a\ {\isasymRightarrow}\ {\isacharprime}a{\isacharcircum}{\isacharprime}n{\isachardoublequoteclose}\isanewline
\ \ \isakeyword{where}\ {\isachardoublequoteopen}vec{\isacharunderscore}upd\ s\ i\ a\ {\isacharequal}\ {\isacharparenleft}{\isasymchi}\ j{\isachardot}\ {\isacharparenleft}{\isacharparenleft}{\isacharparenleft}{\isachardollar}{\isacharparenright}\ s{\isacharparenright}{\isacharparenleft}i\ {\isacharcolon}{\isacharequal}\ a{\isacharparenright}{\isacharparenright}\ j{\isacharparenright}{\isachardoublequoteclose}\isanewline
\isanewline
\isacommand{definition}\isamarkupfalse%
\ assign\ {\isacharcolon}{\isacharcolon}\ {\isachardoublequoteopen}{\isacharprime}n\ {\isasymRightarrow}\ {\isacharparenleft}{\isacharprime}a{\isacharcircum}{\isacharprime}n\ {\isasymRightarrow}\ {\isacharprime}a{\isacharparenright}\ {\isasymRightarrow}\ {\isacharprime}a{\isacharcircum}{\isacharprime}n\ {\isasymRightarrow}\ {\isacharparenleft}{\isacharprime}a{\isacharcircum}{\isacharprime}n{\isacharparenright}\ set{\isachardoublequoteclose}\ {\isacharparenleft}{\isachardoublequoteopen}{\isacharparenleft}{\isadigit{2}}{\isacharunderscore}\ {\isacharcolon}{\isacharcolon}{\isacharequal}\ {\isacharunderscore}{\isacharparenright}{\isachardoublequoteclose}\ {\isacharbrackleft}{\isadigit{7}}{\isadigit{0}}{\isacharcomma}\ {\isadigit{6}}{\isadigit{5}}{\isacharbrackright}\ {\isadigit{6}}{\isadigit{1}}{\isacharparenright}\ \isanewline
\ \ \isakeyword{where}\ {\isachardoublequoteopen}{\isacharparenleft}x\ {\isacharcolon}{\isacharcolon}{\isacharequal}\ e{\isacharparenright}\ {\isacharequal}\ {\isacharparenleft}{\isasymlambda}s{\isachardot}\ {\isacharbraceleft}vec{\isacharunderscore}upd\ s\ x\ {\isacharparenleft}e\ s{\isacharparenright}{\isacharbraceright}{\isacharparenright}{\isachardoublequoteclose}\ \isanewline
\isanewline
\isacommand{lemma}\isamarkupfalse%
\ ffb{\isacharunderscore}assign{\isacharbrackleft}simp{\isacharbrackright}{\isacharcolon}\ {\isachardoublequoteopen}fb\isactrlsub {\isasymF}\ {\isacharparenleft}x\ {\isacharcolon}{\isacharcolon}{\isacharequal}\ e{\isacharparenright}\ Q\ {\isacharequal}\ {\isacharbraceleft}s{\isachardot}\ {\isacharparenleft}{\isasymchi}\ j{\isachardot}\ {\isacharparenleft}{\isacharparenleft}{\isacharparenleft}{\isachardollar}{\isacharparenright}\ s{\isacharparenright}{\isacharparenleft}x\ {\isacharcolon}{\isacharequal}\ {\isacharparenleft}e\ s{\isacharparenright}{\isacharparenright}{\isacharparenright}\ j{\isacharparenright}\ {\isasymin}\ Q{\isacharbraceright}{\isachardoublequoteclose}\isanewline
\ \ $\langle \isa{proof}\rangle$\isanewline
\end{isabellebody}

\noindent The \isa{{\isacharparenleft}{\isachardollar}{\isacharparenright}}
applies the bijection \isa{\isachardollar} as a function in prefix
notation. 

We write
${\isacharparenleft}x\ {\isacharcolon}{\isacharcolon}{\isacharequal}\
e{\isacharparenright}$ for the semantic illusion for a syntactic
assignment commands, as Isabelle uses
$f{\isacharparenleft}i\ {\isacharcolon}{\isacharequal}\
a{\isacharparenright}$ for function update $f[i\mapsto a]$. Lemma
$\mathit{ffb}{\isacharunderscore}\mathit{assign}$ is then a direct
consequence of $\mathit{ffb}{\isacharunderscore}\mathit{eq}$, and it
coincides with~(\ref{eq:wlp-asgn}) in Section~\ref{sec:discrete-store}
up to minor syntactic differences. In the verification examples that
feature in this article, we have not attempted to hide the functions
that impersonate syntactic expressions and the lambda abstractions
they require. This may be unwieldy for users. It is nevertheless
routine to program more elegant notation with Isabelle~\cite{FosterGMS21}. 

Similarly, $\wlp$s for the control structure commands of hybrid
programs (equations~\ref{eq:wlp-seq}, \ref{eq:wlp-cond}
and~\ref{eq:wlp-star}) are easily derivable.
\begin{isabellebody}
\isanewline
\isacommand{lemma}\isamarkupfalse%
\ ffb{\isacharunderscore}kcomp{\isacharbrackleft}simp{\isacharbrackright}{\isacharcolon}\ {\isachardoublequoteopen}fb\isactrlsub {\isasymF}\ {\isacharparenleft}G\ {\isacharsemicolon}\ F{\isacharparenright}\ P\ {\isacharequal}\ fb\isactrlsub {\isasymF}\ G\ {\isacharparenleft}fb\isactrlsub {\isasymF}\ F\ P{\isacharparenright}{\isachardoublequoteclose}\isanewline
\ \ $\langle \isa{proof}\rangle$\isanewline

\isacommand{lemma}\isamarkupfalse%
\ ffb{\isacharunderscore}if{\isacharunderscore}then{\isacharunderscore }else{\isacharbrackleft}simp{\isacharbrackright}{\isacharcolon}\ {\isachardoublequoteopen}fb\isactrlsub {\isasymF}\ {\isacharparenleft}IF\ T\ THEN\ X\ ELSE\ Y{\isacharparenright}\ Q\ {\isacharequal}\isanewline 
\ \ {\isacharbraceleft}s{\isachardot}\ T\ s\ {\isasymlongrightarrow}\ s\ {\isasymin}\ fb\isactrlsub {\isasymF}\ X\ Q{\isacharbraceright}\ {\isasyminter}\ {\isacharbraceleft}s{\isachardot}\ {\isasymnot}\ T\ s\ {\isasymlongrightarrow}\ s\ {\isasymin}\ fb\isactrlsub {\isasymF}\ Y\ Q{\isacharbraceright}{\isachardoublequoteclose}\isanewline
\ \ $\langle \isa{proof}\rangle$\isanewline

\isacommand{lemma}\isamarkupfalse%
\ ffb{\isacharunderscore}loopI{\isacharcolon}\ {\isachardoublequoteopen}P\ {\isasymle}\ {\isacharbraceleft}s{\isachardot}\ I\ s{\isacharbraceright}\ \ {\isasymLongrightarrow}\ {\isacharbraceleft}s{\isachardot}\ I\ s{\isacharbraceright}\ {\isasymle}\ Q\ {\isasymLongrightarrow}\ {\isacharbraceleft}s{\isachardot}\ I\ s{\isacharbraceright}\ {\isasymle}\ fb\isactrlsub {\isasymF}\ F\ {\isacharbraceleft}s{\isachardot}\ I\ s{\isacharbraceright}\ {\isasymLongrightarrow}\isanewline 
\ \ P\ {\isasymle}\ fb\isactrlsub {\isasymF}\ {\isacharparenleft}LOOP\ F\ INV\ I{\isacharparenright}\ Q{\isachardoublequoteclose}\isanewline
\ \ $\langle \isa{proof}\rangle$\isanewline
\end{isabellebody}

\noindent In these lemmas, ${\isacharsemicolon}$ is syntactic sugar
for the forward Kleisli composition $\circ_K$ and \isa{LOOP} stands
for the Kleene star for state transformers with its annotated
loop-invariant after the keyword \isa{INV}, along the lines of
Section~\ref{sec:mka-pt}.

As in Section~\ref{sec:generalisation}, the general semantics of
evolution commands for continuous vector fields is given by
$G$-guarded orbitals of $f$ along $T$. We have formalised the $\wlp$s
in Proposition~\ref{P:wlpprop-gen}, and a specialisation to local
flows in the context of our locale $\isa{local{\isacharunderscore}flow}$
given by Lemma~\ref{P:wlpprop-var} (equation~(\ref{eq:wlp-evl})).

\begin{isabellebody}
\isanewline\isacommand{notation}\isamarkupfalse%
\ g{\isacharunderscore}orbital\ {\isacharparenleft}{\isachardoublequoteopen}{\isacharparenleft}{\isadigit{1}}x{\isasymacute}{\isacharequal}{\isacharunderscore}\ {\isacharampersand}\ {\isacharunderscore}\ on\ {\isacharunderscore}\ {\isacharunderscore}\ {\isacharat}\ {\isacharunderscore}{\isacharparenright}{\isachardoublequoteclose}{\isacharparenright}\isanewline

\isacommand{lemma}\isamarkupfalse%
\
ffb{\isacharunderscore}g{\isacharunderscore}orbital{\isacharcolon}\isanewline
\ \  {\isachardoublequoteopen}fb\isactrlsub {\isasymF}\
{\isacharparenleft}x{\isasymacute}{\isacharequal} f\ {\isacharampersand}\ G\ on\ T\ S\ {\isacharat}\ t\isactrlsub {\isadigit{0}}{\isacharparenright}\ Q\ {\isacharequal}\ \isanewline
\ \ \ \ {\isacharbraceleft}s{\isachardot}\ {\isasymforall}X{\isasymin}Sols\ {\isacharparenleft}{\isasymlambda}t{\isachardot}\ f{\isacharparenright}\ T\ S\ t\isactrlsub {\isadigit{0}}\ s{\isachardot}\ {\isasymforall}t{\isasymin}T{\isachardot}\ {\isacharparenleft}{\isasymforall}{\isasymtau}{\isasymin}down\ T\ t{\isachardot}\ G\ {\isacharparenleft}X\ {\isasymtau}{\isacharparenright}{\isacharparenright}\ {\isasymlongrightarrow}\ {\isacharparenleft}X\ t{\isacharparenright}\ {\isasymin}\ Q{\isacharbraceright}{\isachardoublequoteclose}\isanewline
\ \ $\langle \isa{proof}\rangle$\isanewline

\isacommand{lemma}\isamarkupfalse%
\ {\isacharparenleft}\isacommand{in}\
local{\isacharunderscore}flow{\isacharparenright}\
ffb{\isacharunderscore}g{\isacharunderscore}ode{\isacharcolon}\isanewline
\ \  {\isachardoublequoteopen}fb\isactrlsub {\isasymF}\
{\isacharparenleft}x{\isasymacute}{\isacharequal} f\ {\isacharampersand}\ G\ on\ T\ S\ {\isacharat}\ {\isadigit{0}}{\isacharparenright}\ Q\ {\isacharequal}\isanewline
\ \ \ \ {\isacharbraceleft}s{\isachardot}\ s\ {\isasymin}\ S\ {\isasymlongrightarrow}\ {\isacharparenleft}{\isasymforall}t{\isasymin}T{\isachardot}\ {\isacharparenleft}{\isasymforall}{\isasymtau}{\isasymin}down\ T\ t{\isachardot}\ G\ {\isacharparenleft}{\isasymphi}\ {\isasymtau}\ s{\isacharparenright}{\isacharparenright}\ {\isasymlongrightarrow}\ {\isacharparenleft}{\isasymphi}\ t\ s{\isacharparenright}\ {\isasymin}\ Q{\isacharparenright}{\isacharbraceright}{\isachardoublequoteclose}\isanewline
\ \ $\langle \isa{proof}\rangle$\isanewline
\end{isabellebody}

\noindent As Lemma $\isa{ffb-g-ode}$ is defined in locale
$\isa{local-flow}$, users are required to check the conditions of the
Picard-Lindel{\"o}f theorem to access this locale and certify that
$\flow$ is indeed a solution of the IVP as part of our first
workflow. 

Finally, we describe our component for reasoning with differential
invariants in the general setting of continuous vector fields, using
our second workflow. We start with definitions and a basic
property from Proposition~\ref{P:inv-prop}.

\begin{isabellebody}
\isanewline
\isacommand{definition}\isamarkupfalse%
\ diff{\isacharunderscore}invariant\ {\isacharcolon}{\isacharcolon}\ {\isachardoublequoteopen}{\isacharparenleft}{\isacharprime}a\ {\isasymRightarrow}\ bool{\isacharparenright}\ {\isasymRightarrow}\ {\isacharparenleft}{\isacharparenleft}{\isacharprime}a{\isacharcolon}{\isacharcolon}real{\isacharunderscore}normed{\isacharunderscore }vector{\isacharparenright}\ {\isasymRightarrow}\ {\isacharprime}a{\isacharparenright}\ {\isasymRightarrow}\ real\ set\ {\isasymRightarrow}\ \isanewline
\ \ {\isacharprime}a\ set\ {\isasymRightarrow}\ real\ {\isasymRightarrow}\ {\isacharparenleft}{\isacharprime}a\ {\isasymRightarrow}\ bool{\isacharparenright}\ {\isasymRightarrow}\ bool{\isachardoublequoteclose}\ \isanewline
\ \ \isakeyword{where}\ {\isachardoublequoteopen}diff{\isacharunderscore}invariant\ I\ f\ T\ S\ t\isactrlsub {\isadigit{0}}\ G\ {\isacharequal}\ {\isacharparenleft}{\isacharparenleft}{\isacharparenleft}g{\isacharunderscore}orbital\ f\ G\ T\ S\ t\isactrlsub {\isadigit{0}}{\isacharparenright}\isactrlsup {\isasymdagger}{\isacharparenright}\ {\isacharbraceleft}s{\isachardot}\ I\ s{\isacharbraceright}\ {\isasymsubseteq}\ {\isacharbraceleft}s{\isachardot}\ I\ s{\isacharbraceright}{\isacharparenright}{\isachardoublequoteclose}\isanewline

\isacommand{lemma}\isamarkupfalse%
\ ffb{\isacharunderscore}diff{\isacharunderscore}inv{\isacharcolon}\ \isanewline
\ \ {\isachardoublequoteopen}diff{\isacharunderscore}invariant\ I\ f\
T\ S\ t\isactrlsub {\isadigit{0}}\ G\ {\isacharequal}\
{\isacharparenleft}{\isacharbraceleft}s{\isachardot}\ I\
s{\isacharbraceright}\ {\isasymle}\ fb\isactrlsub {\isasymF}\
{\isacharparenleft}x{\isasymacute}{\isacharequal} f\ {\isacharampersand}\ G\ on\ T\ S\ {\isacharat}\ t\isactrlsub {\isadigit{0}}{\isacharparenright}\ {\isacharbraceleft}s{\isachardot}\ I\ s{\isacharbraceright}{\isacharparenright}\ {\isachardoublequoteclose}\isanewline
\ \ $\langle \isa{proof}\rangle$\isanewline
\end{isabellebody}

We have formalised the most important rules for reasoning with
differential invariants, including those for the procedure of
Section~\ref{sec:differential-invariants} via
Corollary~\ref{P:wlpprop-gen2} and Lemmas~\ref{P:inv-lemma}
and~\ref{P:invrules}. The formalisation of the first two is
straightforward. We have proved the clauses of~\ref{P:invrules} in
various lemmas, and bundled them under the name
$\mathit{diff}{\isacharunderscore}\mathit{invariant}{\isacharunderscore}\mathit{rules}$. We
show one of these clauses as an example.

\begin{isabellebody}
\isanewline

\isacommand{named{\isacharunderscore}theorems}\isamarkupfalse%
\ diff{\isacharunderscore}invariant{\isacharunderscore}rules\ {\isachardoublequoteopen}compilation\ of\ rules\ for\ differential\ invariants{\isachardot}{\isachardoublequoteclose}\isanewline
\isanewline
\isacommand{lemma}\isamarkupfalse%
\ {\isacharbrackleft}diff{\isacharunderscore}invariant{\isacharunderscore }rules{\isacharbrackright}{\isacharcolon}\isanewline
\ \ \isakeyword{assumes}\
{\isachardoublequoteopen}is{\isacharunderscore}interval\
T{\isachardoublequoteclose}\isanewline
\ \  \isakeyword{and}\ {\isachardoublequoteopen}t\isactrlsub {\isadigit{0}}\ {\isasymin}\ T{\isachardoublequoteclose}\isanewline
\ \  \isakeyword{and}\ {\isachardoublequoteopen}{\isasymforall}X{\isachardot}\ {\isacharparenleft}D\ X\ {\isacharequal}\ {\isacharparenleft}{\isasymlambda}{\isasymtau}{\isachardot}\ f\ {\isacharparenleft}X\ {\isasymtau}{\isacharparenright}{\isacharparenright}\ on\ T{\isacharparenright}\ {\isasymlongrightarrow}\ {\isacharparenleft}D\ {\isacharparenleft}{\isasymlambda}{\isasymtau}{\isachardot}\ {\isasymmu}\ {\isacharparenleft}X\ {\isasymtau}{\isacharparenright}\ {\isacharminus}\ {\isasymnu}\ {\isacharparenleft}X\ {\isasymtau}{\isacharparenright}{\isacharparenright}\ {\isacharequal}\ {\isacharparenleft}{\isacharparenleft}{\isacharasterisk}\isactrlsub R{\isacharparenright}\ {\isadigit{0}}{\isacharparenright}\ on\ T{\isacharparenright}{\isachardoublequoteclose}\isanewline
\ \ \isakeyword{shows}\ {\isachardoublequoteopen}diff{\isacharunderscore}invariant\ {\isacharparenleft}{\isasymlambda}s{\isachardot}\ {\isasymmu}\ s\ {\isacharequal}\ {\isasymnu}\ s{\isacharparenright}\ f\ T\ S\ t\isactrlsub {\isadigit{0}}\ G{\isachardoublequoteclose}\isanewline
\ \ $\langle \isa{proof}\rangle$\isanewline

\isacommand{lemma}\isamarkupfalse%
\ ffb{\isacharunderscore}g{\isacharunderscore}odei{\isacharcolon}\ {\isachardoublequoteopen}P\ {\isasymle}\ {\isacharbraceleft}s{\isachardot}\ I\ s{\isacharbraceright}\ {\isasymLongrightarrow}\ {\isacharbraceleft}s{\isachardot}\ I\ s{\isacharbraceright}\ {\isasymle}\ fb\isactrlsub {\isasymF}\ {\isacharparenleft}x{\isasymacute}{\isacharequal}\ f\ {\isacharampersand}\ G\ on\ T\ S\ {\isacharat}\ t\isactrlsub {\isadigit{0}}{\isacharparenright}\ {\isacharbraceleft}s{\isachardot}\ I\ s{\isacharbraceright}\ {\isasymLongrightarrow}\ \isanewline
\ \ {\isacharbraceleft}s{\isachardot}\ I\ s\ {\isasymand}\ G\ s{\isacharbraceright}\ {\isasymle}\ Q\ {\isasymLongrightarrow}\ P\ {\isasymle}\ fb\isactrlsub {\isasymF}\ {\isacharparenleft}x{\isasymacute}{\isacharequal}\ f\ {\isacharampersand}\ G\ on\ T\ S\ {\isacharat}\ t\isactrlsub {\isadigit{0}}\ DINV\ I{\isacharparenright}\ Q{\isachardoublequoteclose}\isanewline
\ \ $\langle \isa{proof}\rangle$\isanewline
\end{isabellebody}

\noindent Lemma
\isa{ffb{\isacharunderscore}g{\isacharunderscore}odei} 
completes the procedure of
Section~\ref{sec:differential-invariants} by formalising step 2, which
annotates invariants in evolution commands, following the approach
outlined for loops and general commands in $\MKA$ at the end of
Section~\ref{sec:mka-pt}.  With Isabelle, we use the \isa{DINV} keyword.

The two workflows for proving partial correctness specifications with
evolution commands require users to discharge proof obligations for
derivatives. In the case of flows, these must be solutions for vector
fields; in the case of differential invariants, the procedure of
Section~\ref{sec:differential-invariants} requires proving the
assumptions of Lemma~\ref{P:invrules}. To increase proof automation
when reasoning about derivatives, we have bundled several derivative
properties under the name \isa{poly{\isacharunderscore}derivatives} as
a proof method.

\begin{isabellebody}
\isanewline
\isacommand{named{\isacharunderscore}theorems}\isamarkupfalse%
\ poly{\isacharunderscore}derivatives\ {\isachardoublequoteopen}compilation\ of\ optimised\ miscellaneous\ derivative\ rules{\isachardot}{\isachardoublequoteclose}\isanewline
\isanewline
\isacommand{declare}\isamarkupfalse%
\ has{\isacharunderscore}vderiv{\isacharunderscore}on{\isacharunderscore}const\ {\isacharbrackleft}poly{\isacharunderscore}derivatives{\isacharbrackright }\isanewline
\ \ \ \ \isakeyword{and}\ has{\isacharunderscore}vderiv{\isacharunderscore}on{\isacharunderscore}id\ {\isacharbrackleft}poly{\isacharunderscore}derivatives{\isacharbrackright }\isanewline
\ \ \ \ \isakeyword{and}\ has{\isacharunderscore}vderiv{\isacharunderscore}on{\isacharunderscore}add\ {\isacharbrackleft}THEN\ has{\isacharunderscore}vderiv{\isacharunderscore}on{\isacharunderscore}eq{\isacharunderscore}rhs{\isacharcomma}\ poly{\isacharunderscore}derivatives{\isacharbrackright }\isanewline
\ \ \ \ \isakeyword{and}\ has{\isacharunderscore}vderiv{\isacharunderscore}on{\isacharunderscore}diff\ {\isacharbrackleft}THEN\ has{\isacharunderscore}vderiv{\isacharunderscore}on{\isacharunderscore}eq{\isacharunderscore}rhs{\isacharcomma}\ poly{\isacharunderscore}derivatives{\isacharbrackright }\isanewline
\ \ \ \ \isakeyword{and}\ has{\isacharunderscore}vderiv{\isacharunderscore}on{\isacharunderscore}mult\ {\isacharbrackleft}THEN\ has{\isacharunderscore}vderiv{\isacharunderscore}on{\isacharunderscore}eq{\isacharunderscore}rhs{\isacharcomma}\ poly{\isacharunderscore}derivatives{\isacharbrackright }\isanewline

\isacommand{lemma}\isamarkupfalse%
\ {\isacharbrackleft}poly{\isacharunderscore}derivatives{\isacharbrackright}{\isacharcolon}\ {\isachardoublequoteopen}D\ f\ {\isacharequal}\ f{\isacharprime}\ on\ T\ {\isasymLongrightarrow}\ g\ {\isacharequal}\ {\isacharparenleft}{\isasymlambda}t{\isachardot}\ {\isacharminus}\ f{\isacharprime}\ t{\isacharparenright}\ {\isasymLongrightarrow}\ D\ {\isacharparenleft}{\isasymlambda}t{\isachardot}\ {\isacharminus}\ f\ t{\isacharparenright}\ {\isacharequal}\ g\ on\ T{\isachardoublequoteclose}\isanewline
\ \ $\langle \isa{proof}\rangle$\isanewline

\isacommand{lemma}\isamarkupfalse%
\ {\isacharbrackleft}poly{\isacharunderscore}derivatives{\isacharbrackright}{\isacharcolon}\ {\isachardoublequoteopen}{\isacharparenleft}a{\isacharcolon}{\isacharcolon}real{\isacharparenright}\ {\isasymnoteq}\ {\isadigit{0}}\ {\isasymLongrightarrow}\ D\ f\ {\isacharequal}\ f{\isacharprime}\ on\ T\ {\isasymLongrightarrow}\ g\ {\isacharequal}\ {\isacharparenleft}{\isasymlambda}t{\isachardot}\ {\isacharparenleft}f{\isacharprime}\ t{\isacharparenright}{\isacharslash}a{\isacharparenright}\ {\isasymLongrightarrow}\isanewline
\ \ D\ {\isacharparenleft}{\isasymlambda}t{\isachardot}\ {\isacharparenleft}f\ t{\isacharparenright}{\isacharslash}a{\isacharparenright}\ {\isacharequal}\ g\ on\ T{\isachardoublequoteclose}\isanewline
\ \ $\langle \isa{proof}\rangle$\isanewline

\isacommand{lemma}\isamarkupfalse%
\ {\isacharbrackleft}poly{\isacharunderscore}derivatives{\isacharbrackright}{\isacharcolon}\ {\isachardoublequoteopen}n\ {\isasymge}\ {\isadigit{1}}\ {\isasymLongrightarrow}\ D\ {\isacharparenleft}f{\isacharcolon}{\isacharcolon}real\ {\isasymRightarrow}\ real{\isacharparenright}\ {\isacharequal}\ f{\isacharprime}\ on\ T\ {\isasymLongrightarrow}\isanewline
\ \ g\ {\isacharequal}\ {\isacharparenleft}{\isasymlambda}t{\isachardot}\ n\ {\isacharasterisk}\ {\isacharparenleft}f{\isacharprime}\ t{\isacharparenright}\ {\isacharasterisk}\ {\isacharparenleft}f\ t{\isacharparenright}{\isacharcircum}{\isacharparenleft}n{\isacharminus}{\isadigit{1}}{\isacharparenright}{\isacharparenright}\ {\isasymLongrightarrow}\ D\ {\isacharparenleft}{\isasymlambda}t{\isachardot}\ {\isacharparenleft}f\ t{\isacharparenright}{\isacharcircum}n{\isacharparenright}\ {\isacharequal}\ g\ on\ T{\isachardoublequoteclose}\isanewline
\ \ $\langle \isa{proof}\rangle$\isanewline

\isacommand{lemma}\isamarkupfalse%
\ {\isacharbrackleft}poly{\isacharunderscore}derivatives{\isacharbrackright}{\isacharcolon}\ {\isachardoublequoteopen}D\ {\isacharparenleft}f{\isacharcolon}{\isacharcolon}real\ {\isasymRightarrow}\ real{\isacharparenright}\ {\isacharequal}\ f{\isacharprime}\ on\ T\ {\isasymLongrightarrow}\isanewline
\ \ g\ {\isacharequal}\ {\isacharparenleft}{\isasymlambda}t{\isachardot}\ {\isacharminus}\ {\isacharparenleft}f{\isacharprime}\ t{\isacharparenright}\ {\isacharasterisk}\ sin\ {\isacharparenleft}f\ t{\isacharparenright}{\isacharparenright}\ {\isasymLongrightarrow}\ D\ {\isacharparenleft}{\isasymlambda}t{\isachardot}\ cos\ {\isacharparenleft}f\ t{\isacharparenright}{\isacharparenright}\ {\isacharequal}\ g\ on\ T{\isachardoublequoteclose}\isanewline
\ \ $\langle \isa{proof}\rangle$\isanewline

\isacommand{lemma}\isamarkupfalse%
\ {\isacharbrackleft}poly{\isacharunderscore}derivatives{\isacharbrackright}{\isacharcolon}\ {\isachardoublequoteopen}D\ {\isacharparenleft}f{\isacharcolon}{\isacharcolon}real\ {\isasymRightarrow}\ real{\isacharparenright}\ {\isacharequal}\ f{\isacharprime}\ on\ T\ {\isasymLongrightarrow}\ g\ {\isacharequal}\ {\isacharparenleft}{\isasymlambda}t{\isachardot}\ {\isacharparenleft}f{\isacharprime}\ t{\isacharparenright}\ {\isacharasterisk}\ cos\ {\isacharparenleft}f\ t{\isacharparenright}{\isacharparenright}\ {\isasymLongrightarrow}\isanewline
\ \ D\ {\isacharparenleft}{\isasymlambda}t{\isachardot}\ sin\ {\isacharparenleft}f\ t{\isacharparenright}{\isacharparenright}\ {\isacharequal}\ g\ on\ T{\isachardoublequoteclose}\isanewline
\ \ $\langle \isa{proof}\rangle$\isanewline

\isacommand{lemma}\isamarkupfalse%
\ {\isacharbrackleft}poly{\isacharunderscore}derivatives{\isacharbrackright}{\isacharcolon}\ {\isachardoublequoteopen}D\ {\isacharparenleft}f{\isacharcolon}{\isacharcolon}real\ {\isasymRightarrow}\ real{\isacharparenright}\ {\isacharequal}\ f{\isacharprime}\ on\ T\ {\isasymLongrightarrow}\ g\ {\isacharequal}\ {\isacharparenleft}{\isasymlambda}t{\isachardot}\ {\isacharparenleft}f{\isacharprime}\ t{\isacharparenright}\ {\isacharasterisk}\ exp\ {\isacharparenleft}f\ t{\isacharparenright}{\isacharparenright}\ {\isasymLongrightarrow}\isanewline
\ \ D\ {\isacharparenleft}{\isasymlambda}t{\isachardot}\ exp\ {\isacharparenleft}f\ t{\isacharparenright}{\isacharparenright}\ {\isacharequal}\ g\ on\ T{\isachardoublequoteclose}\isanewline
\ \ $\langle \isa{proof}\rangle$\isanewline
\end{isabellebody}

Isabelle can now apply rules iteratively and check, for pairs of
functions, if one is a derivative of the other. This is often fully
automatic. The following lemma shows an example that involves a mix of
polynomials and transcendental functions beyond differential fields
with $a_0$ to $a_5$ being constants and $t$ the polynomial variable.

\begin{isabellebody}
\isanewline
\isacommand{lemma}\isamarkupfalse%
\ {\isachardoublequoteopen}c\ {\isasymnoteq}\ {\isadigit{0}}\ {\isasymLongrightarrow}\ D\ {\isacharparenleft}{\isasymlambda}t{\isachardot}\ $a_5$\ {\isacharasterisk}\ t{\isacharcircum}{\isadigit{5}}\ {\isacharplus}\ $a_3$\ {\isacharasterisk}\ {\isacharparenleft}t{\isacharcircum}{\isadigit{3}}\ {\isacharslash}\ c{\isacharparenright}\ {\isacharminus}\ $a_2$\ {\isacharasterisk}\ exp\ {\isacharparenleft}t{\isacharcircum}{\isadigit{2}}{\isacharparenright}\ {\isacharplus}\ $a_1$\ {\isacharasterisk}\ cos\ t\ {\isacharplus}\ $a_0${\isacharparenright}\isanewline
\ \ {\isacharequal}\ {\isacharparenleft}{\isasymlambda}t{\isachardot}\ {\isadigit{5}}\ {\isacharasterisk}\ $a_5$\ {\isacharasterisk}\ t{\isacharcircum}{\isadigit{4}}\ {\isacharplus}\ {\isadigit{3}}\ {\isacharasterisk}\ $a_3$\ {\isacharasterisk}\ {\isacharparenleft}t{\isacharcircum}{\isadigit{2}}\ {\isacharslash}\ c{\isacharparenright}\ {\isacharminus}\ {\isadigit{2}}\ {\isacharasterisk}\ $a_2$\ {\isacharasterisk}\ t\ {\isacharasterisk}\ exp\ {\isacharparenleft}t{\isacharcircum}{\isadigit{2}}{\isacharparenright}\ {\isacharminus}\ $a_1$\ {\isacharasterisk}\ sin\ t{\isacharparenright}\ on\ T{\isachardoublequoteclose}\isanewline
\ \ \isacommand{by}\isamarkupfalse%
{\isacharparenleft}auto\ intro{\isacharbang}{\isacharcolon}\ poly{\isacharunderscore}derivatives{\isacharparenright}\isanewline
\end{isabellebody}

The formalisation of more advanced heuristics for such functions, and
the integration of decision procedures for suitable classes, is left
for future work

The complete Isabelle formalisation, including the other two predicate
transformer algebras and the relational semantics, can be found in the
Archive of Formal Proofs~\cite{afp:hybrid}.

We briefly reflect on our experience with the Isabelle formalisation
of our framework. $\MKA$, its relational model and the concrete
relational semantics for traditional while-programs are so far the
most developed and versatile starting point for our hybrid systems
verification components. The full formalisation of a rudimentary Hoare
logic component for this setting using a generalised Kleene algebra
from Isabelle's main libraries fits on two A4 pages~\cite{Struth18}; a
similar development for a Hoare logic for hybrid programs is discussed
in a successor paper~\cite{FosterMS20}. Our standalone $\MKA$-based
verification component for traditional while programs fills about
seven A4 pages.  For hybrid programs, in theory, only a concrete
semantics for hybrid programs needs to be plugged in as a replacement
of the semantics described in Section~\ref{sec:discrete-store}. In
practice, however, Isabelle's instantiations often make theory
hierarchies non-compositional as each type can only be instantiated in
one way. We faced such a clash of instances between Isabelle's Kleene
algebra and analysis hierarchies and hence had to customise the former
for our purposes.

Replacing the intermediate relational semantics by state transformers
required some background work, simply because the former are well
supported by Isabelle whereas the latter are new. Interestingly, it is
possible to propagate theorems automatically along the isomorphisms
between these semantics like for type classes, locales and their
instantiations and interpretation. Isabelle's transfer and lifting
packages provide an infrastructure for this, which remains by and
large unexplored. We leave this for future work.

The categorical approach to transformer quantaloids is more
complex---both conceptually and from a formalisation point of
view---than the $\MKA$ based one, in particular when state
transformers are integrated via the powerset monad. At the level of
verification conditions generation, however, there are almost no
differences.  Once again a stripped down component can be generated
that just suffices for verification condition generation, and we are
using it in subsequent work~\cite{FosterGMS21}. Relative to Isabelle's
main libraries it fills merely four pages~\cite{afp:hybrid}. Working
with quantales instead of quantaloids might seem mathematically
simpler, but with Isabelle it is actually more tedious, as subtypes
for endofunctions need to be created.

In sum, for simple verification tasks, the lightweight stripped down
predicate transformer algebras obtained from $\MKA$ or quantaloids
seem preferable; for more complex program transformations or
refinements, the integration into the full $\MKA$ hierarchy or
categorical predicate transformer component is certainly beneficial.


\section{Isabelle Support for $\dL$-Style Reasoning}\label{sec:isa-dL}

This section lists our formalisation of semantic variants of the
most important axioms and inference rules of $\dL$ in Isabelle
outlined in Section~\ref{sec:dL}.  It covers  all
three predicate transformer semantics as well as the relational and
state transformer model. Once again, we only show state
transformers in the categorical approach.

We have formalised a generalised version of the $\dL$-rules with
parameters $T$, $S$ and $t_0$ with intervals $U$ and for orbitals. We
can easily instantiate them to $\reals$, $\reals^V$ and $0$,
respectively. This enables users to perform verification proofs in the
style of $\dL$ and establishes soundness of these rules relative to
our semantics as a side effect. First we show our formalisations of
(\ref{eq:DS}) and (\ref{eq:dSolve}).

\begin{isabellebody}
\isanewline
\isacommand{lemma}\isamarkupfalse%
\ DS{\isacharcolon}\ \isanewline
\ \ \isakeyword{fixes}\ c{\isacharcolon}{\isacharcolon}{\isachardoublequoteopen}{\isacharprime}a{\isacharcolon}{\isacharcolon}{\isacharbraceleft}heine{\isacharunderscore}borel{\isacharcomma}\ banach{\isacharbraceright}{\isachardoublequoteclose}\isanewline
\ \ \isakeyword{shows}\ {\isachardoublequoteopen}fb\isactrlsub
{\isasymF}\ {\isacharparenleft}x{\isasymacute}{\isacharequal} {\isacharparenleft}{\isasymlambda}s{\isachardot}\ c{\isacharparenright}\ {\isacharampersand}\ G{\isacharparenright}\ Q\ {\isacharequal}\ {\isacharbraceleft}x{\isachardot}\ {\isasymforall}t{\isachardot}\ {\isacharparenleft}{\isasymforall}{\isasymtau}{\isasymle}t{\isachardot}\ G\ {\isacharparenleft}x{\isacharplus}{\isasymtau}\ {\isacharasterisk}\isactrlsub R\ c{\isacharparenright}{\isacharparenright}\ {\isasymlongrightarrow}\ {\isacharparenleft}x{\isacharplus}t\ {\isacharasterisk}\isactrlsub R\ c{\isacharparenright}\ {\isasymin}\ Q{\isacharbraceright}{\isachardoublequoteclose}\isanewline
\ \ $\langle \isa{proof}\rangle$\isanewline

\isacommand{lemma}\isamarkupfalse%
\ solve{\isacharcolon}\isanewline
\ \ \isakeyword{assumes}\ {\isachardoublequoteopen}local{\isacharunderscore}flow\ f\ UNIV\ UNIV\ {\isasymphi}{\isachardoublequoteclose}\isanewline
\ \ \ \ \isakeyword{and}\ {\isachardoublequoteopen}{\isasymforall}s{\isachardot}\ s\ {\isasymin}\ P\ {\isasymlongrightarrow}\ {\isacharparenleft}{\isasymforall}t{\isachardot}\ {\isacharparenleft}{\isasymforall}{\isasymtau}{\isasymle}t{\isachardot}\ G\ {\isacharparenleft}{\isasymphi}\ {\isasymtau}\ s{\isacharparenright}{\isacharparenright}\ {\isasymlongrightarrow}\ {\isacharparenleft}{\isasymphi}\ t\ s{\isacharparenright}\ {\isasymin}\ Q{\isacharparenright}{\isachardoublequoteclose}\isanewline
\ \ \isakeyword{shows}\ {\isachardoublequoteopen}P\ {\isasymle}\
fb\isactrlsub {\isasymF}\
{\isacharparenleft}x{\isasymacute}{\isacharequal} f\ {\isacharampersand}\ G{\isacharparenright}\ Q{\isachardoublequoteclose}\isanewline
\ \ $\langle \isa{proof}\rangle$\isanewline
\end{isabellebody}

\noindent Next we list semantic variants of the five $\dL$ axioms and
inference rules for reasoning with differential invariants discussed
in Section~\ref{sec:dL}. Recall that due to our semantic approach, 
evolution commands in these rules only require the vector field
\isa{f\ {\isacharcolon}{\isacharcolon}\ {\isachardoublequoteopen}{\isacharprime}a\ {\isasymRightarrow}\ {\isacharprime}a{\isachardoublequoteclose}} and guard \isa{G\ {\isacharcolon}{\isacharcolon}\ {\isachardoublequoteopen}{\isacharprime}a\ {\isasymRightarrow}\ bool{\isachardoublequoteclose}},
while the \isa{x{\isasymacute}{\isacharequal}} is just syntactic sugar 
to resemble ODEs.

\begin{isabellebody}
\isanewline
\isacommand{lemma}\isamarkupfalse%
\ DW{\isacharcolon}\ {\isachardoublequoteopen}fb\isactrlsub
{\isasymF}\ {\isacharparenleft}x{\isasymacute}{\isacharequal}\ f\
{\isacharampersand}\ G{\isacharparenright}\ Q\ {\isacharequal}\
fb\isactrlsub {\isasymF}\
{\isacharparenleft}x{\isasymacute}{\isacharequal} f\ {\isacharampersand}\ G{\isacharparenright}\ {\isacharbraceleft}s{\isachardot}\ G\ s\ {\isasymlongrightarrow}\ s\ {\isasymin}\ Q{\isacharbraceright}{\isachardoublequoteclose}\isanewline
\ \ $\langle \isa{proof}\rangle$\isanewline

\isacommand{lemma}\isamarkupfalse%
\ dW{\isacharcolon}\
{\isachardoublequoteopen}{\isacharbraceleft}s{\isachardot}\ G\
s{\isacharbraceright}\ {\isasymle}\ Q\ {\isasymLongrightarrow}\ P\
{\isasymle}\ fb\isactrlsub {\isasymF}\
{\isacharparenleft}x{\isasymacute}{\isacharequal} f\ {\isacharampersand}\ G{\isacharparenright}\ Q{\isachardoublequoteclose}\isanewline
\ \ $\langle \isa{proof}\rangle$\isanewline

\isacommand{lemma}\isamarkupfalse%
\ DC{\isacharcolon}\isanewline
\ \ \isakeyword{assumes}\ {\isachardoublequoteopen}fb\isactrlsub
{\isasymF}\ {\isacharparenleft}x{\isasymacute}{\isacharequal} f\ {\isacharampersand}\ G{\isacharparenright}\ {\isacharbraceleft}s{\isachardot}\ C\ s{\isacharbraceright}\ {\isacharequal}\ UNIV{\isachardoublequoteclose}\isanewline
\ \ \isakeyword{shows}\ {\isachardoublequoteopen}fb\isactrlsub
{\isasymF}\ {\isacharparenleft}x{\isasymacute}{\isacharequal}\ f\
{\isacharampersand}\ G{\isacharparenright}\ Q\ {\isacharequal}\
fb\isactrlsub {\isasymF}\
{\isacharparenleft}x{\isasymacute}{\isacharequal} f\ {\isacharampersand}\ {\isacharparenleft}{\isasymlambda}s{\isachardot}\ G\ s\ {\isasymand}\ C\ s{\isacharparenright}{\isacharparenright}\ Q{\isachardoublequoteclose}\isanewline
\ \ $\langle \isa{proof}\rangle$\isanewline

\isacommand{lemma}\isamarkupfalse%
\ dC{\isacharcolon}\isanewline
\ \ \isakeyword{assumes}\ {\isachardoublequoteopen}P\ {\isasymle}\
fb\isactrlsub {\isasymF}\
{\isacharparenleft}x{\isasymacute}{\isacharequal} f\ {\isacharampersand}\ G{\isacharparenright}\ {\isacharbraceleft}s{\isachardot}\ C\ s{\isacharbraceright}{\isachardoublequoteclose}\isanewline
\ \ \ \ \isakeyword{and}\ {\isachardoublequoteopen}P\ {\isasymle}\
fb\isactrlsub {\isasymF}\
{\isacharparenleft}x{\isasymacute}{\isacharequal} f\ {\isacharampersand}\ {\isacharparenleft}{\isasymlambda}s{\isachardot}\ G\ s\ {\isasymand}\ C\ s{\isacharparenright}{\isacharparenright}\ Q{\isachardoublequoteclose}\isanewline
\ \ \isakeyword{shows}\ {\isachardoublequoteopen}P\ {\isasymle}\
fb\isactrlsub {\isasymF}\
{\isacharparenleft}x{\isasymacute}{\isacharequal} f\ {\isacharampersand}\ G{\isacharparenright}\ Q{\isachardoublequoteclose}\isanewline
\ \ $\langle \isa{proof}\rangle$\isanewline

\isacommand{lemma}\isamarkupfalse%
\ dI{\isacharcolon}\isanewline
\ \ \isakeyword{assumes}\ {\isachardoublequoteopen}P\ {\isasymle}\
{\isacharbraceleft}s{\isachardot}\ I\
s{\isacharbraceright}{\isachardoublequoteclose}\isanewline
\ \ \ \  \isakeyword{and}\
{\isachardoublequoteopen}diff{\isacharunderscore}invariant\ I\ f\
UNIV\ UNIV\ {\isadigit{0}}\ G{\isachardoublequoteclose}\isanewline
\ \ \ \  \isakeyword{and}\ {\isachardoublequoteopen}{\isacharbraceleft}s{\isachardot}\ I\ s{\isacharbraceright}\ {\isasymle}\ Q{\isachardoublequoteclose}\isanewline
\ \ \isakeyword{shows}\ {\isachardoublequoteopen}P\ {\isasymle}\
fb\isactrlsub {\isasymF}\
{\isacharparenleft}x{\isasymacute}{\isacharequal} f\ {\isacharampersand}\ G{\isacharparenright}\ Q{\isachardoublequoteclose}\isanewline
\ \ $\langle \isa{proof}\rangle$\isanewline
\end{isabellebody}

Additional $\dL$ rules can easily be formalised. More recent work
features for instance a ghost rule~\cite{FosterGMS21}, which is
heavily used for reasoning with invariants in $\dL$, but seems less
relevant to our semantic approach~\cite{MitschMJZWZ20}.


\section{Verification Examples}\label{sec:examples}

This section explains the formalisation of the bouncing ball examples
from Section~\ref{sec:hybrid-store} and
\ref{sec:differential-invariants} with Isabelle; and we add two
further verification examples using a simple circular pendulum.  All
four of them use Isabelle's type $\isa{\isadigit{2}}$ of two elements.
It denotes the set of variables $V$ of hybrid programs over the state
space $\reals^V$ for $|V|=2$. We follow Isabelle's notation and write
\isa{{\isadigit{0}}{\isacharcolon}{\isacharcolon}{\isadigit{2}}} and
\isa{{\isadigit{1}}{\isacharcolon}{\isacharcolon}{\isadigit{2}}} for
the two variables and their type. As such a formalisation of variables
is rather unwieldy, more recent extensions to our framework support
more general name spaces, more sophisticated store models and a more
user-friendly specification language for hybrid programs and
assertions~\cite{FosterGMS21}. The examples in this section should
therefore be taken cum grano salis.

\begin{example}[Bouncing Ball via Flow]\label{ex:bouncing-ball-flow}
  First, we formalise Example~\ref{ex:ball} with our verification
  components for flows, using our first workflow. We write
  \isa{{\isadigit{0}}{\isacharcolon}{\isacharcolon}{\isadigit{2}}} for
  the ball's position starting from height $h$,
  \isa{{\isadigit{1}}{\isacharcolon}{\isacharcolon}{\isadigit{2}}} for
  its velocity, and $s{\isachardollar}{\isadigit{0}}$ and
  $s{\isachardollar}{\isadigit{1}}$ for $s_x$ and $s_v$. We formalise
  the vector field $f\, (s_x,s_v)^T = (s_v,-g)^T$ for the ball as

\begin{isabellebody}
\isanewline
\isacommand{abbreviation}\isamarkupfalse%
\ fball\ {\isacharcolon}{\isacharcolon}\ {\isachardoublequoteopen}real\ {\isasymRightarrow}\ real{\isacharcircum}{\isadigit{2}}\ {\isasymRightarrow}\ real{\isacharcircum}{\isadigit{2}}{\isachardoublequoteclose}\ {\isacharparenleft}{\isachardoublequoteopen}f{\isachardoublequoteclose}{\isacharparenright}\ \isanewline
\ \ \isakeyword{where}\ {\isachardoublequoteopen}f\ g\ s\ {\isasymequiv}\ {\isacharparenleft}{\isasymchi}\ i{\isachardot}\ if\ i{\isacharequal}{\isadigit{0}}\ then\ s{\isachardollar}{\isadigit{1}}\ else\ g{\isacharparenright}{\isachardoublequoteclose}\isanewline
\end{isabellebody}

\noindent We can now state the partial correctness specification for
the bouncing ball in Isabelle, where the loop invariant $I$ is
that of Section~\ref{sec:hybrid-store}, but written
slightly differently to enhance proof automation.

\begin{isabellebody}
\isanewline
\isacommand{lemma}\isamarkupfalse%
\ bouncing{\isacharunderscore}ball{\isacharcolon}\ {\isachardoublequoteopen}g\ {\isacharless}\ {\isadigit{0}}\ {\isasymLongrightarrow}\ h\ {\isasymge}\ {\isadigit{0}}\ {\isasymLongrightarrow}\ \isanewline
\ \ {\isacharbraceleft}s{\isachardot}\ s{\isachardollar}{\isadigit{0}}\ {\isacharequal}\ h\ {\isasymand}\ s{\isachardollar}{\isadigit{1}}\ {\isacharequal}\ {\isadigit{0}}{\isacharbraceright}\ {\isasymle}\ fb\isactrlsub {\isasymF}\ \isanewline
\ \ {\isacharparenleft}LOOP\ {\isacharparenleft}\isanewline
\ \ \ \ {\isacharparenleft}x{\isasymacute}{\isacharequal}{\isacharparenleft}f\ g{\isacharparenright}\ {\isacharampersand}\ {\isacharparenleft}{\isasymlambda}\ s{\isachardot}\ s{\isachardollar}{\isadigit{0}}\ {\isasymge}\ {\isadigit{0}}{\isacharparenright}{\isacharparenright}\ {\isacharsemicolon}\ \isanewline
\ \ \ \ {\isacharparenleft}IF\ {\isacharparenleft}{\isasymlambda}\ s{\isachardot}\ s{\isachardollar}{\isadigit{0}}\ {\isacharequal}\ {\isadigit{0}}{\isacharparenright}\ THEN\ {\isacharparenleft}{\isadigit{1}}\ {\isacharcolon}{\isacharcolon}{\isacharequal}\ {\isacharparenleft}{\isasymlambda}s{\isachardot}\ {\isacharminus}\ s{\isachardollar}{\isadigit{1}}{\isacharparenright}{\isacharparenright}\ ELSE\ skip{\isacharparenright}{\isacharparenright}\isanewline
\ \ INV\ {\isacharparenleft}{\isasymlambda}s{\isachardot}\ {\isadigit{0}}\ {\isasymle}\ s{\isachardollar}{\isadigit{0}}\ {\isasymand}{\isadigit{2}}\ {\isasymcdot}\ g\ {\isasymcdot}\ s{\isachardollar}{\isadigit{0}}\ {\isacharminus}\ {\isadigit{2}}\ {\isasymcdot}\ g\ {\isasymcdot}\ h\ {\isacharminus}\ s{\isachardollar}{\isadigit{1}}\ {\isasymcdot}\ s{\isachardollar}{\isadigit{1}}\ {\isacharequal}\ {\isadigit{0}}{\isacharparenright}{\isacharparenright}\isanewline
\ \ {\isacharbraceleft}s{\isachardot}\ {\isadigit{0}}\ {\isasymle}\
s{\isachardollar}{\isadigit{0}}\ {\isasymand}\
s{\isachardollar}{\isadigit{0}}\ {\isasymle}\
h{\isacharbraceright}{\isachardoublequoteclose}\isanewline
\end{isabellebody}

\noindent The proof of this lemma is shown below. It follows that in
Example~\ref{ex:ball}, but requires some intermediate lemmas. For
example, if we first apply rule \isa{ffb{\isacharunderscore}loopI}
(\ref{eq:wlp-star}), the subgoals $P\leq I$ and $I\leq Q$, for
$P= (\lambda s.\ s_x = h\land s_v = 0)$ and
$Q = (\lambda s.\ 0\leq s_x\leq h)$, need to be proven. They can be
discharged automatically after supplying some lemmas about real
arithmetic, which have been bundled under the name
\isa{bb{\isacharunderscore}real{\isacharunderscore}arith}. We show one
of them below to give an impression.

\begin{isabellebody}
\isanewline
\isacommand{named{\isacharunderscore}theorems}\isamarkupfalse%
\ bb{\isacharunderscore}real{\isacharunderscore}arith\ {\isachardoublequoteopen}real\ arithmetic\ properties\ for\ the\ bouncing\ ball{\isachardot}{\isachardoublequoteclose}\isanewline
\isanewline
\isacommand{lemma}\isamarkupfalse%
\ {\isacharbrackleft}bb{\isacharunderscore}real{\isacharunderscore }arith{\isacharbrackright}{\isacharcolon}\
{\isachardoublequoteopen}{\isadigit{0}}\ {\isachargreater}\
g{\isachardoublequoteclose}
 {\isasymLongrightarrow}\ {\isadigit{2}}\
{\isasymcdot}\ g\ {\isasymcdot}\ x\ {\isacharminus}\ {\isadigit{2}}\
{\isasymcdot}\ g\ {\isasymcdot}\ h\ {\isacharequal}\ v\ {\isasymcdot}\
 {\isasymLongrightarrow}\
 {\isachardoublequoteopen}{\isacharparenleft}x{\isacharcolon}{\isacharcolon}real{\isacharparenright}\ {\isasymle}\ h{\isachardoublequoteclose}\isanewline
\ \ $\langle \isa{proof}\rangle$\isanewline
\end{isabellebody}

\noindent These properties depend on distributivity and commutativity
properties that Isabelle cannot simplify immediately. As we are not
working within a well defined language, such as differential rings or
fields, we have not attempted to automate them any further, so that
proofs require some user interaction. 

The remaining rules, that is, \isa{ffb{\isacharunderscore}kcomp}
(\ref{eq:wlp-seq}),
\isa{ffb{\isacharunderscore}if{\isacharunderscore}then{\isacharunderscore}else}
(\ref{eq:wlp-cond}), and \isa{ffb{\isacharunderscore}assign}
(\ref{eq:wlp-asgn}), have been added to Isabelle's automatic proof
tools. It then remains to compute the $\wlp$ for the evolution command
of the bouncing ball. To use
\isa{local{\isacharunderscore}flow{\isachardot
  }ffb{\isacharunderscore}g{\isacharunderscore}ode}
(\ref{eq:wlp-evl}), we follow the procedure in
Section~\ref{sec:hybrid-store}. We need to check that the vector field
is Lipschitz continuous, supply the local flow as in
Example~\ref{ex:ball}, and check that it solves the IVP and satisfies
the flow conditions.

\begin{isabellebody}
\isanewline
\isacommand{abbreviation}\isamarkupfalse%
\ ball{\isacharunderscore}flow\ {\isacharcolon}{\isacharcolon}\ {\isachardoublequoteopen}real\ {\isasymRightarrow}\ real\ {\isasymRightarrow}\ real{\isacharcircum}{\isadigit{2}}\ {\isasymRightarrow}\ real{\isacharcircum}{\isadigit{2}}{\isachardoublequoteclose}\ {\isacharparenleft}{\isachardoublequoteopen}{\isasymphi}{\isachardoublequoteclose}{\isacharparenright}\ \isanewline
\ \ \isakeyword{where}\ {\isachardoublequoteopen}{\isasymphi}\ g\ t\ s\ {\isasymequiv}\ {\isacharparenleft}{\isasymchi}\ i{\isachardot}\ if\ i{\isacharequal}{\isadigit{0}}\ then\ g\ {\isasymcdot}\ t\ {\isacharcircum}\ {\isadigit{2}}{\isacharslash}{\isadigit{2}}\ {\isacharplus}\ s{\isachardollar}{\isadigit{1}}\ {\isasymcdot}\ t\ {\isacharplus}\ s{\isachardollar}{\isadigit{0}}\ else\ g\ {\isasymcdot}\ t\ {\isacharplus}\ s{\isachardollar}{\isadigit{1}}{\isacharparenright}{\isachardoublequoteclose}\isanewline
\isanewline
\isacommand{lemma}\isamarkupfalse%
\ local{\isacharunderscore}flow{\isacharunderscore}ball{\isacharcolon}\ {\isachardoublequoteopen}local{\isacharunderscore}flow\ {\isacharparenleft}f\ g{\isacharparenright}\ UNIV\ UNIV\ {\isacharparenleft}{\isasymphi}\ g{\isacharparenright}{\isachardoublequoteclose}\isanewline
\ \ $\langle \isa{proof}\rangle$\isanewline
\end{isabellebody}

\noindent The arithmetic computations with real numbers at the end
of Example~\ref{ex:ball} are then discharged automatically by adding
the rules in \isa{bb{\isacharunderscore}real{\isacharunderscore}arith}
to Isabelle's automatic tools. The resulting two-line proof of the
bouncing ball is shown below.

\begin{isabellebody}\isanewline
\isacommand{apply}\isamarkupfalse%
{\isacharparenleft}rule\ wp{\isacharunderscore}loopI{\isacharcomma}\ simp{\isacharunderscore}all\ add{\isacharcolon}\ local{\isacharunderscore}flow{\isachardot}wp{\isacharunderscore }g{\isacharunderscore}ode{\isacharbrackleft}OF\ local{\isacharunderscore}flow{\isacharunderscore}ball{\isacharbrackright}{\isacharparenright}\isanewline
\ \ \isacommand{by}\isamarkupfalse%
\ {\isacharparenleft}auto\ simp{\isacharcolon}\ bb{\isacharunderscore}real{\isacharunderscore }arith{\isacharparenright}\isanewline
\end{isabellebody}

Overall, the verification proof covers less than a page and a half in
the proof document---and this is mainly due to the few arithmetic
calculations in the background that require user interaction. All
other proofs make heavy use of Isabelle's simplifiers and are by and
large automatic.  \qed
\end{example}

\begin{example}[Bouncing Ball via Invariant]\label{ex:bouncing-ball-inv}
  This example formalises the invariant-based proof from
  Example~\ref{ex:ball-inv} using our second workflow. The correctness
  specification changes in that we annotate the differential invariant
  ab initio.

\begin{isabellebody}
\isanewline
\isacommand{lemma}\isamarkupfalse%
\ bouncing{\isacharunderscore}ball{\isacharunderscore }invariants{\isacharcolon}\ {\isachardoublequoteopen}g\ {\isacharless}\ {\isadigit{0}}\ {\isasymLongrightarrow}\ h\ {\isasymge}\ {\isadigit{0}}\ {\isasymLongrightarrow}\ \isanewline
\ \ {\isacharbraceleft}s{\isachardot}\ s{\isachardollar}{\isadigit{0}}\ {\isacharequal}\ h\ {\isasymand}\ s{\isachardollar}{\isadigit{1}}\ {\isacharequal}\ {\isadigit{0}}{\isacharbraceright}\ {\isasymle}\ fb\isactrlsub {\isasymF}\ \isanewline
\ \ {\isacharparenleft}LOOP\ {\isacharparenleft}\isanewline
\ \ \ \ {\isacharparenleft}x{\isasymacute}{\isacharequal}{\isacharparenleft}f\ g{\isacharparenright}\ {\isacharampersand}\ {\isacharparenleft}{\isasymlambda}\ s{\isachardot}\ s{\isachardollar}{\isadigit{0}}\ {\isasymge}\ {\isadigit{0}}{\isacharparenright}\ DINV\ {\isacharparenleft}{\isasymlambda}s{\isachardot}\ {\isadigit{2}}\ {\isasymcdot}\ g\ {\isasymcdot}\ s{\isachardollar}{\isadigit{0}}\ {\isacharminus}\ {\isadigit{2}}\ {\isasymcdot}\ g\ {\isasymcdot}\ h\ {\isacharminus}\ s{\isachardollar}{\isadigit{1}}\ {\isasymcdot}\ s{\isachardollar}{\isadigit{1}}\ {\isacharequal}\ {\isadigit{0}}{\isacharparenright}{\isacharparenright}\ {\isacharsemicolon}\ \isanewline
\ \ \ \ {\isacharparenleft}IF\ {\isacharparenleft}{\isasymlambda}\ s{\isachardot}\ s{\isachardollar}{\isadigit{0}}\ {\isacharequal}\ {\isadigit{0}}{\isacharparenright}\ THEN\ {\isacharparenleft}{\isadigit{1}}\ {\isacharcolon}{\isacharcolon}{\isacharequal}\ {\isacharparenleft}{\isasymlambda}s{\isachardot}\ {\isacharminus}\ s{\isachardollar}{\isadigit{1}}{\isacharparenright}{\isacharparenright}\ ELSE\ skip{\isacharparenright}{\isacharparenright}\isanewline
\ \ INV\ {\isacharparenleft}{\isasymlambda}s{\isachardot}\ {\isadigit{0}}\ {\isasymle}\ s{\isachardollar}{\isadigit{0}}\ {\isasymand}{\isadigit{2}}\ {\isasymcdot}\ g\ {\isasymcdot}\ s{\isachardollar}{\isadigit{0}}\ {\isacharminus}\ {\isadigit{2}}\ {\isasymcdot}\ g\ {\isasymcdot}\ h\ {\isacharminus}\ s{\isachardollar}{\isadigit{1}}\ {\isasymcdot}\ s{\isachardollar}{\isadigit{1}}\ {\isacharequal}\ {\isadigit{0}}{\isacharparenright}{\isacharparenright}\isanewline
\ \ {\isacharbraceleft}s{\isachardot}\ {\isadigit{0}}\ {\isasymle}\ s{\isachardollar}{\isadigit{0}}\ {\isasymand}\ s{\isachardollar}{\isadigit{0}}\ {\isasymle}\ h{\isacharbraceright}{\isachardoublequoteclose}\isanewline
\ \ \isacommand{apply}\isamarkupfalse%
{\isacharparenleft}rule\ ffb{\isacharunderscore}loopI{\isacharcomma}\ simp{\isacharunderscore}all{\isacharparenright}\isanewline
\ \ \ \ \isacommand{apply}\isamarkupfalse%
{\isacharparenleft}force{\isacharcomma}\ force\ simp{\isacharcolon}\ bb{\isacharunderscore}real{\isacharunderscore }arith{\isacharparenright}\isanewline
\ \ \isacommand{by}\isamarkupfalse%
{\isacharparenleft}rule\ ffb{\isacharunderscore}g{\isacharunderscore}odei{\isacharparenright }\ {\isacharparenleft}auto\ intro{\isacharbang}{\isacharcolon}\ diff{\isacharunderscore}invariant{\isacharunderscore}rules\ poly{\isacharunderscore}derivatives{\isacharparenright}\isanewline
\end{isabellebody}

As before, the first line of the proof applies the non-evolution
$\wlp$-rules; the second one discharges $P\leq I$ and $I\leq Q$ for
loop invariant $I$. It remains to show that
$I\leq |{x'=f\, \&\, G\ \mathsf{DINV}\ I_d}]I$ for differential
invariant $I_d$.

For this we unfold the annotated invariant rule
\isa{ffb{\isacharunderscore}g{\isacharunderscore}odei}, which performs
step (2) of Example~\ref{ex:ball-inv} and generates the proof
obligation $I_d\leq |{x'=f\, \&\, G}]I_d$. The proof of this fact is
automatic because the rule
\isa{ffb{\isacharunderscore}diff{\isacharunderscore}inv}
(Lemma~\ref{P:inv-props2}) has been added to Isabelle's
simplifiers. Step (1) is checked with our rules for derivatives
\isa{poly{\isacharunderscore}derivatives} and differential invariants
\isa{diff{\isacharunderscore}invariant{\isacharunderscore}rules}
(Proposition~\ref{P:invrules}). The full verification covers less than
a page in the proof document.\qed
\end{example}

\begin{example}[Circular Pendulum via Invariant]\label{ex:pendulum-inv}
The ODEs 
\begin{equation*}
  x'\, t= y\, t\qquad \text{ and } \qquad y'\, t = -x\, t,
\end{equation*}
  which correspond to
  the vector field $f:\reals^V\to \reals^V$,
\begin{equation*}
f\, 
\begin{pmatrix}
  s_x\\
s_y
\end{pmatrix}
=
\begin{pmatrix}
  0 & 1\\
-1 & 0
\end{pmatrix}
\begin{pmatrix}
  s_x\\
s_y
\end{pmatrix},
 \end{equation*}
 for $V=\{x,y\}$, describe the kinematics of a circular pendulum. All
 orbits are ``governed'' by the separable differential equation
 \begin{equation*}
   \frac{\d y}{\d x} = \frac{y'}{x'} = -\frac{x}{y}, 
 \end{equation*}
 obtained by parametric derivation. Rewriting it as $x\d x+y\d y= 0$ and
 integrating both sides yields $x^2+y^2=r^2$, for some constant $r>0$,
 which describes the circular orbits of the ODEs. This leads to the
 differential invariant
\begin{equation*}
  I = \left(\lambda s.\ s_x^2 + s_y^2 = r^2\right),\qquad (r\ge 0).
\end{equation*}

Once again we apply our procedure from
Section~\ref{sec:differential-invariants} to prove
\begin{equation*}
  I = |x'= f\,  \&\,  \top] I
\end{equation*}
using Lemma \ref{P:inv-props3}, as the guard is trivial.
\begin{enumerate}
\item Using Proposition \ref{P:invrules} with $\mu\, s = s_x^2$ and
  $\nu\, s = r^2 - s_y^2$ we check that $I$ is an invariant, showing
  that $(\mu\circ X)' =(\nu\circ X)'$ for all $X\in \Sols\, f\, T\,
  s$, and hence
  \begin{equation*}
    \left((X\, t\, x)^2\right)' = \left(r^2 - (X\, t\, y)^2\right)'.
  \end{equation*}
We calculate
\begin{equation*}
  \left((X\, t\, x)^2\right)' = 2(X\, t\, x)(X'\, t\, x)
  = -2(X'\, t\, y)(X\, t\, y) =
   \left(r^2 - (X\, t\, y)^2\right)'.
\end{equation*}
It therefore follows from Proposition~\ref{P:invrules}(1) that $I$ is an invariant for $f$
along $\reals^V$; $I = |{x'=f\ \&\ \top}]  I$  holds by Lemma \ref{P:inv-props3}.
\item As $P=I=Q$, there is nothing to show.
\end{enumerate}

In the Isabelle formalisation, we introduce a name for the vector
field and show that $I$ is an invariant for it---as the invariant is
the pre- and postcondition, an annotation is not needed. The verification 
is straightforward following the workflow of the previous example, and
even simpler because the pre- and postconditions are just the
differential invariant.

\begin{isabellebody}

\isanewline
\isacommand{abbreviation}\isamarkupfalse%
\ fpend\ {\isacharcolon}{\isacharcolon}\ {\isachardoublequoteopen}real{\isacharcircum}{\isadigit{2}}\ {\isasymRightarrow}\ real{\isacharcircum}{\isadigit{2}}{\isachardoublequoteclose}\ {\isacharparenleft}{\isachardoublequoteopen}f{\isachardoublequoteclose}{\isacharparenright}\isanewline
\ \ \isakeyword{where}\ {\isachardoublequoteopen}f\ s\ {\isasymequiv}\ {\isacharparenleft}{\isasymchi}\ i{\isachardot}\ if\ i{\isacharequal}{\isadigit{0}}\ then\ s{\isachardollar}{\isadigit{1}}\ else\ {\isacharminus}s{\isachardollar}{\isadigit{0}}{\isacharparenright}{\isachardoublequoteclose}\isanewline
\isanewline
\isacommand{lemma}\isamarkupfalse%
\ pendulum{\isacharcolon}\ {\isachardoublequoteopen}{\isacharbraceleft}s{\isachardot}\ r\isactrlsup {\isadigit{2}}\ {\isacharequal}\ {\isacharparenleft}s{\isachardollar}{\isadigit{0}}{\isacharparenright}\isactrlsup {\isadigit{2}}\ {\isacharplus}\ {\isacharparenleft}s{\isachardollar}{\isadigit{1}}{\isacharparenright}\isactrlsup {\isadigit{2}}{\isacharbraceright}\ {\isasymle}\ fb\isactrlsub {\isasymF}\ {\isacharparenleft}x{\isasymacute}{\isacharequal}\ f\ {\isacharampersand}\ G{\isacharparenright}\ {\isacharbraceleft}s{\isachardot}\ r\isactrlsup {\isadigit{2}}\ {\isacharequal}\ {\isacharparenleft}s{\isachardollar}{\isadigit{0}}{\isacharparenright}\isactrlsup {\isadigit{2}}\ {\isacharplus}\ {\isacharparenleft}s{\isachardollar}{\isadigit{1}}{\isacharparenright}\isactrlsup {\isadigit{2}}{\isacharbraceright}{\isachardoublequoteclose}\isanewline
\ \ \isacommand{by}\isamarkupfalse%
\ {\isacharparenleft}auto\ intro{\isacharbang}{\isacharcolon}\ diff{\isacharunderscore}invariant{\isacharunderscore}rules\ poly{\isacharunderscore}derivatives{\isacharparenright}\isanewline
\end{isabellebody}

\item The Isabelle proof is automatic if we supply the
  tactic for derivative rules.  \qed
\end{example}

\begin{example}[Circular Pendulum via Flow]\label{ex:pendulum-flow}
  Alternatively, the kinematic equations for the circular pendulum
  from Example~\ref{ex:pendulum-flow} can of course be solved using
  linear combinations of trigonometric functions.  Yet first we need
  to show that the vector field $f$ is Lipschitz continuous with
  constant 1. Next we supply the flow
\begin{equation*}
\flow_s\, t = 
\begin{pmatrix}
  \cos t  & \sin t\\
-\sin t  &\cos t
\end{pmatrix}
\begin{pmatrix}
  s_x\\
s_y
\end{pmatrix}.
\end{equation*} 
We need to check that it solves the IVP $(f,s)$ for all
$s\in \reals^V$ and that it satisfies the flow conditions for
$T=\reals$ and $S=\reals^V$. As an example calculation,
\begin{align*}
  \flow_s'\, t = 
\begin{pmatrix}
  -\sin t  & \cos t\\
-\cos t  &-\sin t
\end{pmatrix}
\begin{pmatrix}
  s_x\\
s_y
\end{pmatrix}
=
\begin{pmatrix}
  0 & 1\\
-1 & 0
\end{pmatrix}
\begin{pmatrix}
  \cos t  &\sin t\\
-\sin t  &\cos t
\end{pmatrix}
\begin{pmatrix}
  s_x\\
s_y
\end{pmatrix}
= f\, (\flow_s\, t).
\end{align*}
The remaining conditions are left to the reader.

To compute $|x'=f\, \&\, \top]I$, we expand (\ref{eq:wlp-evl}). This
yields 
\begin{align*}
  |x'=f\, \&\, \top]I\, s &= \forall t.\ I\, (\flow_s\, t)\\
& = \left( \forall t.\ (\flow_s\, t\, x)^2 + (\flow_s\, t\, y)^2 =
  r^2\right)\\
& = \left( \forall t.\ (s_x\cos t + s_y\sin t)^2 + (s_y\cos t -
  s_x\sin t)^2 = r^2\right)\\
& = \left( \forall t.\ s_x^2(\sin^2 t +\cos^2 t) + s_y^2(\sin^2 t +
  \cos^2 t) = r^2\right)\\
&= I\, s.
\end{align*}

In the Isabelle proof along these lines, we first prove that the
vector field satisfies the conditions of the  Picard-Lindel\"of
theorem. To this end we need to unfold the locale definitions, then
introduce the Lipschitz constant, and call Isabelle's simplifiers.
Next, to prove that the solution supplied is a flow and a
solution to the IVP, we unfold definitions and finish the proof by
checking that the derivative of the flow in each coordinate coincides
with the vector field in that coordinate. The introduction of the flow
and these lemmas are shown below.

\begin{isabellebody}
\isanewline
\isacommand{abbreviation}\isamarkupfalse%
\ pend{\isacharunderscore}flow\ {\isacharcolon}{\isacharcolon}\ {\isachardoublequoteopen}real\ {\isasymRightarrow}\ real{\isacharcircum}{\isadigit{2}}\ {\isasymRightarrow}\ real{\isacharcircum}{\isadigit{2}}{\isachardoublequoteclose}\ {\isacharparenleft}{\isachardoublequoteopen}{\isasymphi}{\isachardoublequoteclose}{\isacharparenright}\isanewline
\ \ \isakeyword{where}\ {\isachardoublequoteopen}{\isasymphi}\ t\ s\ {\isasymequiv}\ {\isacharparenleft}{\isasymchi}\ i{\isachardot}\ if\ i\ {\isacharequal}\ {\isadigit{0}}\ then\ s{\isachardollar}{\isadigit{0}}\ {\isasymcdot}\ cos\ t\ {\isacharplus}\ s{\isachardollar}{\isadigit{1}}\ {\isasymcdot}\ sin\ t\ else\ s{\isachardollar}{\isadigit{1}}\ {\isasymcdot}\ cos\ t\ {\isacharminus}\ s{\isachardollar}{\isadigit{0}}\ {\isasymcdot}\ sin\ t{\isacharparenright}{\isachardoublequoteclose}\isanewline

\isacommand{lemma}\isamarkupfalse%
\ local{\isacharunderscore}flow{\isacharunderscore}pend{\isacharcolon}\ {\isachardoublequoteopen}local{\isacharunderscore}flow\ f\ UNIV\ UNIV\ {\isasymphi}{\isachardoublequoteclose}\isanewline
\ \ $\langle \isa{proof}\rangle$\isanewline
\end{isabellebody}

\noindent The proof of the correctness specification requires
only an application of the $\wlp$ rule
\isa{local{\isacharunderscore}flow{\isachardot
  }ffb{\isacharunderscore}g{\isacharunderscore}ode}
(\ref{eq:wlp-evl}) and Isabelle's simplifier.

\begin{isabellebody}
\isanewline
\isacommand{lemma}\isamarkupfalse%
\ pendulum{\isacharcolon}\ {\isachardoublequoteopen}{\isacharbraceleft}s{\isachardot}\ r\isactrlsup {\isadigit{2}}\ {\isacharequal}\ {\isacharparenleft}s{\isachardollar}{\isadigit{0}}{\isacharparenright}\isactrlsup {\isadigit{2}}\ {\isacharplus}\ {\isacharparenleft}s{\isachardollar}{\isadigit{1}}{\isacharparenright}\isactrlsup {\isadigit{2}}{\isacharbraceright}\ {\isasymle}\ fb\isactrlsub {\isasymF}\ {\isacharparenleft}x{\isasymacute}{\isacharequal}f\ {\isacharampersand}\ G{\isacharparenright}\ {\isacharbraceleft}s{\isachardot}\ r\isactrlsup {\isadigit{2}}\ {\isacharequal}\ {\isacharparenleft}s{\isachardollar}{\isadigit{0}}{\isacharparenright}\isactrlsup {\isadigit{2}}\ {\isacharplus}\ {\isacharparenleft}s{\isachardollar}{\isadigit{1}}{\isacharparenright}\isactrlsup {\isadigit{2}}{\isacharbraceright}{\isachardoublequoteclose}\isanewline
\ \ \isacommand{by}\isamarkupfalse%
\ {\isacharparenleft}force\ simp{\isacharcolon}\ local{\isacharunderscore}flow{\isachardot}ffb{\isacharunderscore}g{\isacharunderscore}ode{\isacharbrackleft}OF\ local{\isacharunderscore}flow{\isacharunderscore}pend{\isacharbrackright}{\isacharparenright}
\end{isabellebody} 
\hfill\qed
\end{example}

All four examples have been based on the categorical approach and the
state transformer semantics. Alternative formalisations for the other
predicate transformer algebras and the relational semantics can be
found in other verification components~\cite{afp:hybrid}.  In the
$\MKA$-based component, the proofs using the relational and the state
transformer semantics are precisely the same, which underpins the
modularity of our approach. In the other components we could certainly
achieve the same effect by simply rewriting names and adjusting some
types.

Transcendental functions cannot be expressed directly in $\dL$'s term
language, yet we can use them smoothly and easily with Isabelle with
the tactic outlined in Section~\ref{sec:isa-wlp}.  Both the
differential invariant workflow and the flow-based workflow benefit
from these rules. In fact, both approaches are very similar for the
pendulum example: both need a handful of lemmas to prove the partial
correctness specification $I = |{x'=f\ \&\ \true}] I$, and both
require a creative step in the form of introducing a differential
invariant or the flow for the system.

We have presented the pendulum example in matrix notation as this
points to a common feature of many applications: their dynamics can be
described by linear systems of ODEs that are representable by matrices
and have uniform solutions given by a matrix exponential that can be
computed with standard methods from linear algebra.  The development
of domain-specific techniques for linear systems with Isabelle has
been the subject of a successor article~\cite{Munive20}.  Beyond these
simple examples, our approach has successfully tackled a large set of
benchmarks from a systems competition~\cite{MitschMJZWZ20} and been
fine-tuned for proof automation, so that the size of proofs and level
of user interaction reported in this article is no longer
representative. More information about the background theory
development with Isabelle and the methods and heuristics programmed
can be found in the first author's doctoral
dissertation~\cite{Munive21}. A more far-reaching integration of
solvers and decision procedures, or procedures for invariant learning,
as oracles or with correctness guarantees, is of course crucial to the
applicability of this framework, but beyond the semantic
considerations of this article. It is left for future work.


\section{Outlook: A Flow-Based Verification Component
}\label{sec:flow-component}

The verification components presented so far adhered very much to the
pessimistic interactive theorem proving mindset that prefers the
internal reconstruction of all external results. This section briefly
outlines a fourth, more optimistic verification component that
deviates entirely from the vector-field-based approach of $\dL$ and
works directly with flows or solutions to IVPs. It shifts
responsibility for the correctness of solutions entirely to users---or
the computer algebra system they could or should use. This is common
practice for instance when working with hybrid
automata~\cite{DoyenFPP18}, and of course it simplifies proofs
considerably.

For this third workflow supported by our framework, the topological or
differentiable structure of the underlying state space is of secondary
interest. With Isabelle, this kind of structure and additional
conditions can always be imposed by instantiating types with sort
constraints as they arise.  Hence we start from a setting that covers
both discrete and continuous evolutions and use a general type for
time instead of $\isa{real}$, $\isa{rat}$ or $\isa{int}$.  The
evolution commands now specify arbitrary guarded $\flow$-type
functions instead of vector fields. The type of time needs to admit an
order relation, which is indicated by the sort constraint $\isa{ord}$
below, yet specific properties, such as reflexivity or transitivity,
need not be imposed ab initio.

Apart from that, the definition of the guarded-orbit semantics and the
$\wlp$ rule are as before, but side conditions on Lipschitz
continuity or the Picard-Lindel\"of theorem are superfluous.

\begin{isabellebody}
\isanewline
\isacommand{definition}\isamarkupfalse%
\ g{\isacharunderscore}evol\ {\isacharcolon}{\isacharcolon}\ {\isachardoublequoteopen}{\isacharparenleft}{\isacharparenleft}{\isacharprime}a{\isacharcolon}{\isacharcolon}ord{\isacharparenright}\ {\isasymRightarrow}\ {\isacharprime}b\ {\isasymRightarrow}\ {\isacharprime}b{\isacharparenright}\ {\isasymRightarrow}\ {\isacharprime}b\ pred\ {\isasymRightarrow}\ {\isacharprime}a\ set\ {\isasymRightarrow}\ {\isacharparenleft}{\isacharprime}b\ {\isasymRightarrow}\ {\isacharprime}b\ set{\isacharparenright}{\isachardoublequoteclose}\ {\isacharparenleft}{\isachardoublequoteopen}EVOL{\isachardoublequoteclose}{\isacharparenright}\isanewline
\ \ \isakeyword{where}\ {\isachardoublequoteopen}EVOL\ {\isasymphi}\ G\ T\ {\isacharequal}\ {\isacharparenleft}{\isasymlambda}s{\isachardot}\ g{\isacharunderscore}orbit\ {\isacharparenleft}{\isasymlambda}t{\isachardot}\ {\isasymphi}\ t\ s{\isacharparenright}\ G\ T{\isacharparenright}{\isachardoublequoteclose}\isanewline
\isanewline
\isacommand{lemma}\isamarkupfalse%
\ fbox{\isacharunderscore}g{\isacharunderscore}evol{\isacharbrackleft}simp{\isacharbrackright}{\isacharcolon}\ \isanewline
\ \ \isakeyword{fixes}\ {\isasymphi}\ {\isacharcolon}{\isacharcolon}\ {\isachardoublequoteopen}{\isacharparenleft}{\isacharprime}a{\isacharcolon}{\isacharcolon}preorder{\isacharparenright}\ {\isasymRightarrow}\ {\isacharprime}b\ {\isasymRightarrow}\ {\isacharprime}b{\isachardoublequoteclose}\isanewline
\ \ \isakeyword{shows}\ {\isachardoublequoteopen}fb\isactrlsub {\isasymF}\ {\isacharparenleft}EVOL\ {\isasymphi}\ G\ T{\isacharparenright}\ Q\ {\isacharequal}\ {\isacharbraceleft}s{\isachardot}\ {\isacharparenleft}{\isasymforall}t{\isasymin}T{\isachardot}\ {\isacharparenleft}{\isasymforall}{\isasymtau}{\isasymin}down\ T\ t{\isachardot}\ G\ {\isacharparenleft}{\isasymphi}\ {\isasymtau}\ s{\isacharparenright}{\isacharparenright}\ {\isasymlongrightarrow}\ {\isacharparenleft}{\isasymphi}\ t\ s{\isacharparenright}\ {\isasymin}\ Q{\isacharparenright}{\isacharbraceright}{\isachardoublequoteclose}\isanewline
\ \ \isacommand{unfolding}\isamarkupfalse%
\ g{\isacharunderscore}evol{\isacharunderscore}def\ g{\isacharunderscore}orbit{\isacharunderscore}eq\ ffb{\isacharunderscore}eq\ \isacommand{by}\isamarkupfalse%
\ auto\isanewline
\end{isabellebody}

Using the flows of the bouncing ball and the circular pendulum from
previous examples, verification proofs are now fully automatic. 
 
\begin{isabellebody}
\isanewline
\isacommand{lemma}\isamarkupfalse%
\ pendulum{\isacharunderscore}dyn{\isacharcolon}\isanewline 
\ \ {\isachardoublequoteopen}{\isacharbraceleft}s{\isachardot}\ r\isactrlsup {\isadigit{2}}\ {\isacharequal}\ {\isacharparenleft}s{\isachardollar}{\isadigit{0}}{\isacharparenright}\isactrlsup {\isadigit{2}}\ {\isacharplus}\ {\isacharparenleft}s{\isachardollar}{\isadigit{1}}{\isacharparenright}\isactrlsup {\isadigit{2}}{\isacharbraceright}\ {\isasymle}\ fb\isactrlsub {\isasymF}\ {\isacharparenleft}EVOL\ {\isasymphi}\ G\ T{\isacharparenright}\ {\isacharbraceleft}s{\isachardot}\ r\isactrlsup {\isadigit{2}}\ {\isacharequal}\ {\isacharparenleft}s{\isachardollar}{\isadigit{0}}{\isacharparenright}\isactrlsup {\isadigit{2}}\ {\isacharplus}\ {\isacharparenleft}s{\isachardollar}{\isadigit{1}}{\isacharparenright}\isactrlsup {\isadigit{2}}{\isacharbraceright}{\isachardoublequoteclose}\isanewline
\ \ \isacommand{by}\isamarkupfalse%
\ force\isanewline

\isacommand{lemma}\isamarkupfalse%
\ bouncing{\isacharunderscore}ball{\isacharunderscore}dyn{\isacharcolon}\ {\isachardoublequoteopen}g\ {\isacharless}\ {\isadigit{0}}\ {\isasymLongrightarrow}\ h\ {\isasymge}\ {\isadigit{0}}\ {\isasymLongrightarrow}\ \isanewline
\ \ {\isacharbraceleft}s{\isachardot}\ s{\isachardollar}{\isadigit{0}}\ {\isacharequal}\ h\ {\isasymand}\ s{\isachardollar}{\isadigit{1}}\ {\isacharequal}\ {\isadigit{0}}{\isacharbraceright}\ {\isasymle}\ fb\isactrlsub {\isasymF}\ \isanewline
\ \ {\isacharparenleft}LOOP\ {\isacharparenleft}\isanewline
\ \ \ \ {\isacharparenleft}EVOL\ {\isacharparenleft}{\isasymphi}\ g{\isacharparenright}\ {\isacharparenleft}{\isasymlambda}\ s{\isachardot}\ s{\isachardollar}{\isadigit{0}}\ {\isasymge}\ {\isadigit{0}}{\isacharparenright}\ T{\isacharparenright}\ {\isacharsemicolon}\ \isanewline
\ \ \ \ {\isacharparenleft}IF\ {\isacharparenleft}{\isasymlambda}\ s{\isachardot}\ s{\isachardollar}{\isadigit{0}}\ {\isacharequal}\ {\isadigit{0}}{\isacharparenright}\ THEN\ {\isacharparenleft}{\isadigit{1}}\ {\isacharcolon}{\isacharcolon}{\isacharequal}\ {\isacharparenleft}{\isasymlambda}s{\isachardot}\ {\isacharminus}\ s{\isachardollar}{\isadigit{1}}{\isacharparenright}{\isacharparenright}\ ELSE\ skip{\isacharparenright}{\isacharparenright}\isanewline
\ \ INV\ {\isacharparenleft}{\isasymlambda}s{\isachardot}\ {\isadigit{0}}\ {\isasymle}\ s{\isachardollar}{\isadigit{0}}\ {\isasymand}{\isadigit{2}}\ {\isasymcdot}\ g\ {\isasymcdot}\ s{\isachardollar}{\isadigit{0}}\ {\isacharminus}\ {\isadigit{2}}\ {\isasymcdot}\ g\ {\isasymcdot}\ h\ {\isacharminus}\ s{\isachardollar}{\isadigit{1}}\ {\isasymcdot}\ s{\isachardollar}{\isadigit{1}}\ {\isacharequal}\ {\isadigit{0}}{\isacharparenright}{\isacharparenright}\isanewline
\ \ {\isacharbraceleft}s{\isachardot}\ {\isadigit{0}}\ {\isasymle}\ s{\isachardollar}{\isadigit{0}}\ {\isasymand}\ s{\isachardollar}{\isadigit{0}}\ {\isasymle}\ h{\isacharbraceright}{\isachardoublequoteclose}\isanewline
\ \ \isacommand{by}\isamarkupfalse\ %
{\isacharparenleft}rule\ ffb{\isacharunderscore}loopI{\isacharparenright}\ {\isacharparenleft}auto\ simp{\isacharcolon}\ bb{\isacharunderscore}real{\isacharunderscore }arith{\isacharparenright}\isanewline
\end{isabellebody}
These examples no longer link flows with initial specifications in
terms of system of ODEs, from which a user might have started. Hence
there is no longer any formal guarantee from Isabelle that the
function $\flow$ specified satisfies any continuity of
differentiability assumptions such as those of
\isa{local{\isacharunderscore}flow}.

Further elaboration of this approach, in particular towards discrete
systems or in the direction of hybrid automata, is left for future
work.

\section{Related Work}\label{sec:related-work}

Methods for automated verification condition generation for partial
and total correctness assertions with proof assistants date back to
the early days of hardware and software verification by Gordon and
colleagues~\cite{Gordon87}.  Discussions on the benefits of shallow
embeddings of verification methods in proof assistants---among them
faster developments and increased modularity---can be traced back to
the same group of researchers. We generally follow an approach
described in~\cite{ArmstrongGS16} that starts from algebras of
programs to generate verification conditions for the structural
commands of programs, while developing those for basic commands in
concrete semantics of the program store dynamics. 

Mathematical components for classical real analysis have been
developed for the Coq proof assistant in the Coquelicot
library~\cite{BoldoLM15}; others for constructive analysis in the CoRN
library~\cite{Cruz-FilipeGW04}. The Picard-Lindel\"of theorem seems to
be available only in the latter~\cite{MakarovS13}. The proof assistant
HOL-light includes a library with formalised $n$-dimensional Euclidean
spaces. The first formalisation of the Picard-Lindel\"of theorem in
Isabelle, which we rewrite for our purposes and specialise to local
flows, can be found in the AFP entry for ordinary differential
equations~\cite{ImmlerH12a}.

Hybrid systems verfication in general-purpose proof assistants has
also been investigated.  Examples in PVS include semantic invariant
reasoning with hybrid automata~\cite{Abraham-MummSH01} and, after
submission of this article and publication of its
precursor~\cite{MuniveS18}, $\dL$-style verification with
semi-algebraic sets and real analytic functions~\cite{SlagelWD21}. An
earlier formalisation of the control function of an inverted
pendulum~\cite{Rouhling18} uses the Coquelicot library. Also in Coq,
the robot operating system (ROSCoq) framework uses a shallowly
embedded logic of events to reason about hybrid systems but only with
$\dL$'s differential induction rule. The HHL prover~\cite{WangZZ15}
formalises a Hoare-logic for hybrid systems verification within its
calculus of hybrid communicating sequential processes in
Isabelle~\cite{LiuLQZZZZ10}; part of their approach has been deeply
embedded. Their semantics is very different from ours. An integration
of their LZZ method~\cite{LiuZZ11} for finding semi-algebraic
invariants for polynomial dynamical systems could probably be
integrated into our framework to increase proof automation. Finally, a
term-checker for KeYmaera X~\cite{BohrerRVVP17} and, after submission
of this article, a formalisation of differential game logic
($\mathsf{d}\mathcal{GL}$)~\cite{afp:dgl} have been deeply embedded
recently in Isabelle/HOL. None of them aim at hybrid program
verification. For an in-depth description of $\dL$ see~\cite{Platzer18}. 
A thorough study of differential invariants has been pursued 
in~\cite{Platzer12}.

In theory, our own framework should therefore allow the integration of
much of the related work mentioned, so long as it is consistent with
our hybrid store semantics. It is not even necessary to delegate every
task to the proof assistant. One can use external tools implementing
decision procedures as oracles or at least certify their outputs with
Isabelle. The oracle-based approach, however, may jeopardise the
desirable conservative extension property relative to Isabelle's own
kernel. Translations between different proof assistants may not always
be straightforward. For instance, it is yet to be seen if dependent
types or multi-parameter type classes are needed for more flexible
implementations of functions spaces (bounded, linear and continuous)
or complex vector spaces, or if alternative formalisations of
Picard-Lindel\"of theorem and other existence theorems might help us
to alleviate some of the requirements in our workflows.


\section{Conclusion}\label{sec:conclusion}
We have presented a new semantic framework for the deductive
verification of hybrid systems with the Isabelle/HOL proof assistant.
The approach is inspired by differential dynamic logic, but the design
of our verification components, the focus of our framework and the
workflows for verifying hybrid systems are different.

First of all, as we use a shallow embedding, the basic verification
components generated are quite minimalist and conceptually
simple. They merely require the integration of a $\wlp$ semantics for
basic evolution commands of hybrid programs into standard predicate
transformer algebras. Our preferred semantics for such commands are
state transformers, which in most cases simply map states to the
guarded orbits of their temporal evolutions. Beyond that, no
domain-specific inference rules are needed, verification condition
generation is fully automatic, and even our approach to differential
invariants is based entirely on general purpose algebraic invariance
laws. Our examples show that mathematical reasoning about differential
equations follows standard textbook style and hence comes close to the
way mathematicians, physicists or control engineers have been trained
to reason about such systems. Whether this is preferable to the
proof-theoretic approach advocated by KeYmaera X remains to be seen.

Secondly, our approach aims at an open experimental platform that is
only limited by Isabelle's ODE and analysis components, the
expressivity of its higher-order logic and type system, and the proof
methods it provides. We could, for instance, have developed our
semantics for time-dependent vector fields, but the restriction to
autonomous systems, which does not affect generality, seems preferable
in practice. The integration of internal or external solvers for
differential algebras, transcendental functions or computer algebra
systems for computing Lipschitz constants or flows in the style of
Isabelle's Sledgehammer tactic are certainly interesting and very
important avenues for future work, but not a main concern in this
article. So far, our open approach simply offers semantic alternatives
that users may explore, adapt and extend.

Two specialisations of our framework are the topic of successor
papers. The first one restricts our approach to linear systems of
differential equations, where exponential solutions exist and can be
computed with standard methods from linear
algebra~\cite{Munive20}. The second one~\cite{FosterMS20} specialises
the predicate transformer semantics to algebraic variants of Hoare
logics and to refinement calculi for hybrid programs along the lines
of previous components for traditional while-programs~\cite{ArmstrongGS16}.
This shows that our denotational semantics for hybrid programs are
compatible with any Hoare logic, which constitutes another significant
conceptual simplification relative to $\dL$ and KeYmaera X.

Beyond that we expect that a recent formalisation of Poincar\'e maps
with Isabelle~\cite{ImmlerT19} might allow us to extend our framework
to discrete dynamical systems and more computational approaches to
hybrid systems.

Moreover, differential-algebraic systems of
equations~\cite{HairerW96}, which mix differential equations and
algebraic equations, and partial differential equations~\cite{John86}
are important for many applications in control engineering and
physics. Extending our approach is likely to require significant
background work on mathematical components with Isabelle. While, in
both settings, some simple cases can be reduced to systems of ODEs,
numerical methods are usually needed for working with such systems.
Whether the workflow of mathematicians, physicists and engineers with
such more computational approaches can be approximated easily with
Isabelle remains to be seen.

Finally, much work is needed to transform our framework into an
applicable verification tool for hybrid systems. First steps have
meanwhile been taken~\cite{FosterGMS21} with respect to more refined
hybrid stores and a more user-friendly specification language for
hybrid programs and their correctness properties, as already
mentioned.  More important, however, seems the integration of external
solvers and decision procedures, to which much work in the hybrid
systems community has already been
devoted~\cite{FultonMBP17,LiuZZ11,RebihaMM15,SassiGS14,SogokonMTCP19}.
Such procedures already increase the proof automation of KeYmaera X
and we foresee no reason why similar integrations should not lead to
similar benefits within our own framework.

\bigskip

\paragraph{acknowledgements}
  We are grateful to Ana Cavalcanti, Achim Brucker, Thao Dang, Simon
  Foster, Sergey Goncharov, Peter H\"ofner, Sergio Mover, Andr\'e
  Platzer, Andrei Popescu, Dmitriy Traytel as well as the participants
  of the RAMiCS 2018 conference and our Oslo lecture series on
  Isabelle/HOL for fruitful discussions. We would also like to thank
  our referees for pointers to related work and many suggestions that
  helped us to improve the presentation of this work and reflect on
  its contributions.  The first author acknowledges support by
  scholarship no. 440404 of the Mexican Consejo Nacional de Ciencia y
  Tecnolog\'ia (CONACyT).


\bigskip 


\bibliographystyle{abbrv}
\bibliography{ms}


\newpage

\appendix

\section{Cross-References to Isabelle Lemmas}\label{sec:crossref}

\small
\begin{tabular}[t]{ll}
  Result in article & Formalisation in Isabelle~\cite{afp:hybrid}  \\
\hline 
\hline

Proposition~\ref{P:rel-ka} 	& Implied by \isa{\isakeyword{interpretation}\ rel{\isacharunderscore}aka} in Section~\ref{sec:isa-wlp}\\

Proposition~\ref{P:kleisli-ka} 	& Implied by \isa{\isakeyword{instantiation}\ nd{\isacharunderscore}fun} in Section~\ref{sec:isa-wlp}\\

Proposition~\ref{P:rel-mka} 	& \isa{\isakeyword{interpretation}\ rel{\isacharunderscore}aka} in Section~\ref{sec:isa-wlp}\\

Proposition~\ref{P:kleisli-mka} 	& \isa{\isakeyword{instantiation}\ nd{\isacharunderscore}fun} in Section~\ref{sec:isa-wlp}\\

$\wlp$ for sequential composition~(\ref{eq:wlp-seq}) & \isa{ffb{\isacharunderscore}kcomp} in Section~\ref{sec:isa-wlp}\\

$\wlp$ rule for if-then-else~(\ref{eq:wlp-cond}) & \isa{ffb{\isacharunderscore}if{\isacharunderscore}then{\isacharunderscore}else} in Section~\ref{sec:isa-wlp}\\

$\wlp$ rule for finite iterations~(\ref{eq:wlp-loop}) & In the proof of \isa{ffb{\isacharunderscore}loopI} in Section~\ref{sec:isa-wlp}\\

Lemma~\ref{P:inv-lemma} & Part of \isa{\isacommand{named-theorems}} \isa{diff{\isacharunderscore}invariant{\isacharunderscore}rules} of Section~\ref{sec:isa-wlp}\\

Update functions from Section~\ref{sec:discrete-store} &  \isa{\isacommand{definition}} \isa{vec{\isacharunderscore}upd} in Section~\ref{sec:isa-wlp}  \\

Semantics for assignments from Section~\ref{sec:discrete-store} & \isa{\isacommand{definition}} \isa{assign} in Section~\ref{sec:isa-wlp}\\

$\wlp$ for assignments from Section~\ref{sec:discrete-store}~(\ref{eq:wlp-asgn}) & \isa{ffb{\isacharunderscore}assign} in Section~\ref{sec:isa-wlp}\\

Orbits in Section~\ref{sec:ODE} &  \isa{\isacommand{definition}} \isa{orbit} in Section~\ref{sec:isa-wlp}\\

Picard-Lindel\"of Theorem~\ref{P:picard-lindeloef} & \isa{picard{\isacharunderscore}lindeloef{\isacharunderscore }closed{\isacharunderscore}ivl{\isachardot}unique{\isacharunderscore}solution} in Section~\ref{sec:isa-ODE}\\

Monoid action identities for flows from Section~\ref{sec:ODE} & \isa{local{\isacharunderscore}flow{\isachardot }is{\isacharunderscore}monoid{\isacharunderscore}action} in Section~\ref{sec:isa-ODE}\\

$G$-guarded orbit ($\gorbit{\flow}$) in Section~\ref{sec:hybrid-store} & instance of \isa{\isacommand{definition}} \isa{g{\isacharunderscore}orbit}  of Section~\ref{sec:isa-ODE}\\

Lemma~\ref{P:g-orbit-props} & \isa{g{\isacharunderscore}orbital{\isacharunderscore}collapses} in Section~\ref{sec:isa-ODE}\\

Semantics for evolution commands in Section~\ref{sec:hybrid-store} & \isa{\isacommand{notation}} \isa{g{\isacharunderscore}orbital} of Section~\ref{sec:isa-wlp}\\

Proposition~\ref{P:wlpprop} & implied by \isa{ffb{\isacharunderscore}g{\isacharunderscore}ode} in Section~\ref{sec:isa-wlp}\\

Lemma~\ref{P:wlpprop-var}~(\ref{eq:wlp-evl}) & \isa{ffb{\isacharunderscore}g{\isacharunderscore}ode} in Section~\ref{sec:isa-wlp}\\

Example~\ref{ex:ball} & Example~\ref{ex:bouncing-ball-flow}\\

$\Sols\, f\, t_0\, s$ in Section~\ref{sec:generalisation} & \isa{\isacommand{definition}} \isa{ivp{\isacharunderscore}sols} in Section~\ref{sec:isa-ODE}\\

$G$-guarded orbit $\gorbit{X}$ of $X$ along $T$ in Section~\ref{sec:generalisation} & \isa{\isacommand{definition}} \isa{g{\isacharunderscore}orbit} of Section~\ref{sec:isa-ODE}\\

$G$-guarded orbital $\gorbit{f}$ of $f$ along $T$ in Section~\ref{sec:generalisation} & \isa{\isacommand{definition}} \isa{g{\isacharunderscore}orbital} of Section~\ref{sec:isa-ODE}\\

Lemma~\ref{P:gorbital} & \isa{ffb{\isacharunderscore}g{\isacharunderscore}orbital{\isacharunderscore}eq} in Section~\ref{sec:isa-ODE}\\

Semantics for evolution commands in Section~\ref{sec:generalisation} & \isa{\isacommand{notation}} \isa{g{\isacharunderscore}orbital} of Section~\ref{sec:isa-wlp}\\

Proposition~\ref{P:wlpprop-gen} & \isa{ffb{\isacharunderscore}g{\isacharunderscore}orbital} of Section~\ref{sec:isa-wlp}\\

Invariant of IVP $(f,s)$ & \isa{\isacommand{definition}} \isa{diff{\isacharunderscore}invariant} in Section~\ref{sec:isa-wlp}\\

Proposition~\ref{P:inv-prop} & \isa{ffb{\isacharunderscore}diff{\isacharunderscore}inv} in Section~\ref{sec:isa-wlp}\\

Proposition~\ref{P:invrules} & \isa{\isacommand{named-theorems}} \isa{diff{\isacharunderscore}invariant{\isacharunderscore}rules} of Section~\ref{sec:isa-wlp}\\

Example~\ref{ex:ball-inv} & Example~\ref{ex:bouncing-ball-inv}\\

$\dL$-axiom (\ref{eq:DS}) in Section~\ref{sec:dL} & \isa{DS} in Section~\ref{sec:isa-dL}\\

Proposition~\ref{P:wlpprop} ($\dL$-rule~(\ref{eq:dSolve})) & \isa{solve} in Section~\ref{sec:isa-dL}\\

Lemma~\ref{P:dcut} (\ref{eq:DC}) & \isa{DC} in Section~\ref{sec:isa-dL}\\

Lemma~\ref{P:dcut} (\ref{eq:dC}) & \isa{dC} in Section~\ref{sec:isa-dL}\\

Lemma~\ref{P:dcut} (\ref{eq:DW}) & \isa{DW} in Section~\ref{sec:isa-dL}\\

Lemma~\ref{P:dcut} (\ref{eq:dW}) & \isa{dW} in Section~\ref{sec:isa-dL}\\

Lemma~\ref{P:dcut} (\ref{eq:dI}) & \isa{dI} in Section~\ref{sec:isa-dL}\\

 \end{tabular}

\end{document}